\documentclass[11pt,letterpaper]{article}
\usepackage[margin=1in]{geometry}

\usepackage{natbib}
\usepackage[utf8]{inputenc} % allow utf-8 input
\usepackage[T1]{fontenc}    % use 8-bit T1 fonts
\usepackage{url}            % simple URL typesetting
\usepackage{booktabs}       % professional-quality tables
\usepackage{amsfonts}       % blackboard math symbols
\usepackage{nicefrac}       % compact symbols for 1/2, etc.
\usepackage{xcolor}         % colors
\usepackage{amsmath,amsfonts,amssymb,amsthm,array}
\usepackage{mathtools}
\usepackage{comment}
\usepackage{algorithm}
\usepackage[font=large]{caption}
\usepackage{algpseudocode}
\usepackage{bm}
\usepackage{color}
\usepackage{graphicx}
\usepackage{enumitem}
\usepackage{xargs}
\usepackage{multirow}
\usepackage{makecell}
\usepackage{fancyhdr}
\usepackage{wrapfig}
\usepackage{subcaption}
\usepackage{lscape}
\usepackage{bbm, dsfont}
\usepackage{thmtools}
\usepackage[colorinlistoftodos,bordercolor=orange,backgroundcolor=orange!20,linecolor=orange,textsize=scriptsize]{todonotes}

\usepackage{xspace}

\definecolor{bgcolor}{rgb}{0.8,1,1}
\definecolor{bgcolor2}{rgb}{0.8,1,0.8}
\usepackage{threeparttable}

\usepackage{tcolorbox}
\usepackage{pifont}
\definecolor{mydarkgreen}{RGB}{39,130,67}
\definecolor{mydarkred}{RGB}{192,47,25}

\newcounter{todocounter}

\newcommand{\bor}[1]{\textcolor{orange}{(Boris: #1)}}

\allowdisplaybreaks

\usepackage{lmodern} 
\usepackage{physics}
\usepackage{diagbox}
\usepackage{tabularx}
\usepackage{lscape}
\usepackage{tikz}
\usetikzlibrary{shapes.arrows}
\usetikzlibrary{arrows.meta} 

\usepackage{ninecolors}
\usepackage[backref=page]{hyperref}
\hypersetup{
    unicode=false,          % non-Latin characters in Acrobat’s bookmarks
    colorlinks=true,        % false: boxed links; true: colored links
    linkcolor=green6,         % color of internal links (change box color with linkbordercolor)
    citecolor=purple6,       % color of links to bibliography
    filecolor=magenta4,      % color of file links
    urlcolor=cyan4           % color of external links
}

\usepackage[capitalise,noabbrev, nameinlink]{cleveref}

\newcommand{\DeclareAutoPairedDelimiter}[3]{%
  \expandafter\DeclarePairedDelimiter\csname Auto\string#1\endcsname{#2}{#3}%
  \begingroup\edef\x{\endgroup
    \noexpand\DeclareRobustCommand{\noexpand#1}{%
      \expandafter\noexpand\csname Auto\string#1\endcsname*}}%
  \x}
\DeclareAutoPairedDelimiter\p{\lparen}{\rparen}
\DeclareAutoPairedDelimiter\sqb{\lbrack}{\rbrack} % square bracket
\DeclareAutoPairedDelimiter\crb{\lbrace}{\rbrace} % curly braces
\DeclareAutoPairedDelimiter\angle{\langle}{\rangle}
\DeclareAutoPairedDelimiter\ceil{\lceil}{\rceil}
\DeclareAutoPairedDelimiter\floor{\lfloor}{\rfloor}
\DeclareAutoPairedDelimiter\abs{\lvert}{\rvert}   % For absolute value
\DeclareAutoPairedDelimiter\norm{\lVert}{\rVert}  % For norm (double bars)
\newcommand{\R}{\ensuremath{\mathbb{R}}}

\newcommand{\Z}{\ensuremath{\mathbb{Z}}}

\newcommand{\E}{\mathbb{E}}
\newcommand{\poly}{\operatorname{poly}}

\def\<#1,#2>{\left\langle #1,#2 \right\rangle}

\newcommand{\eps}{\ensuremath{\epsilon}}

\newcommand{\ph}{\ensuremath{\varphi}}
\renewcommand{\phi}{\ensuremath{\varphi}}
\newcommand{\bx}{\mathbf{x}}
\newcommand{\expker}{\operatorname{exp}}
\newcommand{\Softmax}{\operatorname{softmax}}
\newcommand{\Attn}{\ensuremath{\mathsf{Attn}}}
\newcommand{\BaseCompress}{\textsc{BaseCompress}}
\newcommand{\Compress}{\textsc{Compress}}

\newcommand{\Var}[1]{\text{\bf Var}\normalfont\Bigl\lbrack #1 \Bigr\rbrack}

\newtheorem{theorem}{Theorem}[section]
\newtheorem*{theorem*}{Theorem}
\newtheorem{proposition}[theorem]{Proposition}
\newtheorem{lemma}[theorem]{Lemma}
\newtheorem{claim}[theorem]{Claim}
\newtheorem{corollary}[theorem]{Corollary}

\newtheorem{definition}{Definition}[section]

\newtheorem{remark}{Remark}[section]
\crefname{assumption}{assumption}{assumptions}
\theoremstyle{plain}

\title{Towards Tight Bounds for Streaming Attention}
\author{
\centering
\begin{tabular}{ccc}
\begin{tabular}{c}
Justin Y. Chen\\
MIT
\end{tabular}
&
\begin{tabular}{c}
Ying Feng\\
MIT
\end{tabular}
&
\begin{tabular}{c}
Piotr Indyk\\
MIT
\end{tabular}
\\[1.4em]
\begin{tabular}{c}
Michael Kapralov\\
EPFL
\end{tabular}
&
\begin{tabular}{c}
Ekaterina Kochetkova\\
EPFL
\end{tabular}
&
\begin{tabular}{c}
Boris Prokhorov\\
EPFL
\end{tabular}
\end{tabular}
}
\date{}

\begin{document}
\maketitle

\begin{abstract}
The attention mechanism is a cornerstone of modern transformer architectures. However, its expressive power comes at the cost of quadratic runtime and linear space usage. In particular, the classical transformer architecture explicitly stores all previously seen input elements (tokens) in order to generate the next one. The problem of implementing a transformer in limited space, known as {\em KV cache compression}, has received much interest over the past few years, spurring the development of powerful heuristics. Recent works of Haris et al, COLT'25 and Kochetkova et al, NeurIPS'25, formalized KV cache compression as the streaming attention approximation problem, providing both upper bounds (based on discrepancy theory) and information theoretic lower bounds. However, those papers left open a significant gap between the upper and lower bounds. For example, the space usage of their algorithms increases with the precision parameter, but the lower bound does not get stronger.

In this work, we revisit the streaming attention approximation problem and provide nearly tight bounds on its space complexity. On the algorithmic side, we achieve the result through a surprisingly tight interplay between three distinct methods for  kernel density estimation: discrepancy-based coreset constructions (e.g., Charikar-Kapralov-Waingarten'24), the polynomial method (e.g., Greengard-Rokhlin'87, Alman-Song'23), and space partitioning (e.g., Andoni-Laarhoven-Razenshteyn-Waingarten'17, Charikar-Kapralov-Nouri-Siminelakis'20). On the lower bound side, our main technical contribution is a new technique for using the \texttt{INDEX} problem with a large amount of side information that we hope will prove useful in other high dimensional geometric estimation problems.
\end{abstract}

\thispagestyle{empty}
\newpage
\thispagestyle{empty}
\tableofcontents
\newpage
\setcounter{page}{1}

\section{Introduction}
Transformer models~\cite{vaswani2017attention} are central to modern machine learning, achieving state-of-the-art performance for a vast range of applications in language processing, coding,  vision and multimodal domains. Most  modern GenAI systems are based on transformer architectures. A key building block in transformers is the {\em attention mechanism}, which computes correlations between data elements (called {\em tokens}) to select and propagate the signals in the data. The tokens are embedded into a geometric space and represented as vectors in $\R^d$. Formally, one is given a sequence $K$ of $n$ {\em key} embeddings $k_i \in \R^d$, an associated sequence $V$ of {\em value} embeddings $v_i \in \R^d$, and a {\em query} embedding $q \in \R^d$. The task is to compute

\begin{equation}
    \Attn(K, q, V) = \sum_{i \in [n]}\frac{\expker(k_i, q)}{\expker(K,q)} v_i
    \label{e:att}
\end{equation}
% \iffalse
% \begin{equation}
%     \mathsf{Attn}(q, K, V) = \frac1{Z(K, q)}\sum_{i \in [n]}e^{k_i^Tq}  v_i
%     \label{e:att}
% \end{equation}
% \fi
where $\expker(k, q) = e^{\langle k, q\rangle}$ and 
$\expker(K, q)=\sum_{i \in [n]} e^{\langle k_i, q\rangle}$.
\footnote{Prior works often normalize the inner product in the exponent. For example~\cite{alman2023fast} uses $\expker(k,q) = e^{\langle k, q\rangle/d}$, while~\cite{kochetkova2025streaming} uses $\expker(k,q) = e^{\langle k, q\rangle/\sqrt{d}}$. In this paper we remove the scaling factor to simplify the notation. }
Intuitively, the weight $\expker(k_i,q)$ captures how relevant the key-value pair $(k_i,v_i)$ is to the query $q$; the most relevant values $v_i$ are aggregated to form the output. 

This calculation is performed for many different query vectors $q$, in multiple layers of the transformer architecture.
%where the inputs are either given, or constructed from outputs of previous steps 
%\piotr{Add a footnote clarifying that in some applications we should technically consider adaptivity. }. 
As a result, attention computation is very resource-intensive. One bottleneck is that the direct algorithm for evaluating Eq.~\eqref{e:att} for $n$ query vectors $q$ scales quadratically in $n$. The value $n$ (often referred to as the {\em context length})  affects the accuracy of the transformers, and can be quite large. As a result, there has been a large body of work devoted to developing sub-quadratic approximate algorithms for attention computation~\cite{kitaev2020reformer,choromanski2020rethinking, chen2021scatterbrain, zandieh2023kdeformer,han2023hyperattention}, as well as studying the time complexity of this task~\cite{keles2023computational, alman2023fast, saha2026computational}. 

Another major computational bottleneck in attention computation is its {\em space} complexity. Specifically, the key-value pairs are typically stored in memory in order to avoid recomputing them at every step.  The storage required to cache these pairs scales linearly in the context length $n$, number of layers and ``attention heads'' of the model, and the number of distinct users. Managing this memory footprint (called {\em KV cache}) is one of the key   challenges in deploying large language models at scale. For example, ~\cite{pope2023efficiently} observes that ``At large batch sizes and context lengths, the KV cache can become very large, putting us at the risk of running out of memory''.  This challenge has stimulated a large body of research focused  on reducing the memory footprint, either by retaining a subset of tokens ~\cite{zhang2023h2o,li2024snapkv,kochetkova2025streaming} or quantizing the embeddings, i.e., representing them using fewer bits~(\cite{liu2024kivi,zhang2024kv}; see also a very recent announcement by Google~\cite{zandieh2026turboquant}).\footnote{Furthermore, NVIDIA maintains a repository containing implementations of dozens of KV cache compression methods: \url{https://github.com/NVIDIA/kvpress}.}
% \todo{justin: added optional footnote, feel free to remove}

Although most of the aforementioned methods for KV caching have been empirical, recent work~\cite{kochetkova2025streaming} proposed provable algorithms for this task, using techniques from discrepancy theory. Specifically, they design a low-storage data structure that maintains the sequence of keys $K$ and associated values $V$, that given any query vector $q$, returns a vector $\hat{z} \in \R^d$ such that\footnote{One could imagine a stronger guarantee where each coordinate is approximated with multiplicative accuracy, i.e.,  $|\hat{z}_j-\Attn(K, q, V)_j| \le \eps \Attn(K, q, V)_j$ for all $j$. Unfortunately, \cite{haris2025compression} showed that to achieve such a guarantee, any algorithm must use linear storage.}
\begin{equation}
    \|\hat{z}-\Attn(K, q, V)\|_2 \le \eps \cdot \|\Softmax(K, q)\|_2 \cdot \|V\|_F
\end{equation}
% \bor{Other works \cite{alman2023fast} actually use simpler additive error $\|\hat{z}-\Attn(K, q, V)\|_2 \le \eps$ with additional assumption $||v_i||_2 < r$ which is reasonable from the practical point of view. Our algorithm and lower bounds are tight up to polylog factors with this simpler guarantee - the only change is to remove Cauchy-Schwarz on page 45.}
where $\Softmax(K, q)_i=\frac{\expker(k_i,q)}{\expker(K,q)}$ for $i=1\ldots n$, and 
$\eps>0$ is the accuracy parameter. 
(Note that the norm of the softmax $\|\Softmax(K, q)\|_2$ is the ratio between the $\ell_2$ norm and the $\ell_1$ norm of the vector $[\expker(k_1,q), \ldots , \expker(k_n,q)]$, and therefore ranges between $\frac{1}{\sqrt{n}}$ and $1$.) The data structure supports adding new key-value pairs $(k_i,v_i)$, and uses 
\[ \tilde{O}\left(d\sqrt{d}\cdot\frac{e^{2r^2}}{\eps}\right) \] bits of space, where $\tilde{O}$ suppresses factors logarithmic in $n$, while $r$ is an upper bound on the $\ell_2$-norm of vectors $q$ and $k_i$.
They also present a lower bound showing that any such data structure must use at least 
\[ \tilde{\Omega}\left(\min\left\{\frac{1}{\eps^2}, d\cdot e^{2r^2}\right\}\right)
\] bits of space. 

These results demonstrate that a sublinear-space data structure for approximating attention is indeed feasible. However, they still leave a large gap between the upper bound (which scales linearly in $1/\eps$) and the lower bound (which does not depend on $1/\eps$ at all, except when the value of $\eps$ is sufficiently high).  

\subsection{Our results}

%\piotr{Is ``their'' $\eps$ the same as ``our'' $\eps$ ? If not, then how to translate and compare the bounds.}

%\piotr{Main points: (1) our algorithmic and lower bounds are tight up to polylog or $n^{o(1)}$ factors; (2) in a natural regime of dimension $d$, we achieve a sub-linear dependence on $1/\eps$. (3) replace $C(q,P)$ with $1/\sqrt{n}$.}

%\|\hat{z}-\mathsf{Attn}(q, P, V)\|_2 \le \eps \cdot \frac{1}{\sqrt{n}} \cdot \|V\|_F

%HERE

In this paper we address the main question left open in~\cite{kochetkova2025streaming}, presenting (nearly) tight upper and lower bounds for the space complexity of approximating attention for a wide range of parameters. 
In parallel, we also study the space complexity of approximating the normalizing constant $\expker(K,q)$. This problem is of separate interest, as it corresponds to the problem of {\em kernel density estimation (KDE)} with respect to the kernel $\expker(k,q)$. Since the latter task returns a scalar value as opposed to a vector, it is somewhat simpler to deal with than attention. At the same time, the techniques used to estimate $\expker(K,q)$ transfer almost immediately to the attention case.

%\piotr{Discuss guarantees, softmax}

We consider the following setting:
    \begin{itemize}
        \item As in~\cite{kochetkova2025streaming}, we assume that for all keys $k$ and queries $q$ we have $\|k\|_2 \leq r, \|q\|_2 \leq r$.\footnote{Following \cite{kochetkova2025streaming} we do not require $\|v\|_2 \leq r$ for all $v\in V$. This assumption was considered in previous works \cite{alman2023fast, gupta2025subquadratic} giving another guarantee $\|\hat{z} - \Attn\| \leq \eps$. We note that our algorithm also produces $\tilde{O}(\eps)$ approximation under this assumption. This is achieved by replacing Cauchy–Schwarz in \eqref{eq:attn-softmax-norm} with a bound $|V_{i,j}| \leq r$. Our lower bound instances also satisfy $\|v\|_2 \leq 1$ for all $v\in V$.}
        
        For simplicity of notation, we define the softmax distribution as proportional to $e^{\langle k, q \rangle}$. This way the radius $r$ becomes the parameter governing the {\em temperature} of the model~\cite{vaswani2017attention,HintonVD15,JangGP17,MaddisonMT17}. Large values of the radius $r$ correspond to the {\em low temperature} regime, in which the attention function can concentrate mass very sharply on a few tokens, and low values of $r$ correspond to the {\em high temperature} regime, in which the softmax distribution cannot peak significantly and therefore cannot concentrate mass on a few tokens. 
        
        \item We focus on the setting where the precision parameter $\eps$ is at least inversely polynomial in $n$, i.e. $\epsilon\in \left(\frac{1}{n^{100}}, \frac{1}{2}\right)$.
         
        \item We focus on the ``high-dimensional'' regime where the dimension $d$ is poly-logarithmic in $n$. That is, we assume $d\le \log^a n$ for some constant $a \geq 1$.
    \end{itemize}

Our main results are two-fold:

\paragraph{Scaling with $\eps$:} Assume that $r = (\log n)^{o(1)}$, $\log n \leq d \leq \log^a n$ for some constant $a > 1$. 
Furthermore, assume\footnote{A similar assumption has been used, e.g., in~\cite{alman2023fast}} that $n^{-100} < \eps < n^{-c}$ for an arbitrarily small constant $c > 0$.
Then streaming attention can be solved using
    \begin{align*}
        \tilde{O}\qty(\qty(\frac{1}{\eps})^{1 - \frac{1}{a} + o(1)})
    \end{align*}
    memory, where $\tilde{O}(\cdot)$ suppresses $O(\poly\log(n))$ factors. Furthermore, this bound is tight for any $d = \log^a n$, $a > 2$, i.e. there is a
    \begin{align*}
        \tilde{\Omega}\qty(\min\qty{n, \qty(\frac{1}{\eps})^{1 - \frac{1}{a} - o(1)}})
    \end{align*}
    lower bound for its space complexity.

Therefore, we provide a (nearly) tight bound for the streaming complexity of attention as a function of $\eps$ for a wide range of parameters. Furthermore, the bound is {\em sublinear} in $1/\eps$, improving over the linear dependence in~\cite{kochetkova2025streaming}.

\paragraph{Scaling with $r$:}  Streaming attention can be solved using   $\tilde{O}\qty(\frac{e^{r^2(1 + o(1))}}{\eps})$ memory. Furthermore,  the problem\footnote{We note that the parameter setting that our lower bound is shown in implies that $\|\Softmax(K, q)\|_2$ is bounded by $O(n^{o(1)}/\sqrt{n})$. %this is due to the assumption $r^2\leq (\log n)/2$.
Showing a lower bound for other ranges of this parameter remains an open problem.} requires
    \[
        \tilde{\Omega}\qty(\min\qty{n, \frac{e^{r^2(1 - o(1))}}{\eps}})
    \]
memory bits, when $d =  (\log n)^{1 + \Omega(1)}$ and $\Omega(\log n) \leq r^2 \leq \frac{\log n}{2}$, i.e. $n^{\Omega(1)} \leq e^{r^2} \leq \sqrt{n}$.

%\piotr{I think the above interprets $\tilde{O}(\cdot)$ as up to $n^{o(1)}$, not up to logarithmic factors. We need to pick one definition and be consistent.} \bor{No, we can do without $n^{o(1)}$. Optimal choice of $\Delta$ gives smth like $e^{r^2(1 - r^{-1/5})}$ and everything else is poly factors. It is going to appear in the proof of this theorem.}
% Piotr: great, thanks!

This yields a (nearly) tight bound for the streaming complexity of attention as a function of $r$ for a wide range of parameters, improving over the bound in~\cite{kochetkova2025streaming}.

\iffalse
Our space ($r = \Omega(\log n)$) for attn with softmax:
\begin{align*}
    \text{u.b.:}\quad \tilde{O}(\frac{e^{r^2}}{\eps}) \quad\quad \text{l.b. (only for $r^2 < \frac{\log n}{2}$}): \quad \tilde{\Omega}(\frac{e^{r^2}}{\eps})
\end{align*}

Our space ($r = O(1), \eps < n^{-c}$) for attn with softmax:
\begin{align*}
    \text{u.b.:}\quad \tilde{O}((\frac{1}{\eps})^{1 - \Theta(1)}) \quad\quad \text{l.b. (only for $d = \omega(\log^2 n)$)}: \quad \tilde{\Omega}((\frac{1}{\eps})^{1 - \Theta(1)})
\end{align*}
\fi

\iffalse
\subsection{Comparison with \cite{kochetkova2025streaming} [RAW]}
\todo{@Michael, @Ekaterina: check pls}
Our guarantee:
\begin{gather*}
        \label{eq:attn}
        \norm{\hat{z} - \frac{\sum\limits_{p\in K} e^{\<p_i, q>}\cdot v_i}{\sum\limits_{p\in K} e^{\<p_i, q>}}}_2 \leq \eps \cdot \frac{1} {\sqrt{n}} \cdot \|V\|_F
\end{gather*}
Guarantee from \cite{kochetkova2025streaming}:
\begin{gather*}
        \label{eq:attn}
        \norm{\hat{z} - \frac{\sum\limits_{p\in K} e^{\<p_i, q>}\cdot v_i}{\sum\limits_{p\in K} e^{\<p_i, q>}}}_2 \leq \eps \cdot {\frac{\sqrt{\sum\limits_{p\in K} e^{2\<p_i, q>}}}{\sum\limits_{p\in K} e^{\<p_i, q>}}} \cdot \|V\|_F
\end{gather*}

For the same choice of $\eps$ our guarantee is always at least as good (from Cauchy–Schwarz inequality) and it is strictly better when softmax is not "flat".

Space of $\cite{kochetkova2025streaming}$:
\begin{align*}
    \text{u.b.:}\quad \tilde{O}(\frac{e^{2r^2}}{\eps}) \quad \quad 
    \text{l.b.:}\quad \tilde{\Omega}(\min\{\frac{1}{\eps^2}, e^{2r^2}\})
\end{align*}
\fi

\section{Technical Overview}
In this section we define our model of computation and problems studied formally, present our results and outline the main techniques used to achieve them.

Throughout, we use $\expker(\cdot, \cdot)$ to refer to the core kernel (henceforth referred to as the {\em softmax kernel}) associated with attention computation. For key and query vectors $k, q$, $\expker(k, q) = e^{\langle k, q \rangle}$. For the full key sequence $K = (k_1, \ldots, k_n)$, $\expker(K, q) = \sum_{i=1}^n e^{\langle k_i, q \rangle}$ is the sum of the kernel over all keys.
We will refer to $\exp(K, q)$ as the \emph{normalizing constant} as it serves the purpose of normalizing the weights $\expker(k_i, q)$ given to each value vector $v_i$ in the attention equation (\cref{e:att}). In the following technical overview we will mostly focus on the problem of computing $\expker(K, q)$, as the techniques needed for this task transfer to the more general attention problem with only minor modifications. 
%\piotr{Is it clear enough?}

In both our upper and lower bounds, we will make use of polynomial expansions of the exponential function. Recall the standard Taylor expansion: $e^x = \sum_{l=0}^\infty \frac{x^l}{l!}$. We denote the truncated kernels $\expker_{\leq t} (k,q ) := \sum_{l=0}^t \frac{\langle k,q \rangle^l}{l!}$  and $\expker_{>t} (k,q ) := \sum_{l=t+1}^{\infty} \frac{\langle k,q \rangle^l}{l!}$, such that: 
\[\expker(K, q) =\sum_{k \in K} \sum_{l=0}^\infty  \frac{\langle k,q \rangle^l}{l!} = \sum_{k \in K}  \expker_{\leq t} (k,q )+ \expker_{> t} (k,q ).\]

We start by formally defining the problem of computing the normalizing constant $\expker(K, q)$ in the streaming model.
We note that this is exactly an instance of the {\em kernel density estimation problem} (KDE for short)  that has received much attention in its own right in the recent literature~\cite{charikar2017hashing,backurs2018efficient,backurs2019space,siminelakis2019rehashing,CKNS20,karppa2022deann,CKW24}.  We note, however, that most of the aforementioned works consider shift-invariant kernels, whereas softmax KDE is a version of maximum inner product search.

\begin{definition}[Streaming softmax KDE, i.e. computing $\expker(K, q)$]\label{def:streaming-kde}
    The algorithm receives keys $k_1, \ldots, k_n\in \R^d$ in a stream and at every point $j = 1,\ldots, n$, 
    %\todo{\ying{We used $t$ for both the truncation threshold and the streaming timestep -- would this be confusing?} \bor{changed to $K_j$ as in \cite{kochetkova2025streaming}}}
    % PIOTR: looks good!
    given any fixed $q\in \R^d$,  outputs an estimate $\widehat \expker(K_{j}, q)$ of $\expker(K_{j}, q)$ that satisfies
    \begin{gather}
        \label{eq:streaming-kde}
        |\widehat{\expker}(K_{j}, q) - \expker(K_{j}, q)| \leq \eps \cdot \expker(K_{j}, q)
    \end{gather}
    with probability at least $1-\delta$ for some $\delta\in (0, 1)$ for any fixed query $q$.  Here $K_{j}=(k_1, \ldots, k_{j})$ is the sequence of keys received by time $j$.
\end{definition}

\begin{definition}[Streaming attention]\label{def:streaming-attn}
    The algorithm receives a sequence of keys, values and queries $(k_j, v_j, q_j), j=1,\ldots, n$ in a stream and at every point $j$ the algorithm must output $\hat z_j$ such that 
    $$
    \|\hat z_j - \Attn(K_{j}, q_j, V_{j})\|_2\leq \epsilon \cdot \|\Softmax(K_{j}, q_j)\|_2 \|V_{j}\|_F,
    $$ 
    with  probability at least $1-\delta$ for some $\delta\in (0, 1)$, where $K_{j}=(k_1,\ldots, k_{j})$ and $V_{j}=(v_1,\ldots, v_{j})$. 
\end{definition}

\subsection{High temperature regime (small $r$).}
Our main result (see \cref{th:const_r_ub_lb}) holds for $r=\log^{o(1)} n$, and shows 
%We focus on $r=O(1)$ here only for the ease of exposition. \bor{do we need to assume this? the analysis doesn't change at all}}} 
that $1/\epsilon$ scaling of the space complexity by~\cite{kochetkova2025streaming} is in fact suboptimal by factors polynomial in $1/\epsilon$ in {\em any fixed polylogarithmic dimension}. Specifically, if $d=\log^a n$, we show that the space complexity of both softmax KDE and attention computation is approximately
%\bor{and running time} PIOTR: the discussion here is mostly focused on space. The theorems should claim time bounds.

\begin{equation}\label{eq:const-r-ub}
\left(\frac1{\epsilon}\right)^{1-1/a+o(1)}.    
\end{equation}

To put this in perspective, it is useful to recall the collection of techniques that have been used for {\em time}-efficient approximation algorithms for the attention function, and, more generally, kernel density estimation. First, a classical approach pioneered in the Fast Multipole Method of Greengard and Rokhlin~\cite{greengard1987fast} for the low dimensional regime (essentially, for constant $d$) amounts to the Taylor expansion of the (softmax) kernel and summarizing the dataset by its moments of order $t=O(\log (1/\epsilon))$. The number of such moments is ${d + t \choose t}$, which is exponential in $t$. A carefully optimized version of these ideas has recently been used to obtain nearly linear time algorithms for {\em computing} the attention function in dimension up to $C\log n$~\cite{alman2023fast}. Interestingly, these algorithms remain subquadratic for $d=C\log n$, with the constant $C$ affecting the runtime exponent. Furthermore, SETH-based lower bounds presented in the same paper show that such dependence is necessary.  
In contrast to these results, we show %~\eqref{eq:const-r-ub}
that the {\em space} 
%\bor{and runtime} %PIOTR: see the earlier comment
complexity of attention in the high temperature regime and logarithmic dimension is {\em sub-polynomial} in $1/\epsilon$ (setting $a \approx 1$ in~\eqref{eq:const-r-ub}). In slightly higher dimension, say, $d=\log^2 n$, our upper bound yields space complexity of 
$$
\approx 1/\sqrt{\epsilon}. 
$$

Surprisingly, our algorithmic  bound is tight up to factors subpolynomial in $1/\epsilon$, for every integer $a> 2$. Formally, we show

\begin{theorem}
    \label{th:const_r_ub_lb}
    Assume that $r = \log^{o(1)} n$, $\log n \leq d \leq \log^a n$ for some constant $a > 1$ and $n^{-100} < \eps < n^{-c}$ for an arbitrarily small constant $c > 0$. Then the Softmax KDE problem \eqref{eq:streaming-kde} can be solved using
    \begin{align*}
        \tilde{O}\qty(\qty(\frac{1}{\eps})^{1 - \frac{1}{a} + o(1)})
    \end{align*}
    memory, where $\tilde{O}(\cdot)$ suppresses $O(\poly\log(n))$ factors. Furthermore, this bound is tight when $d = \Theta(\log^a n)$ for any constant $a > 2$, i.e. there is a
    \begin{align*}
        \tilde{\Omega}\qty(\min\qty{n, \qty(\frac{1}{\eps})^{1 - \frac{1}{a} - o(1)}})
    \end{align*}
    lower bound for the streaming space complexity.
\end{theorem}
\begin{remark}
    \label{remark:const_r_attn}
    This result also holds for the Streaming Attention problem (\Cref{def:streaming-attn}). See \Cref{th:const_r_ub_lb_attn} for the formal statement and \Cref{sec:attn_ub,sec:attn_lb} for details of the reduction.
\end{remark}
Here and throughout, for the upper bounds, we say a randomized streaming algorithm solves a problem if at every point $j = 1, \dots, n$, given a fixed query, the algorithm outputs a correct approximation with $1- \frac{1}{\poly n }$ success probability. For the lower bounds of the space complexity, the entire stream is given at once and the protocol only needs to succeed with $.99$ probability.

\begin{comment}
\begin{remark}\label{rem:const_r_ub_lb}
    Additionally, if we consider larger (sub-polynomial) values of $\eps$, then simple coreset construction gives $\tilde{O}\qty(\frac{1}{\eps})$ upper bound and this is almost optimal, i.e. there is a\todo{proof not written yet}
    \begin{align*}
        \tilde{\Theta}\qty(\min\qty{n, \qty(\frac{1}{\eps})^{1 - o(1)}})
    \end{align*}
    upper/lower bound. $o(1)$ in this case is $\frac{\log \log 1 / \eps - \log \log d}{\log d}$. 
\end{remark}    
\end{comment}

% \begin{remark}
% Our streaming algorithm assumes knowledge of $n$ for simplicity of presentation, but is very robust to approximations: we note that an $O(1)$ approximation to $\log n$, i.e., knowledge of $n$ up to a polynomial, is enough.
% \end{remark}

\noindent\paragraph{Our techniques: upper bound.}

Given a stream of at most $n$ data points (denoted by $K$), the data structure maintains:
\begin{enumerate}
    \item a  sketch $\sum_{k \in K} \psi(k)$ for some feature map $\psi$, and
    \item a coreset $C \subset K$, %\todo{explain more how we maintain $C$ (not just existence)}
    %PIOTR: since we are only briefly stating that such coresets exist and why, there is no need to explain here how to maintain them.
\end{enumerate}
such that for any query $q \in \R^d$ and some fixed  integer $t$  of our choice,  one can use the  sketch $\sum_{k \in K} \psi(k)$ to approximate $\sum_{k \in K}  \expker_{\leq t} (k, q)$, and the coreset $C$ to approximate $\sum_{k \in K}  \expker_{>t} (k, q)$.

 We will give the constructions of $\psi$ and $C$, and then show that there exists $t$ such that the total space usage of
$\sum_{k \in K} \psi(k)$ and $C$ is $\approx (1/\eps)^{1- \frac{\log\log(1/\eps)}{\log d}}$. Formally, we show the existence of a pair of embeddings $\psi_1, \psi_2$, for the keys and the queries respectively, that approximate the low degree part of the exponential function well:
    \[\qty|\Bigl \langle \sum_{k \in K} \psi_1(k), \psi_2(q) \Bigr \rangle - \sum_{k \in K}  \expker_{\leq t}(k, q)| \leq \eps \cdot \expker(K, q),\] and that $ \sum_{k \in K} \psi_1(k)$ can be stored in $\tilde{O}\Bigl(\frac{(d+t)^{t}}{t!}\Bigr)$ bits of space. The embeddings are simple: we store low moments of the dataset of keys, similarly to the Fast Multipole Method. The formal guarantees are given in \cref{lem:map} in Section~\ref{sec:ub-high-temp}. 

Given the aforementioned sketch of the dataset $K$, namely $\sum_{k \in K} \psi_1(k)$, it remains to design a small space compression of $K$ that allows one to approximate $\sum_{k \in K} \expker_{>t} (k, q)$. We accomplish this with a {\em coreset} (a small subset) of $K$ for the truncated kernel $\expker_{>t}$. Specifically, we show that there exists $C \subset K$ such that with high probability, for any fixed $q$:
\[
    \Bigl | \sum_{k \in K}\expker_{>t}(k, q) -  \frac{|K|}{|C|} \cdot \sum_{k \in C} \expker_{>t}(k, q) \Bigr | \leq \eps \sum_{k \in K} \expker(k, q). 
\]
The existence of such a coreset follows by noting that for any integer $t \geq 0$, the kernel 
$$
\expker_{>t}(k, q) := \sum_{l=t+1}^\infty  \frac{\langle k,q \rangle^l}{l!}
$$ 
is positive definite. This means that if we construct an embedding of keys and queries into Hilbert space such that {\bf (a)} the dot products in the Hilbert space are equal to  the kernel similarities $\expker_{> t}(k, q)$ and {\bf (b)} the Hilbert space norms of the embedded keys and queries are appropriately small, then by Banaszczyk's theorem (a deep result in discrepancy theory\footnote{The idea of constructing coresets for kernel density estimation using Banaszczyk's theorem was pioneered in~\cite{PT20}, and the recent work of~\cite{CKW24} gave a powerful asymmetric construction of embeddings (i.e., different embeddings for keys and queries), leading to surprisingly efficient KDE data structures for smooth kernels. }), 
there is a coreset that approximates the KDE value of the keys at every query well.  The  Hilbert space norms of the embeddings govern the size of the coreset (they bound the $\gamma_2$ norm of the truncated kernel), and, intuitively, they are small exactly because $\expker_{>t}$ is a truncated kernel. Specifically, a careful analysis yields Lemma \ref{lem:coreset}, which states that our truncated kernel $\expker_{>t}$ admits an $\epsilon$-approximate coreset $C$ that can be stored in space
 $$
 \tilde{O}\left(\frac{e^{r^2}}{\eps}\cdot \sum_{l=t+1}^\infty  \frac{r^{2l}}{l!} \right).
 $$

Finally, it remains to choose $t$ to minimize the total space usage of the embedding based approach and the coreset size, i.e. 

\begin{equation}\label{eq:ub}
    \min_{t} \left \{ \frac{(d+t)^{t}}{t!} + \frac{e^{r^2}}{\eps}\cdot \sum_{l=t+1}^\infty  \frac{r^{2l}}{l!} \right \}.
\end{equation}

The ultimate value of the cutoff point $t$ that we use is 
 $t\sim \frac{\log (1 / \eps)}{\log d}$, resulting in a total space complexity of $\tilde{O}((1/\eps)^{1- \frac{\log\log(1/\eps)}{\log d} + o(1)})$, as claimed.
 %\todo{it is the first time we see $\frac{\log\log(1/\eps)}{\log d}$, it is not just "as claimed"}.
 %PIOTR: the expression appears earlier.
 The quick approximate derivation for this amounts to truncating the series to the dominant first term, ignoring the $e^{r^2}$ and $r^{2(t+1)}$ as $(1/\eps)^{o(1)}$ factors and noting that $(d+t)^t\approx d^t$, so that the expression above becomes
 $$
 \approx \frac{d^t}{t!}+\frac1{\epsilon}\cdot \frac1{t!}.
 $$
 For $t= \frac{\log(1/\epsilon)}{\log d}$ one has $d^t=1/\epsilon$ and the remaining $\frac{1}{t!}$ amounts to at most $(1/\eps)^{-1/a + o(1)}$ since
 \[
    \frac{\log (t!)}{\log (1/\eps)} \sim \frac{t\log t}{\log (1/\eps)} = \frac{\log t}{\log d} \sim \frac{\log \log (1/\eps)}{a\log \log n} \sim \frac{1}{a}
 \]
 % $$
 %  t!=t^{\Omega(t)}=\left(\frac{\log(1/\epsilon)}{\log d}\right)^{\Omega\left(\frac{\log(1/\epsilon)}{\log d}\right)}
 % =e^{\Omega\left(\log (1/\epsilon)\cdot \frac{\log \log(1/\epsilon)}{a \log \log n}\right)-o(1)}=(1/\epsilon)^{\Omega(1/a)-o(1)}    
 % $$
 where "$\sim$" denotes equality up to $1 + o(1)$ factors and where we used $n^{-100} < \epsilon<n^{-c}$ for a constant $c>0$ and $\log d=\log (\log^a n)=a \log\log n$ in the worst case by assumption. A more careful analysis 
 % shows that the constant multiplying $1/a$ in the exponent is $1$, and 
 % the details 
 can be found in Section~\ref{sec:ub-high-temp}.

\noindent\paragraph{Our techniques: lower bound.} Our lower bound uses a reduction from the communication problem of $\texttt{INDEX}$. In the \texttt{INDEX} problem Alice holds a binary string $\bx$ of length $n$ and Bob has an index $i\in [n]$. Alice sends Bob a message $m$, after receiving which Bob should be able to output $\bx_i$ with a constant advantage over random guessing, i.e. with $1/2+\Omega(1)$ success probability. The $\texttt{INDEX}$ problem requires $\Omega(n)$ communication from Alice to Bob, even if Alice and Bob share randomness and even if the input string $\bx$ contains exactly half zeros and half ones.
Although reductions from $\texttt{INDEX}$ are commonly used in streaming algorithms, see e.g.,~\cite{jayram2008one,haris2025compression,kochetkova2025streaming}, the challenge lies in mapping the instances into the problem of interest so that the bits can be decoded from the answers to the problem.

We start by presenting a lower bound for approximating the normalizing constant $\expker (K, q)$ in the streaming model, and then obtain a lower bound for approximating $\Attn(K, q, V)$ by reduction.

Let \texttt{ALG} be a streaming algorithm that computes an approximation $\widehat{\expker}(K, q)$ to the normalizing constant $\expker (K, q)$ as per Definition~\ref{def:streaming-kde}, namely
$$
\left|\widehat{\expker}(K, q)-\expker (K, q)\right|\leq \epsilon \expker (K, q).
$$
We show how to use \texttt{ALG} to obtain a communication efficient protocol for \texttt{INDEX} on binary vectors $\bx$ of length $2n$.
Let $r' := r/2$, where $r$ is the upper bound on the Euclidean norm of keys and queries in our setting. Using the shared randomness, Alice and Bob sample a point set $K = \{k_1, k_2, \cdots, k_{2n}\} \subset \R^d$ such that $k_i \sim \mathcal{N}\p{0, \frac{r'^2}{d} I_d}$ i.i.d. for each $k_i \in K$. Upon receiving her input $\bx \in \{0,1\}^{2n}$, Alice constructs the dataset
$$
K_\bx = \left\{ k_i \in K\mid \bx_i = 1\right\},
$$
feeds $K_\bx$ to her streaming algorithm and sends the state of the streaming algorithm \texttt{ALG} to Bob. In order to find out whether $\bx_i=1$ or $\bx_i=0$, Bob queries \texttt{ALG} at $q=k_i$. The intuition for this approach is simple. One notes that 
\begin{equation}\label{eq:thresholding}
\begin{split}
\expker(K_\bx, q)&=\bx_i\cdot e^{k_i^Tq}+\sum_{\substack{j=1\\j\neq i}}^{2n} \bx_j\cdot e^{k_j^T q}=\bx_i\cdot e^{r^2}+\sum_{\substack{j=1\\j\neq i}}^{2n} \bx_j\cdot e^{k_j^Tk_i},
\end{split}
\end{equation}
and therefore it is natural to hope that the sum over $j\neq i$ contributes a  `noise' term that is smaller than $e^{r^2}$. If that were the case, Bob would simply threshold $\expker(K_\bx, q)$, concluding that $\bx_i=1$ if $\expker(K_\bx, q)\geq \theta$, and that $\bx_i=0$ if $\expker(K_\bx, q)<\theta$, for a threshold $\theta$. This intuition is further supported by the fact that the random Gaussian vectors $k_j$ and $k_i$ are quite close to orthogonal with high probability. This is the approach used in~\cite{haris2025compression,kochetkova2025streaming}. However, one can verify that this intuition is only correct for extremely large values of $r$. Therefore this natural, but quite basic reduction is in fact not sufficiently powerful to give tight dependence on $r$ and $1/\epsilon$. 
    
\paragraph{Side information.} The problem with thresholding $\expker(K_\bx, q)$ in order to learn $\bx_i$ is that the noise term in \eqref{eq:thresholding} is too large. Therefore, we need to find a way to reduce it.  To this end, Alice sends some side information about $\bx$ to Bob, in the form of the low-degree part of the dataset $K_\bx$, essentially mimicking our upper bound approach. Just like in our upper bound, we write 

\[\expker(K_\bx, q) =\sum_{k \in K_\bx} \sum_{l=0}^\infty  \frac{\langle k,q \rangle^l}{l!} = \sum_{k \in K_\bx}  \expker_{\leq t} (k, q)+ \expker_{> t} (k, q)\]
and define feature maps $\psi_1(k)=(k^{\otimes l}/l!)_{l=0}^t$ and $\psi_2(q)=(q^{\otimes l})_{l=0}^t$ such that 
$\exp_{\leq t}(k, q)=\left\langle \psi_1(k), \psi_2(q)\right\rangle$.
Then Alice sends Bob the sketch  $\sum_{k\in K_\bx} \psi_1(k)$, from which Bob can compute 
$$\exp_{\leq t}(K_\bx, q)=\left\langle \sum_{k_i\in K} \psi_1(k), \psi_2(q)\right\rangle.$$

\paragraph{Signal to noise ratio and the Hermite expansion.} The hope for this approach is that the  truncated softmax kernel $\expker_{> t} (k, q)$ has smaller values and therefore introduces less noise. Specifically, the truncated version of ~\eqref{eq:thresholding} now becomes 

\begin{equation}\label{eq:thresholding-truncated}
\begin{split}
\expker_{> t}(K_\bx, q)&=\bx_i\cdot \expker_{> t}(k_i, k_i)+\sum_{\substack{j=1\\j\neq i}}^{2n} \bx_j\cdot \expker_{> t}(k_j, k_i).
\end{split}
\end{equation}
The noise term contributed by the sum on the rhs is indeed smaller now.
However, so is the signal term $\bx_i\cdot \expker_{> t}(k_i, k_i)$.
In principle, it is not clear which of the terms is reduced faster. This concern turns out to be well founded: if we use Taylor expansions as above, the signal drops too fast relative to the noise for us to achieve the desired bound. However, a more careful choice of the polynomial expansion, specifically, the use of the Hermite basis, allows us to carry out this approach. The detailed analysis is presented in Section~\ref{sec:lb-high-temp}. 
Here we just mention that the upper and lower bounds match quite tightly because the expression for the upper bound (\Cref{eq:ub}) is the same as for the lower bound, namely 
\[\tilde{\Omega}\qty(\max_{t} \min\left \{ \frac{(d+t)^{t}}{t!}, \frac{1}{\eps}\cdot \sum_{l=t+1}^\infty  \frac{r^{2l}}{l!} \right \}).\]

\subsection{Low temperature regime (large $r$)}

In this section we outline our results and techniques for the `low temperature'  regime of large $r$, namely when $r^2=\Omega(\log n)$.  In this regime the exponential dependence on $r$ becomes important, and we resolve the optimal dependence up to lower order terms in the exponent. 

\begin{theorem}
    \label{th:big_r_ub_lb}
    
   Assume that $r^2 = \Omega(\log n)$, i.e. $e^{r^2} = n^{\Omega(1)}$. Then the Softmax KDE problem \eqref{eq:streaming-kde} can be solved using
    \begin{align*}
        \tilde{O}\qty(\frac{e^{r^2(1 + o(1))}}{\eps})
    \end{align*}
    memory. Furthermore, this bound is almost tight for any $d = (\log n)^{1 + \Omega(1)}$, i.e. there is a
    \begin{align*}
        \tilde{\Omega}\qty(\min\qty{n, \frac{e^{r^2(1 - o(1))}}{\eps}})
    \end{align*}
    lower bound for the streaming space complexity.
\end{theorem}

\begin{remark}
    \label{remark:big_r_attn}
    This result also holds for the Streaming Attention problem (\Cref{def:streaming-attn}), but the lower bound holds only under additional assumption $r^2 < (\log n)/2$. See \Cref{th:big_r_ub_lb_attn} for the formal statement and \Cref{sec:attn_ub,sec:attn_lb} for details of the reduction.
\end{remark}

Our starting point is the algorithm BalanceKV (\cite{kochetkova2025streaming}), an online discrepancy minimization method based on~\cite{ALS21} that maintains a coreset approximating the dataset in all directions simultaneously. While this mechanism uses basic geometric properties of the dataset such as the points being bounded by $r$ in Euclidean norm, it is largely oblivious to the distribution of keys inside the Euclidean ball of radius $r$. In particular, for inputs consisting of tight clusters, it is strictly more effective to run separate instances of BalanceKV on each cluster rather than on the dataset as a whole. This is the observation that leads to our improvement over~\cite{kochetkova2025streaming} in terms of the dependence on $r$, which we also show to be optimal. We now give more details of the approach. 

\paragraph{Pseudorandom datasets.} A key tool in  our algorithm is a  \textit{pseudo-randomization procedure}  that decomposes the dataset into two types of components: random-like datasets and dense datasets contained in balls whose radius is smaller by a nontrivial factor, e.g., by a factor of $2$. We show that both types admit very compact coresets. The precise definition we use follows~\cite{ALRW17} (see also~\cite{CKNS20}):

\begin{definition}[Pseudo-random dataset; slightly specialized version of Definition~\ref{def:pseudorandom}]\label{def:prandom}
    A dataset $K$ of $n$ keys of Euclidean norm bounded by $r$ is $(r, \Delta, \tau)$-\emph{pseudo-random}\footnote{Our definition follows \cite{ALRW17,CKNS20}, but with $K$ in a ball instead of a sphere. This makes our notion somewhat counterintuitive. For instance, a dataset that is highly concentrated around $c$ is pseudo-random. We nevertheless retain the term to emphasize the connection with the earlier literature.} if
    \begin{gather}\label{eq:dense-region}
        \forall q, \|q\|_2\leq r \quad |\{k \in K \mid \<k, q> > \Delta r^2\}| \leq \tau n.
    \end{gather}
\end{definition}

We find the dense regions as per~\eqref{eq:dense-region}, i.e. small spherical caps containing a nontrivial fraction of the points, in the input dataset and recurse on them. These dense regions are by definition contained in small spherical caps of the original Euclidean ball of radius $r$, and therefore each such region can be enclosed by a Euclidean ball of a constant factor lower radius. Once the radius is reduced by at least a factor of $2$ in the sequence of recursive calls, we can stop the recursion because a stronger coreset for the dataset in question can be constructed -- see discussion below. The other termination condition is that the dataset is pseudorandom as per Definition~\ref{def:prandom}. In this case, we have a strong lower bound on the value of our softmax estimate. Intuitively, the dataset is close to random with respect to the query, in the sense that softmax KDE value cannot be much smaller than for a random dataset; see \Cref{lem:ps_rand_coreset} for more details. This stronger lower bound on softmax KDE value directly implies stronger coreset size bounds due to the multiplicative nature of our approximation guarantee.

We first present the coreset construction that we will invoke later in the recursive analysis. Although our main object is the kernel $\expker(\cdot,\cdot)$, it is useful to formulate the lemma for a general kernel $\gamma$, since, during the recursive calls, we work with shifted and rescaled versions of $\expker$.

\noindent {\em {\bf Lemma} (Restated from Lemma~\ref{lem:simple_coreset})
    Let $\gamma: \R^d \times \R^d \to \R$  be a positive definite kernel, and let $U \times V\subset   \R^d \times \R^d $ be a predefined subset of the domain. There is a randomized algorithm $\BaseCompress$ which receives $\gamma$ and $K \subset U$ as inputs. $\BaseCompress(\gamma, K)$ outputs $\widetilde{K} \subset K$ such that $|\widetilde{K}| \leq |K|/2$ and for any $q \in V$, with probability $1 - \delta$,
           \[\left|\gamma(K, q) - 2\gamma(\widetilde{K}, q)\right| \leq O\left(\sqrt{\max_{u \in U}\gamma(u, u)\cdot \max_{v \in V}\gamma(v, v)}\cdot \log (|K|/\delta)\right).\]
    $\BaseCompress(\gamma, K)$ runs in $O(|K|^2 \cdot T)$ time, where $T$ is the maximum time to access $\gamma(k,q)$.
}

This lemma crucially relies on algorithms for discrepancy minimization that have been the subject of a beautiful line of work in theoretical computer science, e.g. \cite{Ban10, LM15, BDGL18, MNT20, ALS21}, and, for a particularly efficient implementation, on the result of~\cite{ALS21}.

\paragraph{Low radius dataset as a leaf.}
In this case our algorithm for constructing coresets crucially exploits the asymmetry in Lemma~\ref{lem:simple_coreset} above, i.e. the fact that the error term on the right hand side depends on the product of kernel norms $\sqrt{\gamma(u, u)}$ of $u\in U$ and $\sqrt{\gamma(v, v)}$ of $v\in V$.

To see why this asymmetry helps, compare a dataset $K \subset B(c,r/2)$ ($B(c, r/2)$ stands for the Euclidean ball centered at $c$ with radius $r/2$) with its recentered version $K-c \subset B(0,r/2)$. Finding coresets for $K$ and $K-c$ is essentially the same approximation problem, since
\[
\exp(K,q)=\exp(\langle c,q\rangle)\cdot \exp(K-c,q),
\]
so any weighted coreset for one yields the same multiplicative guarantee for the other. Yet, since our coreset construction primitive uses embeddings of keys to determine the coreset, the actual choice of embeddings matters a lot. When $K$ is far from the origin, the common offset $c$ creates a large global bias in the kernel values and embeddings, worsening approximation quality. After recentering $K$ to $K - c$ and rescaling $K-c$ and $q$ so that both fit inside a smaller radius ball, however, this shared bias disappears, and the algorithm only has to capture the local variation inside the cluster. The fact that the small cluster has low radius leads to a smaller coreset. 

More precisely, any dense part $K'$ of the dataset partition $\mathcal{K}$ given by \Cref{lem:ps_randomification} is contained in a ball $B(c',r')$ much smaller than the query domain $B(0,r)$. If we first shift and rescale it to
\[
\sqrt{\frac{r}{r'}}\,(K'-c'),
\]
and rescale the query domain accordingly, then both the key and query domains lie in
$B(0,\sqrt{rr'})$. Applying our basic coreset construction in these coordinates
yields a multiplicative approximation error proportional to
\[
\frac{
\exp(\langle c',q\rangle)\cdot \exp(rr')
}{
\exp(\langle c',q\rangle)\cdot
\expker\left(\sqrt{\frac{r}{r'}}(K'-c'), \sqrt{\frac{r'}{r}}q\right)
}
=
\frac{\exp(rr')}{
\expker\left(\sqrt{\frac{r}{r'}}(K'-c'), \sqrt{\frac{r'}{r}}q\right)
}.
\]
Since we do not assume any additional structure of the shifted and rescaled
dataset or query, the best universal lower bound on the denominator is
$\exp(-rr')$. Thus the multiplicative error is bounded by
\[
\frac{\exp(rr')}{\exp(-rr')}=\exp(2rr').
\]

By contrast, applying the coreset construction directly to the original dataset
only guarantees an error term of order
\[
\exp\bigl(r^2/2 + (r'+\|c'\|_2)^2/2\bigr) \geq \exp(r^2).
\]
Since the best universal lower bound on $\expker(K',q)$ is $\exp(-r^2)$, this
gives the worst-case multiplicative guarantee
\[
\frac{\exp(r^2)}{\exp(-r^2)}=\exp(2r^2),
\]
which is worse than $\exp(2rr')$ as $r' < r$.

% \bor{Equations above do not showcase any improvement as $r'+\|c'\|_2$ is as big as $r$. What we show is that we can give better multiplicative guarantee:
% \linebreak
% ... we obtain a coreset whose worst-case multiplicative approximation quality is proportional to
% \[
% \frac{\gamma(K, q) - 2\gamma(\widetilde{K}, q)}{\gamma(K, q)} \lesssim \frac{\exp(\langle c',q\rangle)\cdot \exp(rr')}{\exp(\langle c',q\rangle)\cdot \exp(-rr')}
% \leq
% \exp\bigl(2rr'\bigr),
% \]
% whereas applying the coreset construction directly to the original dataset is only guaranteed
% \[
% \frac{\gamma(K, q) - 2\gamma(\widetilde{K}, q)}{\gamma(K, q)} \lesssim \frac{\exp(r^2)}{\exp(-r^2)}
% \leq
% \exp\bigl(2r^2\bigr).
% \]
% }
% \eee{done; see if you like it}

 See \Cref{fig:recenter+rescale} for an illustration of the recentering and rescaling procedure, and Section~\ref{sec:ub_low_temp}, \Cref{lem:small_r_coreset} for more details. Finally, we would like to note that while the step getting us from a coreset construction, or, rather, a compression algorithm, to an actual streaming algorithm is essentially an application of the classical Merge-and-Reduce framework, this application is nontrivial, for example, in how it interacts with the pseudorandomification procedure above -- see \Cref{lem:merge-reduce-w} for more details.

\paragraph{Pseudorandom dataset as a leaf.}  For pseudo-random sets in \Cref{lem:ps_rand_coreset}, the same transformation does not by itself improve the guarantee, because by construction the pseudo-random parts need not lie in balls smaller than $B(0,r)$. Instead, the gain comes from a better lower bound on the KDE value: for an arbitrary dataset, the worst case is when the mass concentrates on
$\text{argmin}_{k\in B(0,r)} \exp(\langle k,q\rangle),$
whereas pseudo-randomness rules out such concentration.

\begin{figure}[ht]
  \centering
  \vspace{5mm}
  \resizebox{\linewidth}{!}{%
  \begin{tikzpicture}[
    scale=1,
    transform shape,
    >=Latex,
    line cap=round,
    line join=round
  ]
    % Radii enlarged by factor 1.3
    \def\R{3.38}
    \def\rp{1.69}

    % Left red-ball center
    \def\cx{-1.85}
    \def\cy{0.35}

    % Separation of the two panels
    \def\Sep{7.9}

    % Shift sending c to 0
    \pgfmathsetmacro{\sx}{-\cx}
    \pgfmathsetmacro{\sy}{-\cy}

    % Radius of the black dashed ball: sqrt(r * r/2)
    \pgfmathsetmacro{\rmid}{sqrt(\R*\rp)}

    % Ten points spread through the overlap region
    \def\PointList{
      -2.30/0.80,
      -2.10/0.10,
      -2.25/-0.65,
      -1.75/1.30,
      -1.65/0.65,
      -1.55/-0.10,
      -1.35/-0.75,
      -1.15/0.95,
      -0.95/0.25,
      -0.85/-0.20
    }

    % ---------------- Left picture ----------------
    \begin{scope}[shift={(-\Sep,0)}]
      % white fillings
      \fill[white] (0,0) circle (\R);
      \fill[white] (\cx,\cy) circle (\rp);

      % blue ball
      \draw[blue, thick] (0,0) circle (\R);
      \fill[blue] (0,0) circle (1.8pt);
      \node[blue] at (0.20,-0.20) {$0$};

      % q on the blue boundary
      \fill ({\R*cos(40)},{\R*sin(40)}) circle (1.4pt);
      \node[black] at ({\R*cos(40)+0.28},{\R*sin(40)+0.20}) {$q$};

      % red ball
      \draw[red, thick] (\cx,\cy) circle (\rp);
      \fill[red] (\cx,\cy) circle (1.6pt);
      \node[red] at ({\cx-0.24},{\cy+0.24}) {$c$};

      % overlap points and label
      \foreach \x/\y in \PointList {
        \fill (\x,\y) circle (1.15pt);
      }
      \node[black] at (-0.55,1.72) {$K$};

      % Left callouts
      \node[black, align=left, anchor=west, font=\small] (Lblue) at (4.0,-1.05)
        {ball with radius $r$\\ centered at $0$};
      \draw[->, black, thick] (Lblue.west) -- (2.70,-0.58);

      \node[black, align=left, anchor=west, font=\small] (Lred) at (-7.0,-2.75)
        {ball with radius $r/2$\\ centered at $c$};
      \draw[->, black, thick] (Lred.east) -- ({\cx-0.32},{\cy-1.08});
    \end{scope}

    % Middle arrow and caption
    \draw[->, thick] (-1.35,0) -- (1.35,0);
    \node[black, font=\small] at (0,0.62) {Recenter + rescale};

    % ---------------- Right picture ----------------
    \begin{scope}[shift={(\Sep,0)}]
      % white fillings
      \fill[white] (0,0) circle (\R);
      \fill[white] (0,0) circle (\rp);

      % blue ball
      \draw[blue, thick] (0,0) circle (\R);

      % q on the blue boundary
      \fill ({\R*cos(40)},{\R*sin(40)}) circle (1.4pt);
      \node[black] at ({\R*cos(40)+0.28},{\R*sin(40)+0.20}) {$q$};

      % dashed black ball
      \draw[black, dashed, thick] (0,0) circle (\rmid);

      % red ball
      \draw[red, thick] (0,0) circle (\rp);

      % shifted points and their label
      \foreach \x/\y in \PointList {
        \fill ({\x+\sx},{\y+\sy}) circle (1.15pt);
      }
      \node[black] at (0, 2) {$K-c$};

      % five uniformly spaced contraction arrows: blue -> black
      \foreach \ang in {20,92,164,236,308} {
        \draw[black, dashed, ->, thick]
          ({\R*cos(\ang)},{\R*sin(\ang)}) --
          ({\rmid*cos(\ang)},{\rmid*sin(\ang)});
      }

      % five uniformly spaced expansion arrows: red -> black
      \foreach \ang in {56,128,200,272,344} {
        \draw[black, dashed, ->, thick]
          ({\rp*cos(\ang)},{\rp*sin(\ang)}) --
          ({\rmid*cos(\ang)},{\rmid*sin(\ang)});
      }

      % Right callouts
      \node[black, align=left, anchor=west, font=\small] (Rblue) at (4.00,-1.05)
        {ball with radius $r$\\ centered at $0$};
      \draw[->, black, thick] (Rblue.west) -- (2.85,-0.60);

      \node[black, align=right, anchor=east, font=\small] (Rred) at (-5.00,-2.75)
        {ball with radius $r/2$\\ centered at $0$};
      \draw[->, black, thick] (Rred.east) -- (-1.18,-1.18);

      \node[black, align=center, font=\small] (Rblack) at (0.5,4.60)
        {ball with radius $\sqrt{r\cdot r/2}$\\ centered at $0$};
      \draw[->, black, thick] (Rblack.south) -- (0.1,\rmid);
    \end{scope}
  \end{tikzpicture}%
  }
  \vspace{4mm}
  \caption{Recentering and rescaling procedure for low radius datasets}
  \label{fig:recenter+rescale}
  \vspace{5mm}
\end{figure}
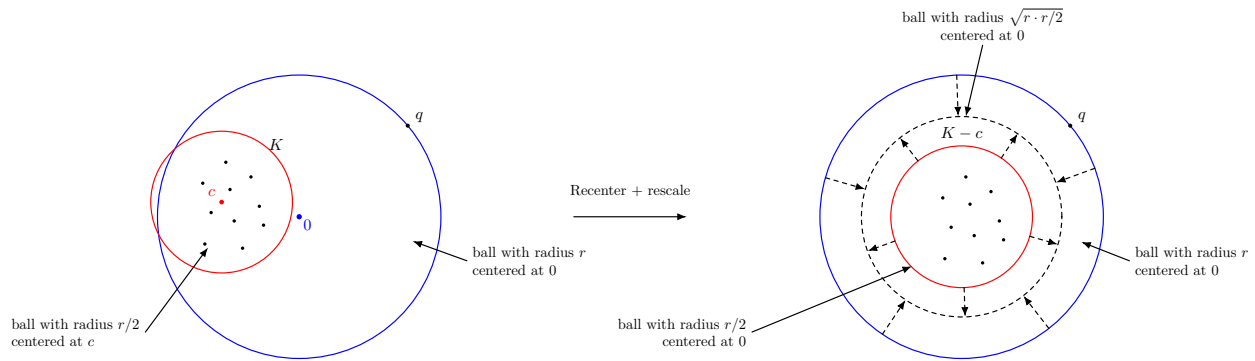

\section{Open Problems}
Our work naturally leads to several interesting open questions.
\begin{itemize}
\item Adaptivity: KV caching is often deployed in scenarios where the input tokens are generated sequentially via a feedback loop, where the attention mechanism is iteratively reapplied to the tokens produced earlier. In this setting, the input stream is generated adaptively, so it would be interesting to develop adversarially robust~\cite{ben2022framework} streaming algorithms for this task.

\item Gaussian kernel density estimation: the currently best results for Gaussian kernel density estimation~\cite{CKNS20} use independent random sampling, incurring a $1/\epsilon^2$ dependence on the precision parameter $\epsilon$. Recent work of~\cite{CKW24} was able to use discrepancy methods together with geometric space partitioning instead, achieving a $1/\epsilon$ dependence for {\em smooth} kernels (kernels decaying inverse polynomially with distance). There are challenges in extending the approach of~\cite{CKW24} beyond smooth kernels, but we hope that our present work for the softmax kernel, which is intimately connected to the Gaussian kernel, may offer insights.
\end{itemize}

\section*{Acknowledgements}
Ekaterina Kochetkova was supported by SNSF grant 10004935.

\section{Preliminaries}
\paragraph{Notations.}
Throughout this paper, we use $\mathcal{S}^{d-1}_r$ to denote the $d$-dimensional Euclidean sphere centered at $0$ with radius $r$, i.e. $\mathcal{S}^{d-1}_r := \{x \in \R^d: \|x\|_2 = r\}$, with the convention that $\mathcal{S}^{d-1}$ denotes the unit sphere. When the dimension is clear in the context, we drop the superscript and write $\mathcal{S}_r$. Similarly, we use $\mathcal{B}_r$  to denote the $d$-dimensional Euclidean ball centered at $0$ with radius $r$, i.e. $\mathcal{B}_r := \{x \in \R^d: \|x\|_2 \leq r\}$. For any integer $n \in \Z_{>0}$, we use $[n]$ to denote the set of integers $\{1, \cdots, n\}$.

\subsection{Multivariate Polynomials}

In this section, we formalize the polynomial sketching primitive used throughout the paper. \cref{{prop:poly_embed_dim}} bounds the number of monomials in a degree-$t$ polynomial, while \cref{lem:embed_memory} uses this bound to obtain a finite-dimensional embedding for degree-$t$ polynomials in the inner product. These together give a compact way to store the low-degree part of the kernel and will be used repeatedly later in the paper.

\begin{proposition}
    \label{prop:poly_embed_dim}
    Let $f: \R^d \to \R$ be a multivariate polynomial with degree at most $t$. The number of monomials of $f$ is at most
    \begin{align*}
        {d - 1 \choose d - 1} + {d \choose d - 1} + \dots + {d + t - 1 \choose d - 1} = {d + t \choose t} \leq \frac{(d+t)^t}{t!}.
    \end{align*}
    Moreover, when $d = d(t) = \omega(t^2)$, this bound is $(1+o_t(1))\frac{d^t}{t!}$.
\end{proposition}

\begin{lemma}\label{lem:embed_memory}
    Let  $\gamma(k, q) := \sum_{l = 0}^{t} a_l \<k, q>^l$ be a degree-$t$ polynomial in the inner product $\langle k, q\rangle$ for some $t$ and $a_l$. Fix an arbitrary $r >0 $. There exist embeddings $\psi_1, \psi_2$ such that for any finite set $K \subset \mathcal{B}_r$ and any $q \in \mathcal{B}_r$, we guarantee
    \[\qty|\Bigl \langle \sum_{k \in K} \psi_1(k), \psi_2(q) \Bigr \rangle - \sum_{k \in K}  \gamma(k, q)| \leq \delta,\] 
    Moreover, $ \sum_{k \in K} \psi_1(k)$ can be stored in $O\Bigl(\frac{(d+t)^{t}}{t!}\log\frac{|K|\cdot (d+t)^{t} \cdot\max_l |a_l|\cdot \max\{r, 1\}^t}{\delta}\Bigr)$ bits of space, and $\psi_1(k)$ can be evaluated in $O\Bigl(\frac{(d+t)^{t}}{t!}\log\frac{|K|\cdot (d+t)^{t} \cdot\max_l |a_l|\cdot \max\{r, 1\}^t}{\delta}\Bigr)$ time.
\end{lemma}
\begin{proof}
Fixing $K$, we can write $\sum_{k \in K} \gamma(k, q)$ as a function of $q$ is a multivariate polynomial from $\R^d$ to $\R$ with degree at most $t$.
We let 
\[
        \psi_1(k) := \qty(a_l \cdot k^{{\otimes l}})_{l = \overline{0..t}} \quad \text{ and } \quad \psi_2(q) := \qty(q^{{\otimes l}})_{l= \overline{0..t}}
\]
where $x^{\otimes k}$ denotes the $k$-th tensor power of the vector of $x$, then
\[
    \gamma(k, q) = \sum_{l = 0}^{t} a_l \<k, q>^l = \sum_{l = 0}^{t}a_l\<k^{\otimes l}, q^{\otimes l}> = \<\psi_1(k), \psi_2(q)>.
\]
Using \Cref{prop:poly_embed_dim}, the embedding dimension of $\psi_1$ after gluing together repeating monomials is $m=O\qty(\frac{(d+t)^{t}}{t!})$. Regarding the bit required per coordinate, we round the embedding vectors and argue that it gives the desired accuracy of approximation. Every coordinate of $k^{{\otimes l}}$, $q^{{\otimes l}}$ are bounded by $\max\{r, 1\}^t$.
Suppose we store a rounded vector $z$ such that $\| z - \sum_{k \in K}  \psi_1(k)\|_\infty \leq \eta$, then
\begin{align*}
    \Bigl |\Bigl \langle z,  \psi_2(q) \Bigr \rangle - \Bigl \langle \sum_{k \in K}  \psi_1(k),  \psi_2(q) \Bigr \rangle  \Bigr |  &\leq m \cdot \| z - \sum_{k \in K}  \psi_1(k)\|_\infty \cdot  \| \psi_2(q)\|_\infty \\
    &\leq m\cdot \eta \cdot  \max\{r, 1\}^t
\end{align*}
 Therefore, to guaranty the total error $\delta_0$, we can pick $\eta = \frac{\delta_0}{m\cdot  \max\{r, 1\}^t }$. Bounding over the sum of all keys, the bit complexity to store a single coordinate for total error $\delta_1$ is 
 \[O\qty(\log \frac{ |K|\max_l |a_l| \max\{r, 1\}^t \cdot (m\cdot \max\{r, 1\}^t)}{\delta_1}).\]
 Multiplying with $m=O\qty(\frac{(d+t)^{t}}{t!})$ gives the final space complexity. The evaluation time of $\psi_1$ is the time to write down these bits, thus is the same as the space complexity.
\end{proof}

\subsection{Coresets for kernel density estimation (KDE)}\label{sec:kde_coresets}

Our upper bounds in both the high- and low-temperature regimes are based on discrepancy minimization for KDE. If \(\gamma\) is positive definite (\Cref{def:pos_def}), then \(\gamma(k,q)\) can be written as the inner product of embeddings \(\psi(k)\) and \(\psi(q)\) in a Hilbert space \(H\). Viewed this way, constructing a coreset that approximates \(\gamma(K,q)\) for an unknown query \(q\) amounts to selecting a coreset of vectors \(\{\psi(k)\}_{k \in K}\) whose sum has approximately the same projection as \(\sum_{k \in K}\psi(k)\) onto an unknown direction. This is precisely the \textit{vector balancing problem}, and in \Cref{lem:simple_coreset} we formalize this connection and use it to build a compression algorithm.

\begin{comment}
The algorithm of \Cref{lem:simple_coreset}, \BaseCompress{}, however, is not directly applicable in the streaming setting, since it requires access to the full dataset. To address this, in \Cref{lem:merge_reduce} we show how to implement \BaseCompress{} via a standard Merge-and-Reduce scheme, so that it only needs to be applied locally.

This compression routine then serves as the main building block in both regimes. In the high-temperature regime, \BaseCompress{} is used to build coresets for the truncated softmax kernel \(\exp_{>t}(k,q)\); see \Cref{sec:ub-high-temp}, especially \Cref{lem:coreset}, for more details. In the low-temperature regime, it is used to build coresets for the softmax kernel \(\exp(k,q)\) on a carefully chosen partition of \(K\); see \Cref{sec:ub-high-temp}, especially \Cref{lem:ps_rand_coreset} and \Cref{lem:small_r_coreset}, for more details.
\end{comment}

\begin{definition}\label{def:pos_def}
    $\gamma: \R^d \times \R^d \to \R$ is a positive definite kernel if there exists a feature map $\psi : \R^d \to H$ to some Hilbert space $H$ such that $\gamma(x, y) = \<\psi (x), \psi(y)>_H$.
\end{definition}

\begin{lemma}
    \label{lem:simple_coreset}
    Let $\gamma: \R^d \times \R^d \to \R$  be a positive definite kernel, and let $U \times V\subset   \R^d \times \R^d $ be a predefined subset of the domain. There is a randomized algorithm $\BaseCompress$ which receives $\gamma$ and $K \subset U$ as inputs. $\BaseCompress(\gamma, K)$ outputs $\widetilde{K} \subset K$ such that $|\widetilde{K}| \leq |K|/2$ and for any $q \in V$, with probability $1 - \delta$,
           \[\left|\gamma(K, q) - 2\gamma(\widetilde{K}, q)\right| \leq O\left(\sqrt{\max_{u \in U}\gamma(u, u)\cdot \max_{v \in V}\gamma(v, v)}\cdot \log (|K|/\delta)\right).\]
    $\BaseCompress(\gamma, K)$ runs in $O(|K|^2 \cdot T)$ time, where $T$ is the maximum time to access $\gamma(k,q)$.

\end{lemma}

\begin{remark}
    In this paper, we by default apply \BaseCompress{} to the softmax kernel. Therefore, we omit the first parameter $\gamma$ and write $\BaseCompress(K)$.
\end{remark}
\begin{proof}
We reduce the problem to the inner-product setting, where the following result is known:
    \begin{theorem}[Theorem 9 of \cite{CKW24}. Adapted from \cite{ALS21}]\label{thm:compress_l2}
         Let $\gamma_0: \R^m \times \R^m \to \R$ be defined as $\gamma_0(k,q):=\langle k,q\rangle$ and let $U, V \subset \R^m$.  There is a randomized algorithm $\mathcal{A}$ which receives $K \subset U$ as input, and outputs $\widetilde{K} \subset K$ such that $|\widetilde{K}| \leq |K|/2$ and for any fixed $q \in V$, with probability $1 - \delta$,
           \[\left|\gamma_0(K, q) - 2\gamma_0(\widetilde{K}, q)\right| \leq O\left(\sqrt{\max_{u \in U}\gamma_0(u, u)\cdot \max_{v \in V}\gamma_0(v, v)}\cdot \log (|K|/\delta)\right).\]
           Furthermore, $\mathcal{A}$ does not require explicit access to $K, q$; it only needs oracle access to $\gamma_0(k,q)$. $\mathcal{A}$ runs in $O(|K|^2 \cdot T)$ time, where $T$ is the maximum time to access $\gamma_0(k,q)$.
    \end{theorem}

It  remains to reduce our setting to this inner-product case. Since $\gamma$ is positive definite, there exists a feature map $\psi:\R^d\to H$ into a Hilbert space $H$ such that
\[
\gamma(x,y)=\langle \psi(x),\psi(y)\rangle_H
\qquad\text{for all }x,y\in\R^d.
\]

We consider the subspace $\mathcal{P}$ spanned  by $\{\psi(k): k \in K\}$ and $\psi(q)$. Choosing an orthonormal
basis of  $\mathcal{P}$, we may identify  $\mathcal{P}$ isometrically with $\R^m$ for some
dimension $m \leq |K|+1$. Under this identification, each $\psi(k)$ and $\psi(q)$ becomes a vector in $\R^m$ , and all inner products among these vectors are preserved, i.e.,
\[
\gamma(k,q)=\langle \psi(k),\psi(q)\rangle_H
=\langle \pi\psi(k),\pi\psi(q)\rangle_{\mathbb R^m} = \gamma_0(\pi\psi(k),\pi\psi(q))
\]
where $\pi: \mathcal{P}\rightarrow\R^m$ is an isometric map defined by the chosen orthonormal basis.
Therefore, \cref{thm:compress_l2} applies to $\pi\psi(k),\pi\psi(q)$ and gives exactly the same error bound as claimed. Moreover, $\mathcal{A}$ only requires oracle access to inner products, which is exactly the kernel evaluations $\gamma(k,q)$. So we don't need to explicitly compute any embedding, and the total runtime is $O(|K|^2\cdot T)$ for $T$ being the evaluation time of $\gamma(k,q)$.

\end{proof}
The offline algorithm of \Cref{lem:simple_coreset} can be adapted to the streaming setting using the standard Merge-and-Reduce framework~-- see \Cref{fig:merge-reduce-standard_0} for an illustration. In this framework, the stream is partitioned into contiguous \emph{blocks}, and the offline compression algorithm is applied to each block. Whenever two contiguous blocks have been compressed, we merge their compressions and recursively compress the result. Equivalently, this process can be viewed as a binary tree: each leaf applies \BaseCompress{} to an initial stream block, and each internal node applies \BaseCompress{} to the union of the compressions produced by its descendant blocks. At any time, at most one node at each level of the tree stores data, and hence at most one node per level invokes  \BaseCompress{}. This yields an algorithm that is efficient in both space and time.  The same Merge-and-Reduce framework is used in \cite{kochetkova2025streaming} to obtain a streaming implementation of their algorithm. For completeness, we present this construction in \Cref{lem:merge_reduce}, specialized to \Cref{lem:simple_coreset} and the other offline routines developed in this paper, and include a proof sketch. Further details can be found in \cite{kochetkova2025streaming}.
\begin{figure}[H]
\centering
\vspace{5mm}
\resizebox{0.98\linewidth}{!}{%
\begin{tikzpicture}[
  >=Latex,
  font=\small,
  level/.style={font=\bfseries\small, anchor=east},
  stream/.style={-{Latex[length=2.2mm]}, line width=0.9pt},
  fan/.style={-{Latex[length=2mm]}, semithick},
  tocomp/.style={-{Latex[length=1.8mm]}, line width=0.85pt},
  block/.style={draw, rounded corners=2pt, fill=white, minimum width=15mm, minimum height=5.5mm, inner sep=0pt},
  comp/.style={draw, rounded corners=2pt, fill=white, minimum width=15.5mm, minimum height=6.3mm, inner sep=1pt},
  dots/.style={font=\large}
]

% ------------------------------------------------------------------
% Legend on the left
% ------------------------------------------------------------------
\node[anchor=east, font=\small] at (-0.15,4.75) {BC = \BaseCompress{}};

% ------------------------------------------------------------------
% Level 0
% ------------------------------------------------------------------
\node[level] at (0.75,0.0) {level 0};
\draw[stream] (1.35,0) -- (14.90,0) node[right] {input stream};

\node[block] (b01) at (2.30,0) {$1$};
\node[block] (b02) at (3.95,0) {$1$};
\node[block] (b03) at (5.60,0) {$1$};
\node[block] (b04) at (7.25,0) {$1$};
\node[block] (b05) at (8.90,0) {$1$};
\node[block] (b06) at (10.55,0) {$1$};
\node[block] (b07) at (12.20,0) {$1$};
\node[block] (b08) at (13.85,0) {$1$};

\node[comp] (c01) at (2.30,1.10) {BC};
\node[comp] (c02) at (3.95,1.10) {BC};
\node[comp] (c03) at (5.60,1.10) {BC};
\node[comp] (c04) at (7.25,1.10) {BC};
\node[comp] (c05) at (8.90,1.10) {BC};
\node[comp] (c06) at (10.55,1.10) {BC};
\node[comp] (c07) at (12.20,1.10) {BC};
\node[comp] (c08) at (13.85,1.10) {BC};

\foreach \b/\c in {b01/c01,b02/c02,b03/c03,b04/c04,b05/c05,b06/c06,b07/c07,b08/c08}{
  \draw[tocomp] (\b.north) -- (\c.south);
}

% ------------------------------------------------------------------
% Level 1
% ------------------------------------------------------------------
\node[level] at (0.75,3.65) {level 1};
\draw[stream] (1.40,3.65) -- (14.90,3.65);

\node[block] (b11) at (3.13,3.65) {$2$};
\node[block] (b12) at (6.43,3.65) {$2$};
\node[block] (b13) at (9.73,3.65) {$2$};
\node[block] (b14) at (13.03,3.65) {$2$};

\draw[fan] (c01.north) to[out=105,in=-90] ([xshift=-2pt]b11.south);
\draw[fan] (c02.north) to[out=75,in=-90]  ([xshift= 2pt]b11.south);

\draw[fan] (c03.north) to[out=105,in=-90] ([xshift=-2pt]b12.south);
\draw[fan] (c04.north) to[out=75,in=-90]  ([xshift= 2pt]b12.south);

\draw[fan] (c05.north) to[out=105,in=-90] ([xshift=-2pt]b13.south);
\draw[fan] (c06.north) to[out=75,in=-90]  ([xshift= 2pt]b13.south);

\draw[fan] (c07.north) to[out=105,in=-90] ([xshift=-2pt]b14.south);
\draw[fan] (c08.north) to[out=75,in=-90]  ([xshift= 2pt]b14.south);

\node[comp] (c11) at (3.13,4.75) {BC};
\node[comp] (c12) at (6.43,4.75) {BC};
\node[comp] (c13) at (9.73,4.75) {BC};
\node[comp] (c14) at (13.03,4.75) {BC};

\foreach \b/\c in {b11/c11,b12/c12,b13/c13,b14/c14}{
  \draw[tocomp] (\b.north) -- (\c.south);
}

% ------------------------------------------------------------------
% Level 2
% ------------------------------------------------------------------
\node[level] at (0.75,7.30) {level 2};
\draw[stream] (1.40,7.30) -- (14.90,7.30);

\node[block] (b21) at (4.78,7.30) {$4$};
\node[block] (b22) at (11.38,7.30) {$4$};

\draw[fan] (c11.north) to[out=105,in=-90] ([xshift=-2pt]b21.south);
\draw[fan] (c12.north) to[out=75,in=-90]  ([xshift= 2pt]b21.south);

\draw[fan] (c13.north) to[out=105,in=-90] ([xshift=-2pt]b22.south);
\draw[fan] (c14.north) to[out=75,in=-90]  ([xshift= 2pt]b22.south);

\node[comp] (c21) at (4.78,8.40) {BC};
\node[comp] (c22) at (11.38,8.40) {BC};

\foreach \b/\c in {b21/c21,b22/c22}{
  \draw[tocomp] (\b.north) -- (\c.south);
}

% continuation
\draw[fan] (c21.north) to[out=105,in=-120] (7.55,9.55);
\draw[fan] (c22.north) to[out=75,in=-60]  (8.65,9.55);
\node[dots] at (8.10,10.10) {$\vdots$};

\end{tikzpicture}%
}
\vspace{4mm}
\caption{Illustration of the standard Merge-and-Reduce algorithm (\Cref{lem:merge_reduce}).}
\label{fig:merge-reduce-standard_0}
\vspace{5mm}
\end{figure}
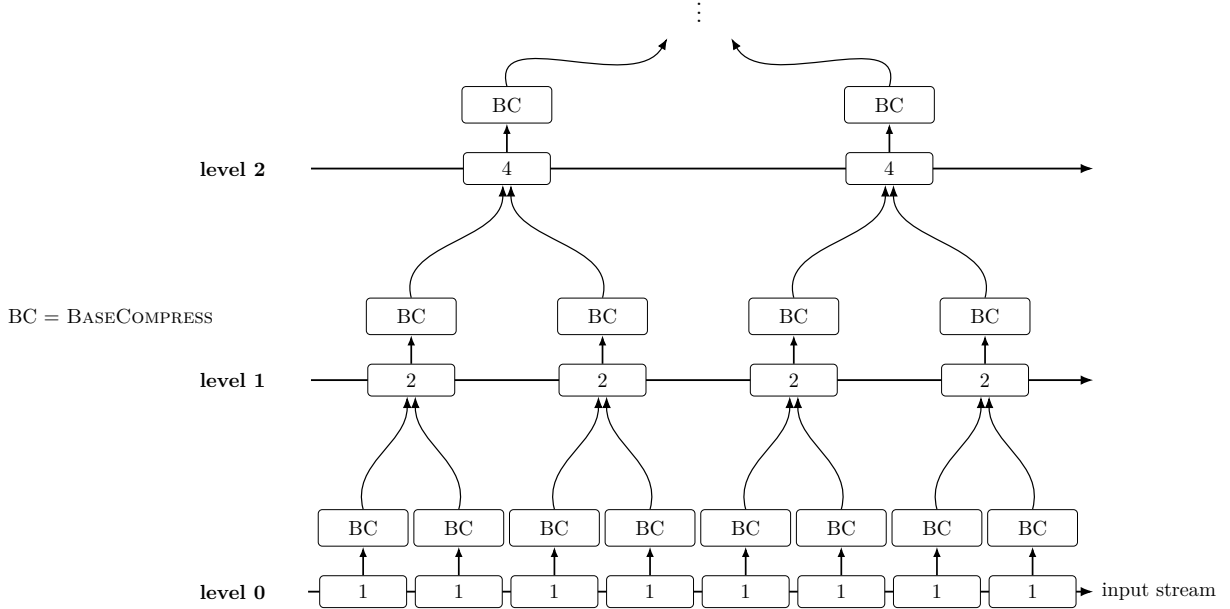

\begin{lemma}[Merge-and-Reduce]\label{lem:merge_reduce}
   Fix $n, \gamma, U, V$. Let $\mathcal{A}$ be an offline algorithm which receives $K \subset U$  as input, and outputs $\widetilde{K} \subset K$ such that $|\widetilde{K}| \leq |K|/2$ and for any $q \in V$, with probability $1 - \frac{1}{\poly n}$, 
   \begin{equation*}
       \left|\gamma(K, q) - 2\gamma(\widetilde{K}, q)\right| \leq \eta
   \end{equation*}
   for some error bound $\eta =  \eta(|K|)$.
   Let $T = T(|K|)$ be a deterministic upper bound on the runtime of $\mathcal{A}$. 
   Then for every block size $b \geq 1$, there is a streaming algorithm that processes a stream $k_1,  \cdots, k_n$ and, at every time step $j$, maintains a weighted coreset $C_j$ for the prefix $K_j = (k_1, \cdots, k_j)$ such that the following holds: 
    
    For any $q \in U$, with probability $1-\frac{1}{\poly n}$,
   \[\left|
  \gamma(K_j,q)-\gamma(C_j,q)\right|
   \leq O\qty(\eta(b) \cdot \frac{j}{b} \log\frac{j}{b}), \]
   where $\gamma(C_j,q) := \sum_{x \in C_j} w_x \gamma(x,q)$
   for $w_x$ denoting the weight of $x$ in the coreset $C_j$.
   
Moreover, $|C_j|  = O\left(b\log\frac{j}{b}\right)$. The streaming algorithm uses $\tilde{O}\left(b\log\frac{j}{b}\right )$ bits of storage space and $O\left(T(b)\cdot \log\frac{j}{b}\right)$ update time.

\end{lemma}

\begin{proof}[Proof sketch]
The algorithm partitions the stream into consecutive blocks of size $b$. Whenever a block is completed, the algorithm applies $\mathcal{A}$ to that block. And whenever two compressed sets appear at the same level of the of the Merge-and-Reduce binary tree, the algorithm merges them and applies $\mathcal{A}$  again to their union.
Therefore, for each $\ell \geq 0$, a node at level $\ell$ represents the union of $2^\ell$ initial blocks. Hence each surviving point at level $\ell$ carries weight $ 2^{\ell}$. Let $C_j$ be the union of all weighted compressed sets stored by the Merge-and-Reduce tree after processing $K_j$.

We first bound the size of $C_j$ and the space usage. At any time, there is at most one nonempty compressed set at each level, and the number of levels is $O(\log\frac{j}{b})$. Since each such set contains at most $b$ points, it follows that $|C_j|  = O\left(b\log\frac{j}{b}\right)$ and the streaming algorithm uses $\tilde{O}\left(bd\log\frac{j}{b}\right )$ bits space\footnote{Here and throughout, we assume that every coordinate of the stream elements can be stored in polylog$(n)$ bits.}, which is $\tilde{O}\left(b\log\frac{j}{b}\right)$  in our regime $d \leq \poly\log(n)$.
For the update time, a newly arriving stream element can trigger at most one invocation of 
$\mathcal{A}$ per level, which takes $O\left(T(b)\cdot \log\frac{j}{b}\right)$ time.

It remains to bound the error at time step $j$. At the bottom level, each invocation of $\mathcal{A}$ is applied to a block of size $b$, and therefore contributes additive error at most $\eta(b)$. Since there are at most $O(\frac{j}{b})$ such invocations, the bottom level gives  $O(\eta(b)\frac{j}{b})$ error.

Now consider a level $\ell\geq 1$. Each invocation at level $\ell$ corresponds to a compressed set representing $2^\ell$ original blocks, so its weighted error is larger by a factor of $2^\ell$. On the other hand, the number of invocations at level  $\ell$ is smaller by the same factor, therefore, the total contribution of level \(\ell\) is again $O(\eta(b)\frac{j}{b})$.
Since the number of levels is $O(\log \frac{j}{b})$, summing over all levels yields the claimed error bound.

Finally, the total number of invocations of $\mathcal{A}$ up to time $j$ is indeed polynomially bounded in $n$. Since each invocation succeeds with probability at least $1-\frac{1}{\poly n}$, a union bound implies that all invocations succeed simultaneously with probability at least $1-\frac{1}{\poly n}$.

\end{proof}

\subsection{Hermite Polynomials}

Let $\{\text{He}_l\}_{l=0}^{\infty}$ denote the sequence probabilists' Hermite polynomials; i.e. 
\[\text{He}_l(x) = (-1)^l e^{x^2/2}\cdot \frac{d^l}{dx^l}(e^{-x^2/2}).\] And let $\ph_{l, \lambda}(x)$ denote the rescaled version $\ph_{l, \lambda}(x) := \frac{1}{\sqrt{l!}}\text{He}_l\qty(x / \lambda)$ with rescaling factor $\lambda > 0$. Below are all the facts we use about Hermite polynomials. Proposition \ref{prop:hermite:basis} is the fundamental property of Hermite polynomials (see, e.g., Corollary 3.1 of ~\cite{A18}). Proposition \ref{prop:hermite:coefficients} and \ref{prop:hermite:exponent_expansion} can be found in Chapter 18 of  ~\cite{DLMF}.

\begin{proposition}
    \label{prop:hermite:basis}
    For any $\lambda > 0$, $\{\ph_{l, \lambda}\}_{l=0}^{\infty}$ is an orthonormal basis in $L^2(\mu)$ for $\mu = \mathcal{N}(0, \lambda^2)$.
\end{proposition}

\begin{proposition}[Explicit Coefficient Formula]
    \label{prop:hermite:coefficients}  
    $\text{He}_l(x)$ is a polynomial of degree exactly $l$, with the explicit formula $\text{He}_l(x) = l!\sum_{h=0}^{\floor{l/2}}\frac{(-1)^h}{h!2^h}\frac{x^{l-2h}}{(l-2h)!}$.
\end{proposition}

\begin{corollary}
    \label{cor:hermite:asymptotic}
    $\text{He}_l(x) = x^l - \frac{l(l-1)}{2}x^{l-2} + l!\sum_{h=2}^{\floor{l/2}}\frac{(-1)^h}{h!2^h}\frac{x^{l-2h}}{(l-2h)!}$. When $x = \omega(l)$, the last term is upper bound by $l^4 x^{l-4}$, so $\text{He}_l(x) = x^l \pm O(l^2x^{l-2}) = (1 \pm o(1))x^l$.
\end{corollary}

\begin{proposition}[Generating Function]
    \label{prop:hermite:exponent_expansion}
    The exponential generating function of $\text{He}_l(x)$ is
    \[\sum\limits_{l = 0}^{\infty}\frac{z^l}{l!}\text{He}_l(x) = e^{xz - \frac{z^2}{2}}.\]
\end{proposition}

\begin{corollary}
\label{cor:hermite:exponent_expansion}
   Substituting $x$  with $x /\lambda$ and $z$ with $\lambda$,  the  generating function gives
    \[e^x = e^{\lambda^2/2}\sum\limits_{l = 0}^{\infty}\frac{\lambda^l}{\sqrt{l!}}\ph_{l, \lambda}(x).\]
\end{corollary}

\begin{corollary}
\label{cor:hermite:derivative}
    For any $x$, the function $f(s):=e^{xs-s^2/2}$ satisfies $f^{(l)}(s) = \text{He}_l(x-s)f(s)$
    for all $s$.
\end{corollary}
\begin{proof}
    Substituting $x$  with $x-s$,  the  generating function gives
    \[\sum\limits_{l = 0}^{\infty}\frac{z^l}{l!}\text{He}_l(x-s) = e^{(x-s)z - \frac{z^2}{2}}.\]
    For $f(s):=e^{xs-s^2/2}$, this implies
    \[f(s+z)  = f(s)e^{(x-s)z - \frac{z^2}{2}} = f(s)\sum\limits_{l = 0}^{\infty}\frac{z^l}{l!}\text{He}_l(x-s).\] 
    On the other hand, $f(s+z) = \sum\limits_{l = 0}^{\infty}\frac{z^l}{l!}f^{(l)}(s)$ by Taylor's theorem.
    Comparing the coefficients of $z^l$ gives 
    \[f^{(l)}(s) = \text{He}_l(x-s)f(s).\]
\end{proof}

\subsection{Communication Complexity}

For a function $f: \mathcal{X} \times \mathcal{Y} \to \{0,1\}$, we denote by $R^{\text{pub}, \to}(f)$ the minimum communication of a randomized, public coin one-way protocol for two parties, Alice and Bob, to compute $f$ with probability at least $4/5$.
We consider the indexing problem with fixed support size where Alice has $2n$ bit string with Hamming weight $n$, and Bob has an index in $\{1, \ldots, 2n\}$.
Alice sends Bob a message, and Bob's goal is to determine the bit in Alice's string corresponding to his index.
We say $f = \texttt{INDEX}_{2n, n}$ where $\mathcal{X} = \{\mathbf{x} \in \{0,1\}^{2n}: \|\mathbf{x}\|_0 = n\}$, $\mathcal{Y} = \{1,\ldots, 2n\}$, and $f(\mathbf{x},j) = \mathbf{x}_j$. We will make use of the following theorem which follows the same argument as Theorem 5 of \cite{kremer95communication}.%along with the observation that the VC dimension of $f$ does not change by more than a constant factor by fixing the support size of $\mathbf{x}$ to be half of the domain.
\begin{theorem}\label{thm:indexing}
    $$R^{\text{pub}, \to}(\texttt{INDEX}_{2n, n}) = \Omega(n).$$
\end{theorem}

% % \todo{We want to change this theorem (a) to be about the distributional comm complexity over random queries $i$ (we are fine with worst-case or random choice of $\bx$). We also want the success probability of the protocol to be only slightly greater than $1/2$, say $2/3$ would suffice.}

% \bor{Actually I think current statement is fine - we don't need random $\bx$ as our dataset is already random. Our recovery probability is also $0.9$ so far (see \Cref{lem:gap})}

\subsection{High dimensional geometry}
\begin{proposition}[\cite{vershynin2026highdimprob}]
    \label{prop:norm_gauss}
   Let $g \sim \mathcal{N}(0, I_d)$. Then there exists a universal  constant $C>0$ such that for any $w \geq 1$,
    \begin{equation*}
        \Pr\p{\abs{\|g\|_2 - \sqrt{d}} \leq w} \geq 1- 2e^{-Cw^2} .
    \end{equation*}
\end{proposition}

\begin{proposition}[\cite{vershynin2026highdimprob}]
    \label{prop:rand_scalar_prod}
    Let $n \geq 10$ and let $v_1, \cdots, v_n$  be independent, uniform samples from the unit sphere $\mathcal{S}^{d-1}$. Then there exists a universal constant $C>0$ such that
    \begin{equation*}
        \Pr\p{\forall i\neq j \in [n]: \abs{\<v_i, v_j>} \leq C \sqrt{\frac{\log n}{d}}} \geq 1 - \frac{1}{n^3}.
    \end{equation*}
\end{proposition}

\begin{lemma}\label{lem:expectation_product}
     Let  $u, v$  be independent, uniform samples from the unit sphere $\mathcal{S}^{d-1}$. Then for any integer $t >1$,
    \begin{align*}
        \E \<u, v>^{2t}  \leq \frac{2^t t!}{d^t}.
    \end{align*}
\end{lemma}
\begin{proof}
    Since the uniform distribution over $\mathcal{S}^{d-1}$ is rotationally invariant, we can fix $v = (1, 0, \dots, 0)$ without loss of generality. In addition, the uniformly random $u$ from $\mathcal{S}^{d-1}$can be equivalently sampled as $u = \frac{g}{\|g\|_2}$ for $g  = (g_1, \cdots, g_d) \sim \mathcal{N}(0, I_d)$. Let $u_1$ denote the first coordinate of $u$, we have that 
    \[\<u, v>^2 = u_1^2 = \frac{g_1}{\|g\|^2_2} = \frac{g_1^2}{g^2_1+ \cdots+ g^2_d}\]
    for i.i.d. $g_1, \cdots, g_d \sim \mathcal{N}(0, 1)$. By the definition of Gamma distributions, $g^2_1 \sim \chi^2_1 = \Gamma(\frac{1}{2},2)$ and independently, $(g^2_2 + \cdots + g^2_d)  \sim  \chi^2_{d-1}  = \Gamma(\frac{d-1}{2},2)$. Therefore, 
    \[\<u, v>^2 \sim \frac{\Gamma(\frac{1}{2},2)}{\Gamma(\frac{1}{2},2) + \Gamma(\frac{d-1}{2},2)} = \text{Beta}(\frac{1}{2}, \frac{d-1}{2})\]
    by the standard relationship between Gamma and Beta distributions. By the density of the Gamma distribution, for $Z \sim \text{Beta}(\alpha, \beta)$, the $t$-th moment is $\E[Z^t] = \prod_{l=0}^{t-1} \frac{\alpha+l}{\alpha+\beta+l}$. Setting $\alpha = 1/2$ and $\beta = (d-1)/2$, we have
    \[
    \E[\<u, v>^{2t}] = \prod_{l=0}^{t-1} \frac{1/2 + l}{d/2 + l} = \prod_{l=0}^{t-1} \frac{2l + 1}{2l + d}.
    \]
    To prove the claimed upper bound, we observe that $2l+1 \leq 2(l+1)$ and $2l+d \geq d$, thus
    \[
    \prod_{l=0}^{t-1} \frac{2l + 1}{2l + d} \leq \prod_{l=0}^{t-1} \frac{2(l+1)}{d} = \frac{2^t}{d^t} \prod_{l=0}^{t-1} (l+1) = \frac{2^t t!}{d^t}.
    \]
\end{proof}

\begin{lemma}
    \label{lem:rand_dataset_kde} 
    Fix an arbitrary $r >0 $. Let $K \subset \mathcal{S}_r$ be a set containing i.i.d. uniformly random samples  from $\mathcal{S}_r$. Then for any  $\alpha > 0$ and any $q \in  \mathcal{B}_{r}$ (which can be chosen from $K$), with high probability,
    \begin{align*}
        \Omega(|K|) \leq \sum_{k \in K} e^{\alpha\langle k, q \rangle} \leq |K|e^{O\left(\alpha r^2\sqrt{\frac{\log |K|}{d}}\right)} + e^{\alpha r^2}.
    \end{align*}
\end{lemma}
\begin{proof}
    First, we note that $q$ is independent of at least $|K|-1$ points from $K$, and these $|K|-1$ points are independent of each others. Since the distribution of each $k \in K$ is rotationally invariant, we have that $\Pr[\<k, q> \geq 0] = 1/2$, and the indicator variable $\mathds{1}(\langle k, q \rangle \geq 0)$ are independent. Therefore, 
    \begin{align*}
        \sum_{k \in K} e^{\alpha\langle k, q \rangle} \geq \sum_{\substack{k \in K \\ k \neq q}} e^{\alpha\langle k, q \rangle} \geq \frac{|K|}{4}e^0 = \Omega(|K|)
    \end{align*}
    where the second inequality holds with high probability by applying the Chernoff bound to the sum of indicators $\sum_{\substack{k \in K :  k \neq q}}\mathds{1}(\langle k, q \rangle \geq 0)$ and bound the deviation from the mean $\frac{|K|}{2}$ by $\frac{|K|}{4}$.
   On the other hand, by \Cref{prop:rand_scalar_prod} and the  union bound, we have, with high probability,
    \begin{align*}
        \<k, q> < C\|k\|_2\|q\|_2\sqrt{\frac{\log |K|}{d}} = O\qty(r^2\sqrt{\frac{\log |K|}{d}})
    \end{align*}
    simultaneously for all $k \in K$ such that $k\neq q$. Therefore, 
    \begin{align*}
        \sum_{k \in K} e^{\alpha\langle k, q \rangle} \leq e^{\alpha r^2} + \sum_{\substack{k \in K \\ k \neq q}} e^{\alpha\langle k, q \rangle} \leq |K|e^{O\qty(\alpha r^2\sqrt{\frac{\log |K|}{d}})} + e^{\alpha r^2}.
    \end{align*}
\end{proof}

\section{Proof of \Cref{th:const_r_ub_lb}}
In this section we prove \Cref{th:const_r_ub_lb}, restated here for convenience of the reader:
\begin{theorem*}[Restated from \Cref{th:const_r_ub_lb}]
    Assume that $r = \log^{o(1)} n$, $\log n \leq d \leq \log^a n$ for some constant $a > 1$ and $n^{-100} < \eps < n^{-c}$ for an arbitrarily small constant $c > 0$. Then the Softmax KDE problem \eqref{eq:streaming-kde} can be solved using
    \begin{align*}
        \tilde{O}\qty(\qty(\frac{1}{\eps})^{1 - \frac{1}{a} + o(1)})
    \end{align*}
    memory, where $\tilde{O}(\cdot)$ suppresses $O(\poly\log(n))$ factors. Furthermore, this bound is tight for any $d \geq \log^a n$, $a > 2$, i.e. there is a
    \begin{align*}
        \tilde{\Omega}\qty(\min\qty{n, \qty(\frac{1}{\eps})^{1 - \frac{1}{a} - o(1)}})
    \end{align*}
    lower bound for the streaming space complexity.
\end{theorem*}

We prove the upper bound in \Cref{sec:ub-high-temp} below, the lower bound in \Cref{sec:lb-high-temp} below. These two results are put together in \Cref{sec:high_tem_putting_together} to obtain the proof of \Cref{th:const_r_ub_lb}.
\subsection{Upper Bound}\label{sec:ub-high-temp}
In this section, we assume that the parameters satisfy $r = \log^{o(1)} n$, $\log n \leq d \leq \log^a n$ for some constant $a > 1$ and $n^{-100} < \eps < n^{-c}$ as in the statement of \cref{th:const_r_ub_lb}.
Recall that for any $k, q\in \R^d$, we denote $\expker_{\leq t} (k,q ) := \sum_{l=0}^t \frac{\langle k,q \rangle^l}{l!}$  and $\expker_{>t} (k,q ) := \sum_{l=t+1}^{\infty} \frac{\langle k,q \rangle^l}{l!}$.
So, for any nonnegative integer $t$,
\[\expker(K, q) =\sum_{k \in K} \sum_{l=0}^\infty  \frac{\langle k,q \rangle^l}{l!} = \sum_{k \in K}  \expker_{\leq t} (k,q )+ \expker_{> t} (k,q ).\]

Given a stream of at most $n$ data points (denoted as $K$), the data structure maintains:
\begin{enumerate}
    \item a  sketch $\sum_{k \in K} \psi(k)$ for some feature map $\psi$, and
    \item a coreset $C \subset K$,
\end{enumerate}
such that for any query $q \in \R^d$ and some fixed  integer $t$  of our choice,  one can use the  sketch $\sum_{k \in K} \psi(k)$ to approximate $\sum_{k \in K}  \expker_{\leq t} (k,q)$, and the coreset $C$ to approximate $\sum_{k \in K}  \expker_{>t} (k,q)$.

This is formalized in the following \cref{lem:map} and \cref{lem:coreset}. We will give the constructions of $\psi$ and $C$, and then show that there exists $t$ such that the total space usage of
$\sum_{k \in K} \psi(k)$ and $C$ is $\tilde{O}\qty( (1/\eps)^{1- \frac{1}{a} + o(1)})$.

\begin{lemma}\label{lem:map}
Let  $t$ be an integer such that $0 \leq t \leq d$. There is a streaming algorithm that processes a stream $K = (k_1,  \cdots, k_n) \subset  \mathcal{B}_r$ and, at every time step $j \in [n]$, maintains a sketch $ \sum_{k \in K_j} \psi_1(k)$  for the prefix $K_j = (k_1, \cdots, k_j)$ such that the following holds: 
    
    There is an explicit feature map $\psi_2$ such that for any $q \in \mathcal{B}_r$,
    \[\qty|\left\langle \sum_{k \in K_j} \psi_1(k), \psi_2(q)  \right\rangle - \sum_{k \in K_j}  \expker_{\leq t}(k, q)| \leq \eps \expker(K_j, q).\]
    Moreover, $ \sum_{k \in K_j} \psi_1(k)$ can be stored in $\tilde{O}\Bigl(\frac{(d+t)^{t}}{t!}\Bigr)$ bits of space, and can be updated in $\tilde{O}\Bigl(\frac{(d+t)^{t}}{t!}\Bigr)$  time
\end{lemma}
\begin{proof}
    The construction of $\psi_1, \psi_2$ is given by \Cref{lem:embed_memory}, with $\gamma$ set to $\expker_{\leq t}$ and the target additive error $\delta$ set to $\eps e^{-r^2}$. Since $\expker(K_j,q)\ge e^{-r^2}$ for every $q\in \mathcal{B}_r$, this gives error at most
$\epsilon\expker(K_j,q)$.

Moreover, the logarithmic factor in \Cref{lem:embed_memory} is
\[
O\left(
\log n+\log \left(\frac{(d+t)^t}{t!}\right)+t\log\max\{r,1\}+\log(1/\epsilon)+r^2
\right).
\]
Under the parameter regime of Theorem 2.1, this factor is $\tilde O(1)$.
Hence the space and update time are
$\tilde O((d+t)^t/t!)$.
\end{proof}

\begin{lemma}\label{lem:coreset}
Let  $t$ be an integer such that $0 \leq t \leq d$.  There is a streaming algorithm that processes a stream $K = (k_1,  \cdots, k_n) \subset \mathcal{B}_r$ and, at every time step $j$,  maintains a weighted coreset $C_j$ for the prefix $K_j = (k_1, \cdots, k_j)$ such that the following holds: 
    
    For any fixed $q \in \mathcal{B}_r$, with probability $1-\frac{1}{\poly n}$,
    \[
        \left| \sum_{k \in K_j}\expker_{> t}(k, q) -   \sum_{k \in C_j} w_k \expker_{> t}(k, q) \right| \leq \eps  \expker(K_j, q)
     \]
     where $w_k$ denoting the weight of $k$ in the coreset  $C_j$. Moreover, $C_j$ can be stored in  $\tilde{O} \qty(\frac{ e^{r^2}}{\eps} \cdot \sum_{l=t+1}^\infty  \frac{r^{2l}}{l!} )$ bits of space, and can be updated in  $\tilde{O} \qty((\frac{ e^{r^2}}{\eps} \cdot \sum_{l=t+1}^\infty  \frac{r^{2l}}{l!} )^2)$ time.

\end{lemma}
\begin{proof}
For any integer $t \geq 0$, the kernel $\expker_{> t}$ is positive definite. Therefore, we can apply \cref{lem:simple_coreset} with $\gamma = \expker_{> t}$, $\delta = \frac{1}{\poly n}$ and $U, V = \mathcal{B}_r$. This gives an offline black-box algorithm, \BaseCompress, which we can plug into the Merge-and-Reduce framework of \cref{lem:merge_reduce} to obtain a streaming algorithm, whose approximation error and space and time complexities depend on a block size parameter $b$ that we may choose.

We want to choose $b$ so that the streaming algorithm gives the desired error bound $ \eps  \expker(K_j, q)$, i.e. 
\[\eta(b) \cdot  \frac{j}{b} \log \frac{j}{b} \leq \Theta(\eps  \expker(K_j, q)).\]
$\eta (b)$ here is the error of \BaseCompress, which, regardless of the choice of $b$, is always at most 
\[O\qty(\max_{u: \|u\|_2 \leq r} \expker_{> t}(u, u) \cdot \log n) =\tilde{O} \qty(\sum_{l=t+1}^\infty  \frac{r^{2l}}{l!}).\] Therefore, it suffices to choose $b \in [n]$ that satisfies
\[b \geq \tilde{\Theta} \qty(\frac{1}{\eps}\cdot  \frac{j}{ \expker(K_j, q)} \cdot \sum_{l=t+1}^\infty  \frac{r^{2l}}{l!} ).\]
Since the algorithm does not know the exact value of $ \expker(K_j, q)$ ahead of time, we lower bound it as
$ \expker(K_j, q) \geq j e^{-r^2} $. So we may choose 
\[b  = \tilde{\Theta} \qty(\frac{ e^{r^2}}{\eps} \cdot \sum_{l=t+1}^\infty  \frac{r^{2l}}{l!} )\]
which gives the desired approximation error. As stated in \cref{lem:merge_reduce}, the space complexity is $\tilde{O}\left(b\log\frac{j}{b}\right )$; the update time is $\tilde{O}(b^2)$ since $\exp_{>t}$ can be computed in $\tilde{O}(d+t) = \tilde{O}(1)$ time.  Plugging in our choice of $b$ gives the claimed space and time.

\end{proof}

Running both streaming algorithms from \cref{lem:map} and \cref{lem:coreset}, we can approximate both $ \expker_{\leq t} (k,q)$ and $ \expker_{> t} (k,q)$, with $O(\eps \expker(K_j, q))$ additive error, which satisfies the error guarantee in the upper bound of \cref{th:const_r_ub_lb}.

Finally, it remains to choose $t$ with $0 \leq t\leq d$  to obtain the claimed upper bound on the total space usage. For a given $t$,
the total space is
\[
\tilde O\left(
\frac{(d+t)^t}{t!}
+
\frac{e^{r^2}}{\epsilon}
\sum_{l=t+1}^{\infty}\frac{r^{2l}}{l!}
\right).
\]
We choose $t>2r^2$. 
Then we can derive an upper bound:

\begin{align*}
    \frac{(d+t)^{t}}{t!} + \frac{e^{r^2}}{\eps}\cdot \sum_{l=t+1}^\infty  \frac{r^{2l}}{l!} &\leq \frac{(d+t)^{t}}{t!} + \frac{1}{\eps}\cdot   \frac{2r^{2t}}{t!} \\
    &\leq (d+t)^{t} t^{-t}e^t + \frac{2e^{r^2}}{\eps}\cdot  t^{-t}e^tr^{2t} 
\end{align*}
where the first inequality follows from the fact that when $t > 2r^2$, for any $l \geq t$, 
\[\frac{r^{2(l+1)}/(l+1)!}{r^{2l}/l!} = \frac{r^2}{l+1} \leq \frac{r^2}{t+1} < 1/2,\] so the tail terms $\frac{r^{2l}}{l!} $ decrease geometrically with a ratio $< 1/2$. The second inequality follows from Stirling's approximation that $\frac{1}{t!} \leq t^{-t}e^t$.
To balance the two terms in this upper bound, we observe that their ratio is 
\[\rho := \qty((d+t)^{t} t^{-t}e^t)\Big / \qty(\frac{2e^{r^2}}{\eps}\cdot  t^{-t}e^tr^{2t}) = \frac{\eps}{2e^{r^2}} (\frac{d+t}{r^2})^t.  \]
We choose $t = \floor{\frac{\log (e^{r^2} / \eps)}{\log (d/r^2)}}$, which indeed satisfies $2r^2 \ll t \ll d$ for our regime of $d \in [\log n, \log^a n]$, $r = (\log n)^{o(1)}$, and $\log (1/\eps) = \Theta(\log n)$. Moreover, this choice of $t$ implies 
\[  \frac{2e^{r^2}}{\eps} \asymp (\frac{d}{r^2})^t ,\]
so the ratio $\rho$ is at most $O(1) \cdot (1+\frac{t}{d})^t = O(e^{t^2/d})$. Importantly, this is $(1/\eps)^{o(1)}$ since $t^2/d = O(\frac{\log^2(1/\eps)}{d \log^2d}) = o(\log(1/\eps))$. In other words, the two terms are balanced, up to an $(1/\eps)^{o(1)}$ factor. Finally, to derive the space complexity bound, we compute
\begin{align}
\begin{split}
    t^{-t}e^tr^{2t} &= \qty(\frac{(1 + o(1))\log d\log^{o(1)} n}{\log (1/\eps)})^t \label{eq:small_r_runtime} \\
    &=  \exp(t \cdot \log \frac{\log^{o(1)} n}{\log (1/\eps)}) 
    \\
    &= \exp(-(1 + o(1))\frac{\log (1 / \eps)}{\log d}\log \log (1/\eps)) \\
    &= (1/\eps)^{-\frac{\log\log(1/\eps)}{\log d} (1 + o(1))}  
    \end{split}
\end{align}

This, multiplied with $\frac{2e^{r^2}}{\eps}$, is the second term in the upper bound of the space.  We also observe that $2e^{r^2} = (1/\eps)^{o(1)}$ in our regime. Therefore, the total  space usage is 
\[\tilde{O}(1/\eps \cdot (1/\eps)^{o(1)} \cdot (1/\eps)^{-\frac{\log\log(1/\eps)}{\log d} (1 + o(1))}.\]
In particular, because $d\le \log^a n$ and $\log(1/\epsilon)=\Theta(\log n)$,
we have
\[
\frac{\log\log(1/\epsilon)}{\log d}
\ge
\frac1a-o(1),
\]
thus the space usage is
$\tilde O\!\left(
(1/\epsilon)^{1-\frac{1}{a}+o(1)}
\right)$ as claimed.

\begin{remark}
    As shown in \cref{lem:map} and \cref{lem:coreset}, the update time of this algorithm is at most quadratic in the storage space. 
\end{remark}
\subsection{Lower Bound}\label{sec:lb-high-temp}

We will use a reduction from the communication problem of $\texttt{INDEX}_{2n,n}$ to prove a space lower bound on the Streaming KDE problem (\cref{def:streaming-kde}). 

\begin{lemma}\label{lem:reduction}
Let $S_{n, d, r, \eps}$ denote a deterministic upper bound on the space, measured in bits, of any (randomized) data structure for the Streaming KDE problem with parameters $n, d, r, \eps$ and a failure probability $1/100$.
Assume that $r = \log^{o(1)}n$, $d = \Theta(\log^a n)$ for some constant $a > 2$, and $n^{-100} < \eps  < n^{-c}$ for an arbitrarily small constant $c > 0$. 

Then, there always exists some $m$ such that $\min\{(1/\eps)^{1-\frac{1}{a}-o(1)}, n\} \leq m\leq n$, and that there is a communication protocol with a message size $S_{n, d, r, \eps} + o(m)$ which solves the $\texttt{INDEX}_{2m, m}$ problem with probability at least $9/10$.
\end{lemma}

% Using the public random coins, Alice and Bob will decide on a point set $P$ of $2n$.
% Consider Alice's input $\bx \in \{0,1\}^{2n}$ with $\|\bx\|_0 = n$. Let $K_\bx \subset P$ be the set of points indexed by the ones in $\bx$.
% Let $S_{n, d, r, \eps}$ be a deterministic upper bound on the space, measured in bits, of any (randomized) data structure for the Streaming KDE problem with error parameter $\eps$ and failure probability $1/100$.
% The core lemma of this section is the reduction from \texttt{INDEX}.
% \begin{lemma}\label{lem:reduction}
% \justin{add bounds on $d, \eps, r, n$}
% Consider a point set $P$ such that each $k \in P$ is drawn i.i.d.\ from $\mathcal{N}\p{0, \frac{r'^2}{d} I_d}$ generated with public coins where $r' = r/2$.
% Let $\bx$ be Alice's input from the $\texttt{INDEX}_{2n, n}$ problem.
% Then, there exists a communication protocol with a message size $S_{n, d, r, \eps} + o(n)$ which solves the $\texttt{INDEX}_{2n, n}$ problem with probability at least $9/10$.
% \end{lemma}
% Concretely, that is when there exists some integer $t$ such that the storage space of \cref{lem:psi} simplifies to $o(n)$  and the premise of  \cref{lem:gap} holds. We will discuss the parameter setting at the end of this section.
 
This lemma, combined with the bound $R^{\text{pub}, \to}(\texttt{INDEX}_{2m, m}) = \Omega(m)$ from \cref{thm:indexing} implies the lower bound of \cref{th:const_r_ub_lb}. 
% as it must be the case that $S_{n, d, r, \eps} = \Omega(n)$.
The rest of the section is dedicated to proving \cref{lem:reduction}.

\paragraph{The reduction.}
Assume that we have chosen an appropriate $m$, which we will describe later. Let $r' := r/2$. Using the public randomness, Alice and Bob agree on a point set $K = \{k_1, k_2, \cdots, k_{2m}\} \subset \R^d$ such that $k_i \sim \mathcal{N}\p{0, \frac{r'^2}{d} I_d}$ i.i.d. for each $k_i \in K$.
Given her input string $\bx \in \{0,1\}^{2m}$ from the $\texttt{INDEX}_{2m, m}$ problem, Alice defines a corresponding dataset $K_\bx = \{ k_i \in K\mid \bx_i = 1\}$. Then she constructs her message to Bob, which contains two components:
for the first component, Alice builds a streaming KDE data structure  with parameters $n, d, r, \eps$. This takes $S_{n, d, r, \eps}$ bits. Note that for $m \leq n$, such a KDE data structure can indeed solve the instance with the dataset $K_\bx$. Also, note that the use of the data structure requires that the input points are promised to have norm at most $r$. On the other hand, the point set $K$ contains vectors which are sampled from $\mathcal{N}\p{0, \frac{r'^2}{d} I_d}$ with unbounded worst-case norm. We make the following standard argument:

\begin{claim}\label{lem:shell}
     With $.99$ probability  over the random choice of $K$, all vectors $k \in K$ satisfy
    $\|k\|_2 \leq r$.
    Similarly, for any single $v \sim  \mathcal{N}\p{0, \frac{r'^2}{d} I_d}$, with $.99$ probability, $\|v\|_2 = (1\pm O(\frac{1}{\sqrt{d}}))r'$
\end{claim}
\begin{proof}
We write each $k_i \in K$ as $k_i = \frac{r'}{\sqrt{d}}g_i$ for i.i.d. $g_i  \sim \mathcal{N}(0,1)$. Then applying  \cref{prop:norm_gauss}  and a union bound give
$\|g_i\|_2 = \sqrt{d} \pm \Theta(\sqrt{\log m}) = \sqrt{d} \pm O(\sqrt{\log n})$ simultaneously for all $g_i$ with $.99$ chance. On this event,  $\|k_i\|_2 \leq \frac{r'}{\sqrt{d}}( \sqrt{d} \pm O(\sqrt{\log n})) \leq r$ for  $d \gg \log n$ and sufficiently large $n$.

Similarly, applying \cref{prop:norm_gauss}  on scaled $v$ gives $\|v\|_2 \leq \frac{r'}{\sqrt{d}}( \sqrt{d} \pm \Theta(1))$, which is $ (1 \pm O(\frac{1}{\sqrt{d}}))r'$ as claimed.
\end{proof}

% So the failure probability of the first component is at most $.01$ from \cref{lem:shell} plus $.1$ from the KDE data structure.

For the second component of Alice's message,  similar to the upper bound proof, Alice defines a feature map associated with the degree-$t$ truncation of the exponential kernel and sends the aggregated sketch of $K_\bx$. However, instead of truncating the Taylor expansion as in \cref{lem:map}, here we will truncate the Hermite expansion of the exponential kernel. Concretely, recall from \cref{cor:hermite:exponent_expansion} that the exponential function can be written as 
\[e^x = \sum_{l=0}^\infty c_l \varphi_{l, \lambda}(x),\] where $c_{l, \lambda} := e^{\lambda^2/2} \frac{\lambda^l}{\sqrt{l!}}$ and $\ph_{l, \lambda}$ is the rescaled $l$-th Hermite polynomial with rescaling factor $\lambda$.
We will set $\lambda = \frac{r'^2}{\sqrt{d}}$.
We denote
\[T_{\leq t}(x) := \sum_{l=0}^t c_{l, \lambda} \ph_{l, \lambda}(x),\]
and
\[
    T_{> t}(x) := \sum_{l=t+1}^\infty c_{l, \lambda} \ph_{l, \lambda}(x).
\]
The goal of the second component is to allow Bob to approximate
$\sum_{k \in K_\bx} T_{\leq t} (\langle k, q\rangle)$ at his query $q$. We formalize this in the following lemma:

\begin{lemma}\label{lem:psi}
    For any integer $t \geq 0$, there exist embeddings   $\psi_A, \psi_B$ such that 
    \[
    \abs{\left \langle \sum_{k \in K_\bx} \psi_A(k), \psi_B(q) \right \rangle - \sum_{k \in K_\bx} T_{\leq t} (\langle k, q\rangle)} \leq \eps m,
    \]
    and that $\sum_{k \in K_\bx} \psi_A(k)$ can be stored in $O\p{\frac{(d+t)^t}{t!} \p{t\log(d+t) + \log(n/\eps)}}$ bits of space.
\end{lemma}

\begin{proof}
Using \Cref{lem:embed_memory}, it suffices to give polynomial bounds on the coefficients $\{a_l\}_{l=0}^t$ of $T_{\leq t}$.
The coefficients of $T_{\leq t}$ have magnitude upper bounded by $t\cdot\max_{l}^t |c_{l, \lambda}|$ times the maximum coefficient norm of $\ph_{l, \lambda}$. These coefficients are given in \Cref{prop:hermite:coefficients}, therefore this expands to:
\begin{align*}
    \max_{l = 0}^t |a_l|
    &= t\max_{l=0}^t \abs{e^{\lambda^2/2} \frac{\lambda^l}{\sqrt{l!}} \p{\sqrt{l!} \sum_{h=0}^{\floor{l/2}} \frac{(-1)^h}{h!2^h}\frac{(1/\lambda)^{l-2h}}{(l-2h)!}}} \\
    &= t\max_{l=0}^t \abs{e^{\lambda^2/2} \p{ \sum_{h=0}^{\floor{l/2}} \frac{(-1)^h}{h!2^h}\frac{\lambda^{2h}}{(l-2h)!}}}.
\end{align*}
Recall that we chose $\lambda = r'^2/\sqrt{d}$. Since $d = \Omega(\log^a n), a > 2$ and $r' = (\log n)^{o(1)}$, we have $\lambda = o(1) < 1$ for all sufficiently large $n$.
Therefore,
\[
    \max_{l = 0}^t |a_l|
    \leq O(t) \cdot \max_{l = 0}^t \abs{\sum_{h=0}^{\floor{l/2}} \frac{(-1)^h}{h!2^h}\frac{1}{(l-2h)!}}
    \leq O(t).
\]
% Recall that we have conditioned on $\|p\|_2 \leq r$ for all $k \in P$.
% Therefore, each coordinate of every point $k \in P$ has magnitude upper bounded by $r = O(1)$.
% Then,
% \[
%     \max_{k=0}^t \|p\|_\infty^k \leq r^t.
% \]
% The same holds for $q$: the maximum value of any coordinate of $\psi_B(q)$ is $r^t$.
% Let $\psi_A(p)$ denote the truncation of each coordinate of $\psi_A(p)$ to $b$ bits of precision for some parameter $b > 0$.
% We will upper bound the precision error by summing over all coordinates of the embedding.
% For any $\gamma > 0$, with $b = O\p{\log(\p{\frac{(d+t)^t}{t!}}r^{2t}/\gamma)} = O\p{t\log(d+t) + \log(1/\gamma)}$ bits of precision, 
% \[
% \abs{\left\langle \psi_A(p), \psi_B(q) \right\rangle - \left\langle \psi_A(p), \psi_B(q) \right\rangle} \leq \gamma.
% \]
% Therefore,
% \[
% \abs{\left \langle \sum_{k \in K_\bx} \psi_A(p), \psi_B(q) \right \rangle - \langle \sum_{k \in K_\bx} \psi_A(p), \psi_B(q) \right \rangle}
% = \abs{\left \langle \sum_{k \in K_\bx} \psi_A(p), \psi_B(q) \right \rangle - \sum_{k \in K_\bx} T_{\leq t} (\langle k, q\rangle)}
% \leq \gamma n.
% \]
% Setting $\gamma = \eps/n$ and $b = O(t\log(d+t) + \log(n/\eps))$ completes the proof.
\end{proof}

After receiving the  KDE data structure of $K_\bx$ and $\sum_{k \in K_\bx} \psi_A(k)$, Bob queries the data structure with $q$ and computes $ \langle \sum_{k \in K_\bx} \psi_A(k), \psi_B(q)  \rangle$, where $q = k_j\cdot\frac{r'}{\|k_j\|}$ for $j\in [2n]$ being his input in the $\texttt{INDEX}$  problem.

With $.99$ chance, the KDE data structure outputs $\widehat{\expker}(K_\bx , q)$ with 
\[ |\widehat{\expker}(K_\bx , q) - \expker(K_x, q) | \leq \eps \expker(K_\bx , q) \leq \Theta(\eps m).\]
In the last inequality we used $\exp(k, q) \leq e^{r^2} = n^{o(1)} = \eps^{o(1)}$ and rescaled $\eps \to \eps^{1 - o(1)}$ to unclutter notation. This doesn't affect the final bound since $o(1)$ is already in the exponent.
Equivalently,
\[ \left |\widehat{\expker}(K_\bx , q) - \sum_{k \in K_\bx}\left(T_{\leq t}(\langle k, q\rangle) + T_{> t}(\langle k, q\rangle)\right)\right | \leq\Theta(\eps m), \]
thus
\begin{equation*}
    \left| \left(\widehat{\expker}(K_\bx , q) -\sum_{k \in K_\bx} T_{\leq t} (\langle k, q\rangle) \right) - \sum_{k \in K_\bx}T_{> t}(\langle k, q\rangle) \right | \leq \Theta(\eps m).
\end{equation*} 
Since $\langle \sum_{k \in K_\bx} \psi_A(k), \psi_B(q)  \rangle$ approximates $\sum_{k \in K_\bx} T_{\leq t} (\langle k, q\rangle)$ up to $\eps m$ error,
\[
    S := \widehat{\expker}(K_\bx , q) - \left\langle \sum_{k \in K_\bx} \psi_A(k), \psi_B(q)  \right\rangle
\]
approximates $\sum_{k \in K_\bx}T_{> t}(\langle k, q\rangle)$ up to $\Theta(\eps m)$. We argue that Bob can decide whether $x_j = 1$ (i.e., whether $k_j \in K_\bx$) by thresholding $S$. This follows from the next lemma:
\begin{lemma}\label{lem:gap}
    Let $q_1 = k_{j_1}\cdot\frac{r'}{\|k_{j_1}\|}$, for some $k_{j_1} \in K_\bx$ and let $q_2 = k_{j_2}\cdot\frac{r'}{\|k_{j_2}\|}$ for some $k_{j_2} \in (K \setminus K_\bx)$.  Given $t$ such that 
    \begin{itemize}
        \item $t= o(\sqrt{d})$,
        \item  $m\frac{(t+1)!}{d^{t+1}} = o(1)$, and
        \item $\eps m = o( \frac{r'^{2t+2}}{(t+1)!})$,
    \end{itemize}  with $.96$ probability, 
    \[\sum_{k \in K_\bx} T_{> t}(\langle k, q_1\rangle)- \sum_{k \in K_\bx} T_{> t}(\langle k, q_2\rangle) = \omega(\eps m).\]
\end{lemma}
% \todo{Should we use Gaussian instead so we can make use of Hermite polynomials?}
% \bor{We should use Gaussian with $r' = r / 2$. Then we do the same proof as before but for $r'$ instead of $r$, and the only thing that we didn't check last time that the dataset actually lies in the ball with radius $r$. But now we have it with high probability.}
% \todo{I'm a bit worried that for Gaussians, at the bottom of page 6 we only say $\|p\|_2 \leq r$, which means the LHS of claim 3 needs to argue for some quantity $T_{>t}(<r^2)$}

Observe that $q_1$ corresponds to the case  $x_j = 1$ and $q_2$ corresponds to $x_j = 0$. Therefore, \cref{lem:gap} implies that in these two cases, $S$ has an $\omega(\eps m)$ gap after tolerating $\Theta(\eps m)$ approximation error, thus Bob can distinguish the two cases. 

\begin{proof}[Proof (of \cref{lem:gap})]
It suffices to show that with high probability,
\[T_{>t}(\langle k_{j_1}, q_1 \rangle) = \omega(\eps m) +  \sum_{k \in K_\bx} T_{> t}(\langle k, q_2\rangle) - \sum_{k \in (K_\bx\setminus \{k_{j_1}\})} T_{> t}(\langle k, q_1\rangle) .\]

We proceed by lower bounding $T_{>t}(\langle k_{j_1}, q_1 \rangle)$ (\cref{clm:lower_bound}) and upper bounding ~$ \sum_{k \in K_\bx} T_{> t}(\langle k, q_2\rangle) - \sum_{k \in (K_\bx\setminus \{q_1\})} T_{> t}(\langle k, q_1\rangle) $ (\cref{clm:upper_bound}).
\begin{claim}\label{clm:lower_bound}
    In the success case of \cref{lem:shell}, $T_{>t}(\langle k_{j_1}, q_1 \rangle) = \Omega\left(\frac{r'^{2t+2}}{(t+1)!}\right)$.
\end{claim}
\begin{proof}
Since $q_1$ is a scaled version of $k_{j_1}$, $\langle k_{j_1}, q_1 \rangle = \|k_{j_1}\|_2 \|q_1\|_2$.
Also, we note that $\|q_1\|_2 = \|k_{j_1}\|_2 \frac{r'}{\|k_{j_1}\|_2} = r'$ by definition.
Therefore, by \cref{lem:shell} applied to $\|k_{j_1}\|_2$, $\langle k_{j_1}, q_1 \rangle = \left(1 \pm O(1/\sqrt{d})\right) r'^2$.

Then using the Hermite polynomial expansion (\cref{cor:hermite:exponent_expansion}), we can expand as follows:
\begin{align*}
    T_{>t}(\langle k_{j_1}, q_1 \rangle) &= \sum\limits_{l = t + 1}^{\infty}c_l\ph_{l, \lambda}(\langle k_{j_1}, q_1 \rangle) \\
    &= e^{\lambda^2/2}\sum\limits_{l = t + 1}^{\infty}\frac{\lambda^l}{l!}\text{He}_l(\langle k_{j_1}, q_1 \rangle/\lambda) \\
    &= e^{\lambda^2/2} \left (\underbrace{\frac{\lambda^{t+1}}{(t+1)!}\text{He}_{t+1}(\langle k_{j_1}, q_1 \rangle/\lambda)}_{=: \alpha} + \underbrace{\sum\limits_{l = t + 2}^{\infty}\frac{\lambda^l}{l!}\text{He}_l(\langle k_{j_1}, q_1 \rangle/\lambda)}_{=: \beta} \right)
\end{align*}
We will show that $\alpha = \Omega\left(\frac{r'^{2t+2}}{(t+1)!}\right)$ and $\beta  \geq 0$. Then the conclusion follows immediately. 
Since we set $\lambda = \frac{r'^2}{\sqrt{d}}$ and $t = o(\sqrt{d})$, in the success case of \cref{lem:shell}, $\langle k_{j_1}, q_1 \rangle/\lambda \geq \left(1- O\left(\frac{1}{\sqrt{d}}\right)\right) \sqrt{d} = \omega(t)$.
\cref{cor:hermite:asymptotic} gives 
\begin{align*}
\text{He}_{t+1}(\langle k_{j_1}, q_1 \rangle/\lambda) &=  (1 \pm o(1) )(\langle k_{j_1}, q_1 \rangle/\lambda)^{t+1},
\end{align*}
so
\[\alpha = \Omega \left(\frac{\lambda^{t+1}}{(t+1)!} (\langle k_{j_1}, q_1 \rangle/\lambda)^{t+1} \right) =  \Omega\left(\frac{\left(1-O\left(\frac{1}{\sqrt{d}}\right)\right)^{2t+2}r'^{2t+2}}{(t+1)!}\right) =  \Omega\left(\frac{r'^{2t+2}}{(t+1)!}\right).\]
The last inequality follows from $\left(1-O\left(\frac{1}{\sqrt{d}}\right)\right)^{2t+2} \geq e^{-O\left(\frac{1}{\sqrt{d}}\right) (2t+2)} =e^{-o(1)}$.

It remains to show that $\beta \geq 0$. Let $f(s) :=e^{(\langle k_{j_1}, q_1 \rangle/\lambda)s  - s^2/2}$, then  the generating function (\cref{prop:hermite:exponent_expansion}) implies
\[\beta = \sum\limits_{l = t + 2}^{\infty}\frac{\lambda^l}{l!}\text{He}_l(\langle k_{j_1}, q_1 \rangle/\lambda) = f(\lambda) - \sum\limits_{l = 0}^{t+1}\frac{\lambda^l}{l!}\text{He}_l(\langle k_{j_1}, q_1 \rangle/\lambda);\]
applying \cref{cor:hermite:derivative} at $s = 0$, this is further  $\beta = f(\lambda) - \sum\limits_{l = 0}^{t+1}\frac{\lambda^l\cdot f^{(l)}(0)}{l!}.$  On the other hand, by Taylor's theorem (with Lagrange remainder),
\[f(\lambda)  = \sum\limits_{l = 0}^{t+1}\frac{\lambda^l\cdot f^{(l)}(0)}{l!} + \frac{\lambda^{t+2}\cdot f^{(t+2)}(\xi)}{(t+2)!}\]
for some $\xi \in (0, \lambda)$. Together, these give 
$\beta = \frac{\lambda^{t+2}\cdot f^{(t+2)}(\xi)}{(t+2)!}$. At this point, we can apply \cref{cor:hermite:derivative}  again at $s = \xi$ and get

\[\beta = \frac{\lambda^{t+2}\cdot \text{He}_{t+2}(\langle k_{j_1}, q_1 \rangle/\lambda - \xi)f(\xi)}{(t+2)!}.\]

To reach our desired conclusion that $\beta \geq 0$, it suffices to argue that $\text{He}_{t+2}(\langle k_{j_1}, q_1 \rangle/\lambda - \xi) \geq 0$. Indeed, this follows immediately from $\langle k_{j_1}, q_1 \rangle/\lambda - \xi = \Theta(\sqrt{d}) - o(1) = \omega(t)$ and an application of \cref{cor:hermite:asymptotic}.

\end{proof}

\begin{claim}\label{clm:upper_bound}
    With $.96$ probability, $ \sum_{k \in K_\bx} T_{> t}(\langle k, q_2\rangle) - \sum_{k \in (K_\bx\setminus \{k_{j_1}\})} T_{> t}(\langle k, q_1\rangle) = o\left(\frac{r'^{2t+2}}{(t+1)!}\right)$. 
\end{claim}
\begin{proof}
Fixing $q_1$ and over the randomness of any $k \neq k_{j_1}$, $\langle k, q_1\rangle \sim \mathcal{N}(0, \lambda^2)$ for $\lambda = \frac{r'^2}{\sqrt{d}}$.
Note that $\lambda = o(1)$ as we are in the parameter regime where $d$ is super-polynomial in $r$.
Let $\mu=\mathcal{N}(0,\lambda^2)$. Since $T_{>t}$ contains no degree-zero Hermite component, by the orthonormality of the Hermite basis $\{\ph_{l,\lambda}\}_{l=0}^{\infty}$ in $L^2(\mu)$ (\cref{prop:hermite:basis}), we have that  $\E[T_{>t}(\langle k,q_1\rangle)] =  \langle T_{>t}, 1\rangle_{L^2(\mu)} = 0$.
Therefore,
\[
\Var{T_{>t}(\langle k,q_1\rangle)}
=
\E\left[T_{>t}(\langle k,q_1\rangle)^2\right]
=
\|T_{>t}\|^2_{L^2(\mu)}.
\]
By the same orthonormality, this equals the sum of squares of the Hermite coefficients, i.e.
\[
\Var{T_{>t}(\langle k,q_1\rangle)}
=
\sum_{l=t+1}^\infty
\left(e^{\lambda^2/2}\frac{\lambda^l}{\sqrt{l!}}\right)^2,
\]
which is $O\left(\frac{\lambda^{2t+2}}{(t+1)!}\right)$ as $\lambda^2=o(1)$.  Then by Chebyshev's inequality, with $.99$ chance,
\[ \left |\sum_{k \in (K_\bx\setminus \{k_{j_1}\})} T_{> t}(\langle k, q_1\rangle) \right |  = O\left(\sqrt{m \cdot \frac{\lambda^{2t+2}}{(t+1)!}}\right).\]

The exact same argument applies for $q_2$ and all keys $k \in K_\bx$.
% Finally, in the success case of \cref{lem:shell}, $\|q_1\|_2, \|q_2\|_2 \leq (1+O(\frac{1}{\sqrt{d}})) r'$, thus 
% \[
%     \sum_{k \in K_\bx} T_{> t}(\langle k, q_2\rangle) - \sum_{k \in (K_\bx\setminus \{q_1\})} T_{> t}(\langle k, q_1\rangle) \leq O\left(\sqrt{n} \cdot \frac{(1+O(\frac{1}{\sqrt{d}}))^{t+1} r'^{2t+2}}{\sqrt{d^{t+1}(t+1)!}}\right)
% \]
Therefore,
\[
    \sum_{k \in K_\bx} T_{> t}(\langle k, q_2\rangle) - \sum_{k \in (K_\bx\setminus \{k_{j_1}\})} T_{> t}(\langle k, q_1\rangle) \leq O\left(\frac{\sqrt{m} \cdot r'^{2t+2}}{\sqrt{d^{t+1}(t+1)!}}\right)
\]
which is $o\left(\frac{r'^{2t+2}}{(t+1)!}\right)$ as long as $m\frac{(t+1)!}{d^{t+1}} = o(1)$, which holds by the parameter assumptions in \cref{lem:gap}.
\end{proof}
Concluding \cref{clm:lower_bound}, \cref{clm:upper_bound}, and using the last parameter assumption of \cref{lem:gap}, we see that
\[T_{>t}(\langle k_{j_1}, q_1 \rangle) = \Omega(\frac{r'^{2t+2}}{(t+1)!}) = \omega(\eps m) +  \sum_{k \in K_\bx} T_{> t}(\langle k, q_2\rangle) - \sum_{k \in (K_\bx\setminus \{k_{j_1}\})} T_{> t}(\langle k, q_1\rangle)\]
as desired. This completes the proof of \cref{lem:gap}.
\end{proof}

With the above ingredients, we are ready to prove our main lemma of this section.

\begin{proof}[Proof of \cref{lem:reduction}]

The total failure probability of the reduction described above is the sum of $.01$ from the KDE data structure, $.02$ from \cref{lem:shell}, and $.02$ from Chebyshev's inequalities. 

Wrapping up the reduction, we next show how to choose $m$ with $\min\{(1/\eps)^{1-\frac{1}{a} - o(1)}, n\} \leq m\leq n$ and an appropriate truncation threshold $t$ that (1) satisfies the premise of \cref{lem:gap} and (2) gives $o(m)$ storage space for \cref{lem:psi}.

We describe our choice of $m$ and $t$ in separate cases, depending on how $n$ compares with $(1/\eps)^{1-\frac{1}{a}}$. In the first case, when $n \geq (1/\eps)^{1-\frac{1}{a} + o(1)}$, we choose $t = \floor{\frac{\log (1/\eps) - 3\log \log(1/\eps)}{\log (d/r'^2)}}$. In the setting of \cref{th:const_r_ub_lb}, $1/\eps \ll n^{100}$, $r' = (\log n)^{o(1)}$, and $d \geq \log^a n$ for $a > 2$, so such $t$ satisfies $t = o(\sqrt{d})$. It remains to show that for such $t$, there is a nonempty interval of admissible values of $m$ that satisfies the premise of \cref{lem:gap} and gives $o(m)$ storage space for \cref{lem:psi}. Equivalently,  there exists $m$ such that
\begin{equation}\label{eq:interval}
\p{\frac{(d+t)^t}{t!} \p{t\log(d+t) + \log(n/\eps)}} \ll m \ll \min\left\{ \frac{d^{t+1}}{(t+1)!}, \frac{1}{\eps} \frac{r'^{2t+2}}{(t+1)!} \right\}.
\end{equation} 
The lower bound side is at most $\Theta(\frac{(d+t)^t}{t!} \log n) = \Theta(\frac{d^t}{t!} \log n)$. This is asymptotically smaller than $\frac{d^{t+1}}{(t+1)!}$ because 
\[\frac{d^{t+1}/(t+1)!}{d^{t}/t!} = \frac{d}{t+1} = \omega(\sqrt{d}) = \omega(\log n).\] Also, this is asymptotically smaller than $\frac{1}{\eps} \frac{r'^{2t+2}}{(t+1)!} $  because 
\[\frac{1}{\eps} \frac{r'^{2t+2}}{(t+1)!}  = \frac{1}{\eps} \left(\frac{r'^2}{d}\right)^{t} \frac{r'^2}{t+1}\cdot \frac{d^{t}}{t!}.\]
By our choice of $t$, $\frac{1}{\eps} (\frac{r'^2}{d})^{t}  = \Omega(\log^3(1/\eps))$, which, multiplied by $\frac{r'^2}{t+1}$ gives $\Omega(\log^2 (1/\eps)) \gg \log n$. Summarizing these, we see that the lower and upper bounds in \cref{eq:interval} indeed induce a non-empty interval. Moreover,  in this interval, $m \gg \frac{d^t}{t!} \log n = (1/\eps)^{1- \frac{\log\log(1/\eps)}{\log d} + o(1)}$ by Stirling (see similar computation in \Cref{eq:small_r_runtime}), which for $d \geq \log^a  n$ becomes $(1/\eps)^{1 - \frac{1}{a} + o(1)}$. Therefore, we successfully find $m$ that satisfy all of our requirements.

In the second case, when $n \leq (1/\eps)^{1-\frac{1}{a} - o(1)}$, we automatically need to set $m = n$. It remains to choose $t = o(\sqrt{d})$ such that $m = n$ satisfies \cref{eq:interval}. As shown in the previous case, the lower bound side of \cref{eq:interval} is at most $\Theta(\frac{d^t}{t!} \log n)$; also $\frac{d^{t+1} / (t+1)!}{d^t/t!} = \omega(\log n)$. Therefore, there is a non-empty interval between the lower bound and the first term of the upper bound. We choose $t$ such that $n$ lies in this interval, and is arbitrarily close to the upper bound. That is, we choose 
\[t = \left(\frac{a}{a-1} + o(1)\right)\frac{\log n}{\log d} =o(\sqrt{d}).\]
Then $ \frac{d^{t}}{t!}\log n \ll n \ll \frac{d^{t+1}}{(t+1)!}$, and the $o(1)$ term in $t$ can be set arbitrarily small at the cost of pushing $n$ closer to this upper bound. Now we look at the ratio between the two upper bound terms in \cref{eq:interval}, namely,
\[
\frac{
\frac1\epsilon\cdot \frac{r'^{2t+2}}{(t+1)!}
}{
\frac{d^{t+1}}{(t+1)!}
}
=
\frac1\epsilon\left(\frac{r'^2}{d}\right)^{t+1}.
\]
This is $\omega(1)$ as long as $\log (1/\eps) \geq \left(\frac{a}{a-1} + o(1)\right) \log n + \omega(1)$, equivalently, $n \leq  (1/\eps)^{1-\frac{1}{a} - o(1)}$, which is exactly our case assumption. In summary, all premises of \cref{lem:gap} hold, and the side information from \cref{lem:psi} has size $o(m)$. 

Finally, when $
(1/\eps)^{1-\frac1a-o(1)}
\le n \le
(1/\eps)^{1-\frac1a+o(1)}$, we choose a slightly larger error parameter $\eps' \ge \eps$ with $\eps'=\eps^{1-o(1)}$, so that $n \geq
(1/\eps)^{1-\frac1a+o(1)}$. In this way, we are back in the first case, and now we are working with a slightly larger error parameter. Since any data structure that gives $\eps$-relative error also gives $\eps'$-relative error, the lower bound for $\eps'$ also applies to the original data structure. 
This concludes the proof to all cases of \cref{lem:reduction}.

\end{proof}

\subsection{Putting it together}
\label{sec:high_tem_putting_together}

\begin{proof}[Proof of \Cref{th:const_r_ub_lb}]
    The upper bound is shown in \Cref{sec:ub-high-temp} using polynomial embedding and a coreset maintained in a stream.
    
    The lower bound can be found in \Cref{sec:lb-high-temp}. As shown by \Cref{lem:reduction}, for any $n$, there is $m$ such that $\min\{(1/\eps)^{1-\frac{1}{a} - o(1)}, n\} \leq m\leq n$, and there is a $S_{n, d, r, \eps} + o(m)$ size protocol that solves $\texttt{INDEX}_{2m, m}$, where $S_{n, d, r, \eps}$ denotes the space complexity of any algorithm solving Softmax KDE problem. Combined with the lower bound $R^{\text{pub}, \to}(\texttt{INDEX}_{2m, m}) = \Omega(m)$ from \cref{thm:indexing}, we get that $S_{n, d, r, \eps} = \Omega(m) = \Omega(\min\{(1/\eps)^{1-\frac{1}{a} - o(1)}, n\})$,  which is the Softmax KDE lower bound presented in  \cref{th:const_r_ub_lb}.
\end{proof}

\section{Proof of \Cref{th:big_r_ub_lb}}
In this section we prove \Cref{th:big_r_ub_lb}, restated here for convenience of the reader:
\begin{theorem*}[Restated from \Cref{th:big_r_ub_lb}]
    Assume that $r^2 = \Omega(\log n)$, i.e. $e^{r^2} = n^{\Omega(1)}$. Then the Softmax KDE problem \eqref{eq:streaming-kde} can be solved using
    \begin{align*}
    \tilde{O}\qty(\frac{e^{r^2(1 + o(1))}}{\eps})
    \end{align*}
    memory. Furthermore, this bound is almost tight for any $d = (\log n)^{1 + \Omega(1)}$, i.e. there is a
    \begin{align*}
        \tilde{\Omega}\qty(\min\qty{n, \frac{e^{r^2(1 - o(1))}}{\eps}})
    \end{align*}
    lower bound for streaming space complexity.
\end{theorem*}

We prove the upper bound in \Cref{sec:ub_low_temp} below, the lower bound in \Cref{sec:lb_low_temp} below. These two results are put together in \Cref{sec:low_temp_putting_together} to obtain the proof of \Cref{th:big_r_ub_lb}.
\subsection{Upper Bound}
\label{sec:ub_low_temp}
\newcommand{\Certify}{\textsc{Certify}}
\newcommand{\Partition}{\textsc{Partition}}
\newcommand{\cone}{\textsc{cone}}
Our starting point is the algorithm BalanceKV (\cite{kochetkova2025streaming}), an online discrepancy minimization method based on~\cite{ALS21} that maintains a coreset approximating the dataset in all directions simultaneously. It relies on the offline factor-2 compression procedure \BaseCompress{}, which we formalise in \Cref{lem:simple_coreset}:

\medskip

\noindent {\em {\bf Lemma} (Restated from Lemma~\ref{lem:simple_coreset})} Let $\gamma: \R^d \times \R^d \to \R$  be a positive definite kernel, and let $U \times V\subset   \R^d \times \R^d $ be a predefined subset of the domain. There is a randomized algorithm $\BaseCompress$ which receives $\gamma$ and $K \subset U$ as inputs. $\BaseCompress(\gamma, K)$ outputs $\widetilde{K} \subset K$ such that $|\widetilde{K}| \leq |K|/2$ and for any $q \in V$, with probability $1 - \delta$,
           \[\left|\gamma(K, q) - 2\gamma(\widetilde{K}, q)\right| \leq O\left(\sqrt{\max_{u \in U}\gamma(u, u)\cdot \max_{v \in V}\gamma(v, v)}\cdot \log (|K|/\delta)\right).\]
    $\BaseCompress(\gamma, K)$ runs in $O(|K|^2 \cdot T)$ time, where $T$ is the maximum time to access $\gamma(k,q)$.
    
\medskip

While this mechanism captures coarse geometric structure, it is largely oblivious to finer structure: it treats all directions as equally adversarial, even when the data is highly organized. In particular, for inputs consisting of tight clusters, it is strictly more effective to run separate instances of \BaseCompress{} on each cluster rather than on the dataset as a whole.

We address this limitation via a \textit{pseudo-randomization procedure}~-- an approach pioneered in work on data-dependent data structures \cite{AR15, ALRW17, CKNS20}.  Pseudo-randomization allows us to view the dataset as union of two types of components: pseudo-random subsets and low radius subsets. On pseudo-random subsets, the original balancing procedure is near-optimal, as no direction is preferred. For low radius subsets, however, we exploit local structure by recentering and treating each region as contained in a smaller ball, allowing \BaseCompress{} to operate at the appropriate scale.

In \Cref{subsub:partition} we present the algorithm \Partition{} (\Cref{alg:partition}), a procedure that recursively isolates dense regions of a dataset and stops once the dataset is either pseudo-random or contained in a much smaller Euclidean ball. Its guarantees are given in \Cref{lem:ps_randomification}. Next, in \Cref{subsub:recenter_rescale} we explain how to recenter and rescale datasets so that \BaseCompress{} yields the lowest-error coresets, and in \Cref{subsub:leaves} we formally apply \BaseCompress{} to the partition returned by \Cref{alg:partition}.

In \Cref{subsub:compress} we then present \Compress{} (\Cref{alg:fast_2_compress}), our offline constant-factor compression algorithm, which first decomposes the dataset using \Partition{} and then applies \BaseCompress{} to the resulting components. Finally, just as \BaseCompress{} was extended to streaming via Merge-and-Reduce in \Cref{lem:merge_reduce}, in \Cref{subsub:merge_reduce_w} we show how to implement \Compress{} in streaming using a \textit{weighted} version of Merge-and-Reduce.

\subsubsection{Partition into pseudo-random and low radius components}\label{subsub:partition}
\begin{definition}[Pseudo-random dataset]\label{def:pseudorandom}
    For a ball $B = B(c, r) \subset \R^d$, a dataset $K \subset B$ and two positive parameters $\Delta = \Delta(d, n), \tau = \tau(d, n)$, we call $K$ $(B, \Delta, \tau)$\emph{-pseudo-random}
    % \footnote{Our definition follows \cite{ALRW17,CKNS20}, but with $K$ in a ball instead of a sphere. This makes our notion somewhat counterintuitive. For instance, a dataset that is highly concentrated around $c$ is pseudo-random. We nevertheless retain the term to emphasize the connection with the earlier literature.}
    if the following holds:
    \begin{gather}\label{eq:sp_caps}
        \forall q \in B \quad |\{k \in K \mid \<k - c, q - c> > \Delta r^2\}| \leq \tau |K|.
    \end{gather}
    When the ball $B$ or parameters $\Delta, \tau$ are clear from the context we omit them.
\end{definition}

We start with an instance of \cite[Lemma 58]{CKNS20} adapted to our definition. We prove it in \Cref{sec:ub_lowt_proofs}.

\begin{restatable}{lemma}{certifyball}[Certify]\label{cor:certify-ball}
There is a randomized procedure that, given a set
\(K \subseteq B(c,r)\) and parameters \(\Delta,\tau,\delta \in \left(0,\tfrac13\right)\), runs in time
\[
O\!\left(\frac{d}{\Delta^{2}\tau}\log\!\left(\frac{2|K|}{\delta}\right)\cdot |K|\right)
\]
and, with probability at least \(1-\delta\), does one of the following:
\begin{enumerate}
    \item returns a point \(k^* \in B(c,r)\) such that
    \[
    \left|B\!\left(k^*, \sqrt{1-\Delta^4}\,r\right)\cap K\right|
    \;\ge\;
    \Omega(\Delta^2)\cdot \tau |K|,
    \]
    \item certifies that \(K\) is \((\Delta',\tau)\)-pseudo-random, where \(\Delta'=\Theta(\Delta)\).
\end{enumerate}
We refer to this procedure as \(\Certify_{\Delta,\tau,\delta}\).
\end{restatable}

 \Certify{} can be viewed as an algorithm that partitions a non-pseudo-random dataset into a dense component and its complement. The algorithm \Partition{} simply applies \Certify{} recursively to each part of the partition, unless that part is already pseudo-random or has low radius. We may therefore view \Partition{} as building a binary tree whose height is bounded: for every non-pseudo-random dataset, \Certify{} returns one subset with a constant-factor fewer points and another contained in a ball of constant-factor smaller radius.

The analysis of the performance of \Partition{} is standard; we provide it in \Cref{sec:ub_lowt_proofs}.
\begin{restatable}{lemma}{partition}[Performance guarantees of \Partition, \Cref{alg:partition}]
    \label{lem:ps_randomification}
    For any $r_0$, $\Delta \leq 1/3$ and any dataset $K \subset B(c, r)$, $|K| \leq n$, \Partition{} produces a decomposition $\mathcal{K}$ of $K$
    \begin{gather}
        % \label{eq:pseudorandom_decomp}
        \mathcal{K} = \bigsqcup_{i=1}^{N} K_i
    \end{gather}
    such that $K_i \subset B_i = B(c_i, r_i)$ and at least one of the two holds:
    \begin{itemize}
        \item $K_i$ is $(B_i, \Delta, \tau)$-pseudo-random.
        \item $r_i < r_0$. 
    \end{itemize}
    The algorithm is successful with probability at least $1 - \frac{1}{n^2}$, has runtime $O\!\left(\frac{d}{\Delta^{4}\tau^2} \cdot |K|\log\!\left(n\right) \cdot \frac{\log (r/r_0)}{\Delta^4} \right)$ and outputs a partition of size at most
    
    \begin{gather*}
         \qty(\frac{\log |K|}{\Delta^2\cdot \tau})^{O\qty(\frac{\log \frac{r}{r_0}}{\Delta^4})}
    \end{gather*}
\end{restatable}

\begin{algorithm}
\caption{\Partition$(K,r,r_0, \Delta, \tau)$}\label{alg:partition}
\begin{algorithmic}[1]
\State \textbf{input:} $K \subseteq B(c,r)$, parameters $r_0,\Delta,\tau$
\State $\mathcal K \gets \emptyset$
\State $\delta \gets \frac{1}{n^4}$
\medskip
\If{$r \le r_0$}\Comment{stop if the remainder is in a small ball}
    \State $\mathcal K \gets \mathcal K \sqcup \{K\}$
    \State \textbf{return} $\mathcal K$
\EndIf
\medskip
\State \textbf{while} $K \neq \emptyset$
\If{\Certify$_{\Delta,\tau,\delta}(K)$ returns a point $c' \in B(c, r)$}\label{line:while_certify} \Comment{a dense region is found}
    \State $L \gets K \cap B(c', \sqrt{1-\Delta^4}r)$
    \State $\mathcal K' \gets \textsc{Partition}(L,\sqrt{1-\Delta^4}r,r_0,\Delta\tau)$\Comment{recurse on the dense region}
    \State $\mathcal K \gets \mathcal K \sqcup \mathcal K'$
    \State $K \gets K \setminus L$
\Else
    \State $\mathcal K \gets \mathcal K \sqcup \{K\}$
    \State \textbf{break}
\EndIf

\medskip
\State \textbf{output:} $\mathcal K$
\end{algorithmic}
\end{algorithm}

\subsubsection{Recenter \& rescale}\label{subsub:recenter_rescale}
\Cref{lem:ps_rand_coreset} and \Cref{lem:small_r_coreset} show how to build coresets for the pseudo-random and low radius parts of the partition $\mathcal{P}$ given by \Cref{lem:ps_randomification}. In both cases, the algorithm is the procedure \BaseCompress{} from \Cref{lem:simple_coreset}, applied to a \emph{transformation} of the dataset of keys.

To see why this helps, compare a dataset $K \subset B(c,r)$ with its recentered version $K-c \subset B(0,r)$. Finding coresets for $K$ and $K-c$ is essentially the same approximation problem, since
\[
\exp(K,q)=\exp(\langle c,q\rangle)\cdot \exp(K-c,q),
\]
so any weighted coreset for one yields the same multiplicative guarantee for the other. Yet, \BaseCompress{} has a better  error bound on the centered dataset. When $K$ is far from the origin, the common offset $c$ creates a large global trend in the kernel values, and \BaseCompress$(K)$ must spend part of its error budget tracking this trend. After recentering, this shared bias disappears, and the algorithm only has to capture the local variation inside the cluster. Similarly, although approximating $\exp(K,q)$ and $\exp(\alpha K,q/\alpha)$ for $\alpha<1$ is equally difficult at the level of the kernel values, since
\[
\exp(\langle k,q\rangle)=\exp(\langle \alpha k,q/\alpha\rangle),
\]
\BaseCompress{} behaves differently on the two instances, as the error does not scale linearly with the diameters of the domains. 

We exploit this observation in \Cref{lem:small_r_coreset}. Any dense part $K'$ of $\mathcal{K}$ given by \Cref{lem:ps_randomification} is contained in a ball $B(c',r')$ much smaller than the query domain $B(0,r)$. If we first shift and rescale it to
\[
\sqrt{\frac{r}{r'}}\,(K'-c'),
\]
so that both the key and query domains become $B(0,\sqrt{rr'})$, and only then apply \BaseCompress, we obtain a coreset whose approximation quality is proportional to
\[
\exp(\langle c',q\rangle)\cdot \exp(rr')
\leq
\exp\bigl(r(r'+\|c'\|_2)\bigr),
\]
whereas applying \BaseCompress{} directly to the original dataset is only guaranteed error
\[
\exp\bigl(r^2/2 + (r'+\|c'\|_2)^2/2\bigr).
\]

For pseudo-random sets in \Cref{lem:ps_rand_coreset}, the same transformation does not by itself improve the guarantee, because by construction the pseudo-random parts need not lie in balls smaller than $B(0,r)$. Instead, the gain comes from a better lower bound on the KDE value: for an arbitrary dataset, the worst case is when the mass concentrates on
$\text{argmin}_{k\in B(0,r)} \exp(\langle k,q\rangle),$
whereas pseudo-randomness rules out such concentration.

\begin{figure}[ht]
  \centering
  \vspace{5mm}
  \resizebox{\linewidth}{!}{%
  \begin{tikzpicture}[
    scale=1,
    transform shape,
    >=Latex,
    line cap=round,
    line join=round
  ]
    % Radii enlarged by factor 1.3
    \def\R{3.38}
    \def\rp{1.69}

    % Left red-ball center
    \def\cx{-1.85}
    \def\cy{0.35}

    % Separation of the two panels
    \def\Sep{7.9}

    % Shift sending c to 0
    \pgfmathsetmacro{\sx}{-\cx}
    \pgfmathsetmacro{\sy}{-\cy}

    % Radius of the black dashed ball: sqrt(r * r/2)
    \pgfmathsetmacro{\rmid}{sqrt(\R*\rp)}

    % Ten points spread through the overlap region
    \def\PointList{
      -2.30/0.80,
      -2.10/0.10,
      -2.25/-0.65,
      -1.75/1.30,
      -1.65/0.65,
      -1.55/-0.10,
      -1.35/-0.75,
      -1.15/0.95,
      -0.95/0.25,
      -0.85/-0.20
    }

    % ---------------- Left picture ----------------
    \begin{scope}[shift={(-\Sep,0)}]
      % white fillings
      \fill[white] (0,0) circle (\R);
      \fill[white] (\cx,\cy) circle (\rp);

      % blue ball
      \draw[blue, thick] (0,0) circle (\R);
      \fill[blue] (0,0) circle (1.8pt);
      \node[blue] at (0.20,-0.20) {$0$};

      % q on the blue boundary
      \fill ({\R*cos(40)},{\R*sin(40)}) circle (1.4pt);
      \node[black] at ({\R*cos(40)+0.28},{\R*sin(40)+0.20}) {$q$};

      % red ball
      \draw[red, thick] (\cx,\cy) circle (\rp);
      \fill[red] (\cx,\cy) circle (1.6pt);
      \node[red] at ({\cx-0.24},{\cy+0.24}) {$c$};

      % overlap points and label
      \foreach \x/\y in \PointList {
        \fill (\x,\y) circle (1.15pt);
      }
      \node[black] at (-0.55,1.72) {$K$};

      % Left callouts
      \node[black, align=left, anchor=west, font=\small] (Lblue) at (4.0,-1.05)
        {ball with radius $r$\\ centered at $0$};
      \draw[->, black, thick] (Lblue.west) -- (2.70,-0.58);

      \node[black, align=left, anchor=west, font=\small] (Lred) at (-7.0,-2.75)
        {ball with radius $r/2$\\ centered at $c$};
      \draw[->, black, thick] (Lred.east) -- ({\cx-0.32},{\cy-1.08});
    \end{scope}

    % Middle arrow and caption
    \draw[->, thick] (-1.35,0) -- (1.35,0);
    \node[black, font=\small] at (0,0.62) {Recenter + rescale};

    % ---------------- Right picture ----------------
    \begin{scope}[shift={(\Sep,0)}]
      % white fillings
      \fill[white] (0,0) circle (\R);
      \fill[white] (0,0) circle (\rp);

      % blue ball
      \draw[blue, thick] (0,0) circle (\R);

      % q on the blue boundary
      \fill ({\R*cos(40)},{\R*sin(40)}) circle (1.4pt);
      \node[black] at ({\R*cos(40)+0.28},{\R*sin(40)+0.20}) {$q$};

      % dashed black ball
      \draw[black, dashed, thick] (0,0) circle (\rmid);

      % red ball
      \draw[red, thick] (0,0) circle (\rp);

      % shifted points and their label
      \foreach \x/\y in \PointList {
        \fill ({\x+\sx},{\y+\sy}) circle (1.15pt);
      }
      \node[black] at (0, 2) {$K-c$};

      % five uniformly spaced contraction arrows: blue -> black
      \foreach \ang in {20,92,164,236,308} {
        \draw[black, dashed, ->, thick]
          ({\R*cos(\ang)},{\R*sin(\ang)}) --
          ({\rmid*cos(\ang)},{\rmid*sin(\ang)});
      }

      % five uniformly spaced expansion arrows: red -> black
      \foreach \ang in {56,128,200,272,344} {
        \draw[black, dashed, ->, thick]
          ({\rp*cos(\ang)},{\rp*sin(\ang)}) --
          ({\rmid*cos(\ang)},{\rmid*sin(\ang)});
      }

      % Right callouts
      \node[black, align=left, anchor=west, font=\small] (Rblue) at (4.00,-1.05)
        {ball with radius $r$\\ centered at $0$};
      \draw[->, black, thick] (Rblue.west) -- (2.85,-0.60);

      \node[black, align=right, anchor=east, font=\small] (Rred) at (-5.00,-2.75)
        {ball with radius $r/2$\\ centered at $0$};
      \draw[->, black, thick] (Rred.east) -- (-1.18,-1.18);

      \node[black, align=center, font=\small] (Rblack) at (0.5,4.60)
        {ball with radius $\sqrt{r\cdot r/2}$\\ centered at $0$};
      \draw[->, black, thick] (Rblack.south) -- (0.1,\rmid);
    \end{scope}
  \end{tikzpicture}%
  }
  \vspace{4mm}
  \caption{Recentering and rescale procedure for low radius datasets}
  % \label{fig:recenter+rescale}
  \vspace{5mm}
\end{figure}
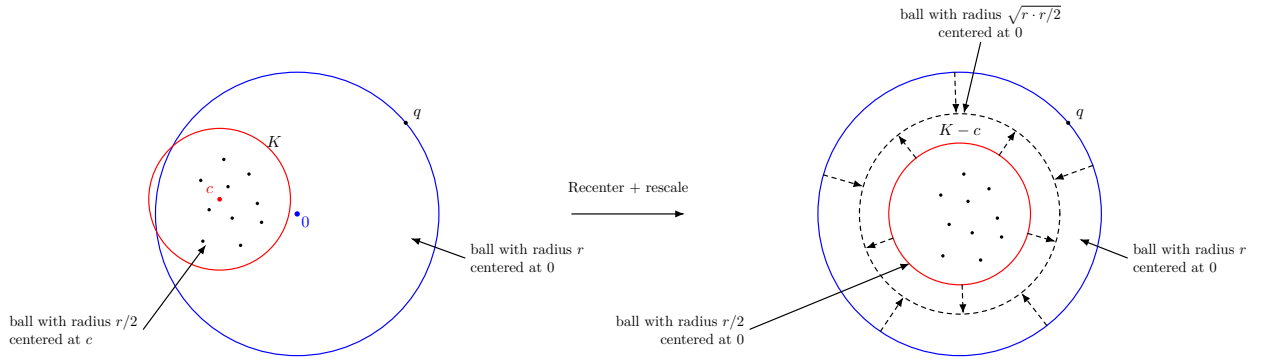

\subsubsection{Coresets for pseudo-random and low-radius datasets}\label{subsub:leaves}

\begin{lemma}[Coresets for pseudo-random datasets]
    \label{lem:ps_rand_coreset}
    Assume that $K'$ is $(B(c', r'), \Delta, \tau)$-pseudo-random, $r' \leq r, \tau \leq \frac{1}{2}$, $|K'| \leq n$. Then there is a randomized procedure that runs in time $O(|K'|^2 \cdot d)$ and outputs $\tilde{K}' \subset K'$ such that with high probability $1 - \frac{1}{\poly(n)}$ for any fixed $\|q\|_2 \leq r$
    \begin{align*}
        \qty|\expker(K', q) - 2\expker(\tilde{K}', q)| \leq O\left(e^{r^2(1 + \Delta)}\cdot \log n \cdot \frac{\exp(K', q)}{|K'|}\right) .
    \end{align*}
\end{lemma}
\begin{proof}
    Denote $\alpha = \sqrt{\frac{r}{r'}}$.
    We get the desired $\tilde{K}$ by invoking \BaseCompress$(\alpha(K' - c'))$ from \Cref{lem:simple_coreset} with probability of success $1 - \frac{1}{\poly(n)}$. From the approximation guarantees of \BaseCompress{} we have that for any fixed $\|q'\|_2 \leq \alpha r' = \sqrt{rr'}$
    \begin{equation}\label{eq:add_bound}
    \begin{aligned}
        \qty|\sum_{k \in K'} e^{\<\alpha (k - c'), q'>} - 2\sum_{k \in K'} e^{\<\alpha (k - c'), q'>}| &\leq O\left(\sqrt{\max_{u \in B(0, \sqrt{rr'})}e^{\|u\|^2_2}\cdot \max_{v \in B(0, \sqrt{rr'})}e^{\|v\|^2_2}}\cdot \log n\right)\\
        & = O\left(e^{rr'}\cdot \log n\right)
    \end{aligned}
    \end{equation}
    
    Because $K'$ is pseudo-random, there are not that many vectors correlated with direction $-q'$, so only a few vectors may have a big negative projection on $q'$. We use this observation to lower bound the softmax KDE value:
    \begin{align*}
        \sum_{k \in K'} e^{\<\alpha(k - c'), q'>} \geq \sum_{\substack{k \in K' \\ \langle k - c', q'\rangle  \geq -\Delta r' \|q'\|_2}} e^{\<\alpha(k - c'), q'>} \geq (1-\tau)|K|e^{-\alpha^2\Delta r'^2} \geq \frac{|K|}{2}e^{-\alpha^2\Delta r'^2}.
    \end{align*}
We plug this bound into \Cref{eq:add_bound} which, using $\langle \alpha(k - c'), q'\rangle = \langle k, q\rangle - \langle c', q\rangle$, yields the desired bound on the approximation guarantee of $\exp(K', q)$:
\begin{align*} 
\qty|\expker(K', q) - 2\expker(\tilde{K}', q)| &= e^{\langle c', q\rangle}\cdot  \qty|\sum_{k \in K'} e^{\<\alpha (k - c'), q'>} - 2\sum_{k \in K'} e^{\<\alpha (k - c'), q'>}| \\
& \leq O\left(e^{\langle c', q\rangle}\cdot e^{rr'}\cdot \log n \cdot \frac{\sum_{k \in K'} e^{\<\alpha (k - c'), q'>}\cdot e^{\alpha^2\Delta r'^2}}{|K'|}\right) \\
& = O\left(e^{rr'(1 + \Delta)}\cdot \log n \cdot \frac{\exp(K', q)}{|K'|}\right).
\end{align*}
It remains to note $r' \leq r$ to finish the proof.
\end{proof}

\begin{lemma}[Coresets for low radius datasets]
    \label{lem:small_r_coreset}
    Assume that $K'$ is contained in $B(c', r')$, $r' \leq r / 2$, $|K'| \leq n$. Then there is a randomized procedure that runs in time $O(|K'|^2 \cdot d)$ and outputs $\tilde{K}'\subset K'$ such that with probability $1 - \frac{1}{\poly(n)}$ for any fixed $\|q\|_2 \leq r$
    \begin{align*}
        \qty|\expker(K', q) - 2\expker(\tilde{K}', q)| \leq O\left(e^{r^2}\cdot \log n \cdot \frac{\exp(K', q)}{|K'|}\right) .
    \end{align*}
\end{lemma}

\begin{proof}
    We proceed as in \Cref{lem:ps_rand_coreset} by invoking \BaseCompress$(\alpha(K'-c'))$ from \Cref{lem:simple_coreset} with success probability $1 - \frac{1}{\poly(n)}$. Just as in \Cref{eq:add_bound}, we get that for any $q'$, $\|q'\|_2 \leq \alpha r' = \sqrt{rr'}$
        \begin{align}\label{eq:add_bound_2}
        \qty|\sum_{k \in K'} e^{\<\alpha (k - c'), q'>} - 2\sum_{k \in K'} e^{\<\alpha (k - c'), q'>}|  \leq O\left(e^{rr'}\cdot \log n\right).
    \end{align}
Next, instead of using pseudo-randomness to lower bound the softmax KDE value, we simply lower bound each term:
    \begin{align*}
        \sum_{k \in K'} e^{\<\alpha(k - c'), q'>} \geq |K'|\cdot \min_{k, q'}e^{-\|k-c\|_2\cdot \|q'\|_2} =  |K'|\cdot e^{-rr'}.
    \end{align*}

Just as before, plugging this into \Cref{eq:add_bound_2} and using $\langle \alpha(k - c'), q'\rangle = \langle k, q\rangle - \langle c', q\rangle$, 
\begin{align*} 
\qty|\expker(K', q) - 2\expker(\tilde{K}', q)| &= e^{\langle c', q\rangle}\cdot  \qty|\sum_{k \in K'} e^{\<\alpha (k - c'), q'>} - 2\sum_{k \in K'} e^{\<\alpha (k - c'), q'>}| \\
& \leq O\left(e^{\langle c', q\rangle}\cdot e^{rr'}\cdot \log n \cdot \frac{\sum_{k \in K'} e^{\<\alpha (k - c'), q'>}\cdot e^{rr'}}{|K'|}\right) \\
& = O\left(e^{2rr'}\cdot \log n \cdot \frac{\exp(K', q)}{|K'|}\right).
\end{align*}
Using $r' \leq r/2$ gives us the desired bound.
\end{proof}

\subsubsection{\Compress{}: offline compression algorithm}\label{subsub:compress}

\begin{algorithm}
\caption{\Compress$(\Delta, K)$}\label{alg:fast_2_compress}
\begin{algorithmic}[1]
\State \textbf{input:} $K \subseteq B(c,r)$
\State $\mathcal{K} \gets \Partition(K, r, r/2, \Delta, 1/2)$\Comment{partition into dense $\&$ pseudo-random parts}
\State
\While{$\sum_i|K_i| > |K|/2$}
\State Sort $|K_1| \geq |K_2| \geq \ldots$
\State $K_1 \gets \BaseCompress\left( \sqrt{\frac{r}{r_1}}(K_1 - c_1)\right)$\label{line:compress_dense}\Comment{recenter \& rescale first}
\State $w_x \gets 2\cdot w_x$ if $x \in K_1$ 
\EndWhile
\State
\State \textbf{output:} $\tilde{K} \coloneqq \bigcup_{i}K_i$

\end{algorithmic}
\end{algorithm}

Our constant-factor compression primitive \Compress{}, \Cref{alg:fast_2_compress} first decomposes a dataset $K \subset B(c,r)$ into pseudo-random and dense parts using \Partition\ (\Cref{alg:partition}), and then applies the factor-$2$ compression routines from \Cref{lem:ps_rand_coreset} and \Cref{lem:small_r_coreset}. Since the guarantees in \Cref{lem:ps_rand_coreset} and \Cref{lem:small_r_coreset} depend on the average Softmax KDE over the parts of $\mathcal{K}$, while we have no control over either these averages or the sizes of the parts, we cannot simply compress every part by the same factor and then relate the resulting error back to $\exp(K,q)$. Instead, we compress only the sufficiently large parts of the partition.

\begin{lemma}\label{lem:fast_2_compress}[Performance guarantees of \Cref{alg:fast_2_compress}]
For any $K \subset B(c, r)$, $|K| \leq  n$, $\Delta \geq \left(\frac{\log\log n}{\log n}\right)^{1/4}$, \Cref{alg:fast_2_compress} outputs $\tilde{K} \subset K$, weight function $w_x \in \{1, 2\}$ $|\tilde{K}| \leq 3|K|/4$ such that with high probability $1 - \frac{1}{\poly(n)}$ for any fixed $\|q\|_2 \leq r$
\begin{align}\label{eq:error_bound_weighted}
\left|\exp(K, q) - \exp(\tilde{K}, w,  q)\right| \leq O\left(e^{r^2(1 + \Delta)}\cdot \qty(\frac{\log |K|}{\Delta^2})^{O\qty(\Delta^{-4})}\cdot \log n \cdot \frac{\exp(K, q)}{|K|}\right),
\end{align}
where $\exp(\tilde{K}, w, q) \coloneqq \sum_{x \in \tilde{K}}w_x e^{\langle x, q\rangle}$ is the weighted KDE.
The runtime of \Cref{alg:fast_2_compress} is \[O\left(\max\left\{|K|^2d, \frac{d}{\Delta^{4}}|K|\log n \cdot \frac{1}{\Delta^4} \right\}\right)\] and the probability of success is $1 - \frac{1}{\poly(n)}$.
\end{lemma}

\begin{remark}[Selection of $\Delta$]\label{rem:Delta} From now on, we select $\Delta$ by approximately minimizing the error bound in \Cref{eq:error_bound_weighted}:
\begin{align}\label{eq:delta}
     e^{r^2(1 + \Delta)}\cdot \qty(\frac{\log |K|}{\Delta^2})^{\Delta^{-4}} = e^{r^2(1 + \Delta) +\Delta^{-4}\log(\log |K|/\Delta^2) }.
\end{align}
Setting $\Delta \coloneqq Const\cdot \left(\frac{\log\log n}{\log n}\right)^{1/5}$ for a big enough constant $Const$, the expression above becomes
\[\leq e^{r^2(1 + o(1))} \cdot e^{(\log n)^{4/5}\cdot (\log\log n)^{1/5}} = e^{r^2(1 + o(1))} \cdot n^{o(1)}.\]
Hence, the error bound \Cref{eq:error_bound_weighted} becomes
\begin{align}\label{eq:error_bound_w-out_delta}
\left|\exp(K, q) - \exp(\tilde{K}, w,  q)\right| \leq O\left(e^{r^2(1 + \Delta)}\cdot n^{o(1)} \cdot \frac{\exp(K, q)}{|K|}\right),
\end{align}
and the runtime becomes
\[O\left(|K|^2\poly(\log n)\right),\]
as $d \leq \poly(\log n)$ and $1/\Delta \leq \poly(\log n)$.

\end{remark}
\begin{proof}
\textbf{Error bound.} Let $I = |\mathcal{K}|$ be the number of parts in the partition $\mathcal{K}$. First, note that \Compress$(\Delta, K)$ never compresses a part $K_i$ with $|K_i| \leq |K|/(2I)$, nor does it ever compress the same dataset twice. Indeed, the union of all parts satisfying $|K_i| \leq |K|/(2I)$ contains at most $n/2$ points. Thus the remaining $|K|/2$ points belong to parts that are compressed by the algorithm, so before we ever need to touch any such small part, the required factor-$3/4$ compression has already been achieved.

We bound the errors incurred in line~\eqref{line:compress_dense} by \Cref{lem:ps_rand_coreset} and \Cref{lem:small_r_coreset}, respectively. Hence the total error is bounded by
\begin{align*}
&\sum_{i}\mathbbm{1}\{|K_i|> |K|/(2I)\}\cdot O\left(e^{r^2(1 + \Delta)}\cdot \log |K_i| \cdot \frac{\exp(K_i, q)}{|K_i|}\right) \\
\leq {}& \sum_i O\left(e^{r^2(1 + \Delta)}\cdot \log n \cdot \frac{\exp(K_i, q)}{|K|/(2I)}\right) \\
\leq {}& O\left(e^{r^2(1 + \Delta)}\cdot I\cdot \log n \cdot \frac{\exp(K, q)}{|K|}\right).
\end{align*}
Finally, by \Cref{lem:ps_randomification},
\[
I \leq \qty(\frac{\log n}{\Delta^2})^{O\qty(\Delta^{-4})},
\]
and the result follows.

\paragraph{Runtime.} By \Cref{lem:ps_randomification}, we obtain the partition $\mathcal{K}$ in time
\[
O\!\left(\frac{d}{\Delta^{4}} |K|\log n \cdot \frac{1}{\Delta^4} \right).
\]
Throughout the algorithm, we invoke the routine in line~\eqref{line:compress_dense} on disjoint datasets. An invocation on a dataset $K_i$, by \Cref{lem:ps_rand_coreset} and \Cref{lem:small_r_coreset}, takes time $O(|K_i|^2 d)$. Hence, we bound the runtime spent in line~\eqref{line:compress_dense} as $O(|K|^2d)$, the total runtime is bounded by
\[
O\left(\max\left\{|K|^2 d, \frac{d}{\Delta^{4}\tau^2} |K|\log n \cdot \frac{1}{\Delta^4} \right\}\right).
\]

\paragraph{Probability of success.} From the assumption $\Delta \geq \left(\frac{\log\log n}{\log n}\right)^{1/4}$, we get
\[
I \leq \left(\log n\cdot \sqrt{\frac{\log n}{\log\log n}}\right)^{O(\log n/\log\log n)} \leq \poly(n).
\]
Our algorithm succeeds provided each of the $I$ calls to \BaseCompress{} succeeds. Since, by \Cref{lem:ps_rand_coreset} and \Cref{lem:small_r_coreset}, each such call succeeds with probability $1 - \frac{1}{\poly(n)}$, a union bound implies that the entire algorithm succeeds with probability $1 - \frac{1}{\poly (n)}$.
\end{proof}

\subsubsection{From \Compress{} to streaming (weighted version of Merge-and-Reduce)}\label{subsub:merge_reduce_w}

 The proof of the upper bound from \Cref{th:big_r_ub_lb} follows by applying the offline factor $3/4$ compression algorithm \Cref{lem:fast_2_compress} in streaming via the Merge-and-Reduce pipeline described in \Cref{sec:kde_coresets}. We want a statement very close to \Cref{lem:merge_reduce}~-- in fact, there are only two subtleties which prevent us from applying this statement as a black-box: it assumes that the base compression doubles the weight $w$ of all vertices, and it assumes that the base compression algorithm returns $1/2$ of the points instead of $3/4$. 
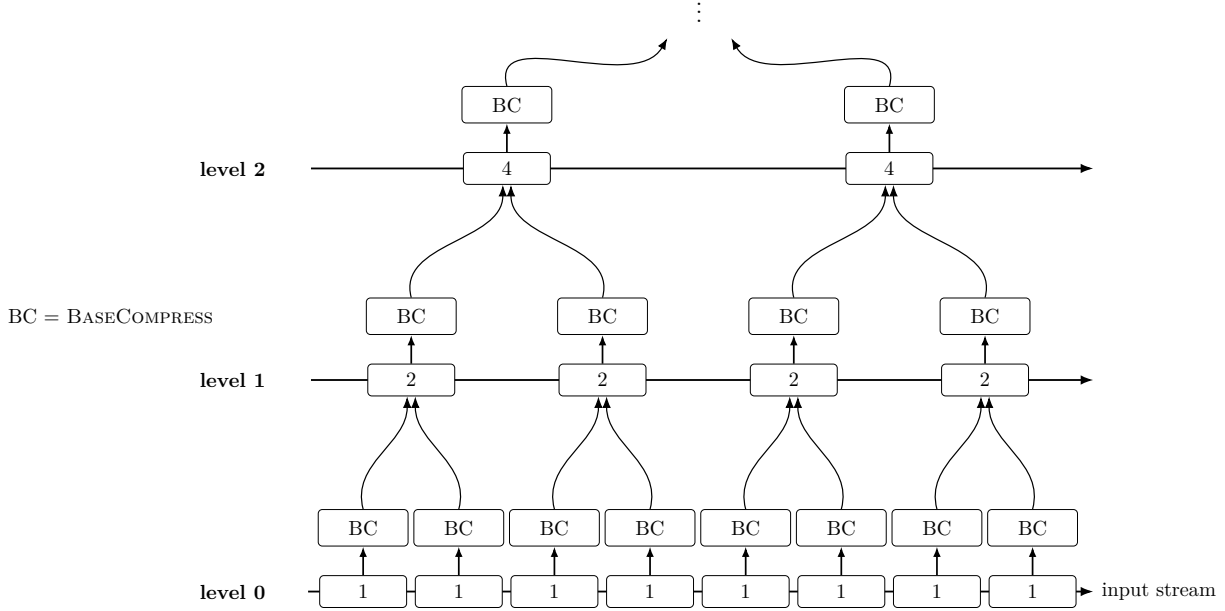
\begin{figure}[H]
\centering
\vspace{5mm}
\resizebox{0.98\linewidth}{!}{%
\begin{tikzpicture}[
  >=Latex,
  font=\small,
  level/.style={font=\bfseries\small, anchor=east},
  stream/.style={-{Latex[length=2.2mm]}, line width=0.9pt},
  fan/.style={-{Latex[length=2mm]}, semithick},
  tocomp/.style={-{Latex[length=1.8mm]}, line width=0.85pt},
  block/.style={draw, rounded corners=2pt, fill=white, minimum width=15mm, minimum height=5.5mm, inner sep=0pt},
  comp/.style={draw, rounded corners=2pt, fill=white, minimum width=15.5mm, minimum height=6.3mm, inner sep=1pt},
  dots/.style={font=\large}
]

% ------------------------------------------------------------------
% Legend on the left
% ------------------------------------------------------------------
\node[anchor=east, font=\small] at (-0.15,4.75) {BC = \BaseCompress{}};

% ------------------------------------------------------------------
% Level 0
% ------------------------------------------------------------------
\node[level] at (0.75,0.0) {level 0};
\draw[stream] (1.35,0) -- (14.90,0) node[right] {input stream};

\node[block] (b01) at (2.30,0) {$1$};
\node[block] (b02) at (3.95,0) {$1$};
\node[block] (b03) at (5.60,0) {$1$};
\node[block] (b04) at (7.25,0) {$1$};
\node[block] (b05) at (8.90,0) {$1$};
\node[block] (b06) at (10.55,0) {$1$};
\node[block] (b07) at (12.20,0) {$1$};
\node[block] (b08) at (13.85,0) {$1$};

\node[comp] (c01) at (2.30,1.10) {BC};
\node[comp] (c02) at (3.95,1.10) {BC};
\node[comp] (c03) at (5.60,1.10) {BC};
\node[comp] (c04) at (7.25,1.10) {BC};
\node[comp] (c05) at (8.90,1.10) {BC};
\node[comp] (c06) at (10.55,1.10) {BC};
\node[comp] (c07) at (12.20,1.10) {BC};
\node[comp] (c08) at (13.85,1.10) {BC};

\foreach \b/\c in {b01/c01,b02/c02,b03/c03,b04/c04,b05/c05,b06/c06,b07/c07,b08/c08}{
  \draw[tocomp] (\b.north) -- (\c.south);
}

% ------------------------------------------------------------------
% Level 1
% ------------------------------------------------------------------
\node[level] at (0.75,3.65) {level 1};
\draw[stream] (1.40,3.65) -- (14.90,3.65);

\node[block] (b11) at (3.13,3.65) {$2$};
\node[block] (b12) at (6.43,3.65) {$2$};
\node[block] (b13) at (9.73,3.65) {$2$};
\node[block] (b14) at (13.03,3.65) {$2$};

\draw[fan] (c01.north) to[out=105,in=-90] ([xshift=-2pt]b11.south);
\draw[fan] (c02.north) to[out=75,in=-90]  ([xshift= 2pt]b11.south);

\draw[fan] (c03.north) to[out=105,in=-90] ([xshift=-2pt]b12.south);
\draw[fan] (c04.north) to[out=75,in=-90]  ([xshift= 2pt]b12.south);

\draw[fan] (c05.north) to[out=105,in=-90] ([xshift=-2pt]b13.south);
\draw[fan] (c06.north) to[out=75,in=-90]  ([xshift= 2pt]b13.south);

\draw[fan] (c07.north) to[out=105,in=-90] ([xshift=-2pt]b14.south);
\draw[fan] (c08.north) to[out=75,in=-90]  ([xshift= 2pt]b14.south);

\node[comp] (c11) at (3.13,4.75) {BC};
\node[comp] (c12) at (6.43,4.75) {BC};
\node[comp] (c13) at (9.73,4.75) {BC};
\node[comp] (c14) at (13.03,4.75) {BC};

\foreach \b/\c in {b11/c11,b12/c12,b13/c13,b14/c14}{
  \draw[tocomp] (\b.north) -- (\c.south);
}

% ------------------------------------------------------------------
% Level 2
% ------------------------------------------------------------------
\node[level] at (0.75,7.30) {level 2};
\draw[stream] (1.40,7.30) -- (14.90,7.30);

\node[block] (b21) at (4.78,7.30) {$4$};
\node[block] (b22) at (11.38,7.30) {$4$};

\draw[fan] (c11.north) to[out=105,in=-90] ([xshift=-2pt]b21.south);
\draw[fan] (c12.north) to[out=75,in=-90]  ([xshift= 2pt]b21.south);

\draw[fan] (c13.north) to[out=105,in=-90] ([xshift=-2pt]b22.south);
\draw[fan] (c14.north) to[out=75,in=-90]  ([xshift= 2pt]b22.south);

\node[comp] (c21) at (4.78,8.40) {BC};
\node[comp] (c22) at (11.38,8.40) {BC};

\foreach \b/\c in {b21/c21,b22/c22}{
  \draw[tocomp] (\b.north) -- (\c.south);
}

% continuation
\draw[fan] (c21.north) to[out=105,in=-120] (7.55,9.55);
\draw[fan] (c22.north) to[out=75,in=-60]  (8.65,9.55);
\node[dots] at (8.10,10.10) {$\vdots$};

\end{tikzpicture}%
}
\vspace{4mm}
\caption{Illustration of the standard Merge-and-Reduce algorithm (\Cref{lem:merge_reduce}).}
\label{fig:merge-reduce-standard}
\vspace{5mm}
\end{figure}

Our updated \textit{weighted} Merge-and-Reduce framework is illustrated schematically in \Cref{fig:merge-reduce-weighted}. As in the standard Merge-and-Reduce framework, illustrated in \Cref{fig:merge-reduce-standard}, we begin by partitioning the stream into blocks and applying the offline compression routine \Compress{} to each block. However, instead of immediately merging the outputs from level \(0\), sending them to level \(1\), and recursing as in \Cref{lem:merge_reduce}, we first separate the outputs of \Compress{} into \textit{weight classes}. This is important because \Compress{} is designed to work with a dataset of points with the same weight. Within each weight class, we then merge the outputs, repartition them into blocks, and apply \Compress{} blockwise. The weights of points in the output are updated multiplicatively: each point’s current weight is multiplied by the weight assigned to it by \Compress{}. The resulting outputs at level \(1\) are again grouped by weight and merged within weight classes. All subsequent levels are handled in the same way.

Note that weighted Merge-and-Reduce no longer has the structure of a binary tree, even if all outputs of \Compress{} have the same weight. The reason is that, unlike \BaseCompress{}, \Compress{} does not reduce the dataset size by a factor of \(2\). In the standard Merge-and-Reduce framework, a block is compressed only after both of its children have been compressed by \BaseCompress{}. In the weighted setting, by contrast, we wait until the memory buffer associated with a given weight class at a given level contains exactly \(b\) points, where \(b\) is the chosen block size. We then compress this buffer and send the resulting coreset to the next level.

\begin{figure}[H]
\centering
\vspace{5mm}
\scalebox{0.9}{
\begin{tikzpicture}[
  >=Latex,
  font=\small,
  level/.style={font=\bfseries\small, anchor=east},
  stream/.style={-{Latex[length=2.2mm]}, line width=0.9pt},
  fan/.style={-{Latex[length=2mm]}, semithick},
  tocomp/.style={-{Latex[length=1.8mm]}, line width=0.85pt},
  block/.style={draw, rounded corners=2pt, fill=white, minimum width=15mm, minimum height=5.5mm, inner sep=0pt},
  comp/.style={draw, rounded corners=2pt, fill=white, minimum width=15.5mm, minimum height=6.3mm, inner sep=1pt},
  dots/.style={font=\large}
]

% ------------------------------------------------------------------
% Level 0
% ------------------------------------------------------------------
\node[level] at (0.75,0.0) {level 0};
\draw[stream] (1.35,0) -- (11.55,0) node[right] {input stream};

\node[block] (b01) at (2.30,0) {$1$};
\node[block] (b02) at (4.25,0) {$1$};
\node[dots] at (6.25,0) {$\ldots$};
\node[block] (b03) at (8.25,0) {$1$};
\node[block] (b04) at (10.20,0) {$1$};

\node[comp] (c01) at (2.30,1.10) {$\Compress$};
\node[comp] (c02) at (4.25,1.10) {$\Compress$};
\node[comp] (c03) at (8.25,1.10) {$\Compress$};
\node[comp] (c04) at (10.20,1.10) {$\Compress$};

\foreach \b/\c in {b01/c01,b02/c02,b03/c03,b04/c04}{
  \draw[tocomp] (\b.north) -- (\c.south);
}

% ------------------------------------------------------------------
% Level 1
% ------------------------------------------------------------------
\node[level] at (0.75,3.65) {level 1};

\draw[stream] (1.40,3.05) -- (11.55,3.05);
\draw[stream] (1.40,5.35) -- (11.55,5.35);

% fan-out from level 0; each compressor contributes to both weight streams
\draw[fan] (c01.north west) to[out=108,in=-90] (2.075,2.775);
\draw[fan] (c01.north east) to[out=72,in=-90]  (2.925,5.075);
\draw[fan] (c02.north west) to[out=108,in=-90] (2.825,2.775);
\draw[fan] (c02.north east) to[out=72,in=-90]  (3.675,5.075);
\draw[fan] (c03.north west) to[out=108,in=-90] (8.275,2.775);
\draw[fan] (c03.north east) to[out=72,in=-90]  (9.125,5.075);
\draw[fan] (c04.north west) to[out=108,in=-90] (9.025,2.775);
\draw[fan] (c04.north east) to[out=72,in=-90]  (9.875,5.075);

% partition blocks, opaque so arrows/stream do not show through
\node[block] (b11) at (2.45,3.05) {$1$};
\node[block] (b12) at (8.65,3.05) {$1$};
\node[dots] at (5.80,3.05) {$\ldots$};

\node[block] (b21) at (3.30,5.35) {$2$};
\node[block] (b22) at (9.50,5.35) {$2$};
\node[dots] at (6.40,5.35) {$\ldots$};

\node[comp] (c11) at (2.45,3.95) {$\Compress$};
\node[comp] (c12) at (8.65,3.95) {$\Compress$};
\node[comp] (c21) at (3.30,6.25) {$\Compress$};
\node[comp] (c22) at (9.50,6.25) {$\Compress$};

\foreach \b/\c in {b11/c11,b12/c12,b21/c21,b22/c22}{
  \draw[tocomp] (\b.north) -- (\c.south);
}

% omitted intermediate levels
\node[dots] at (6.50,7.10) {$\vdots$};

% ------------------------------------------------------------------
% Generic level i
% ------------------------------------------------------------------
\node[level] at (0.75,9.20) {level $i$};

\draw[stream] (4.20,8.20) -- (11.55,8.20);
\draw[stream] (4.20,9.80) -- (11.55,9.80);
\draw[stream] (4.20,13.00) -- (11.55,13.00);

\node[block] (bi11) at (5.55,8.20) {$1$};
\node[block] (bi12) at (9.45,8.20) {$1$};
\node[dots] at (7.50,8.20) {$\ldots$};

\node[block] (bi21) at (5.55,9.80) {$2$};
\node[block] (bi22) at (9.45,9.80) {$2$};
\node[dots] at (7.50,9.80) {$\ldots$};

\node[dots] at (7.50,11.40) {$\vdots$};

\node[block] (bi31) at (5.55,13.00) {$2^i$};
\node[block] (bi32) at (9.45,13.00) {$2^i$};
\node[dots] at (7.50,13.00) {$\ldots$};

\node[comp] (ci11) at (5.55,9.10)  {$\Compress$};
\node[comp] (ci12) at (9.45,9.10)  {$\Compress$};
\node[comp] (ci21) at (5.55,10.70) {$\Compress$};
\node[comp] (ci22) at (9.45,10.70) {$\Compress$};
\node[comp] (ci31) at (5.55,13.90) {$\Compress$};
\node[comp] (ci32) at (9.45,13.90) {$\Compress$};

\foreach \b/\c in {bi11/ci11,bi12/ci12,bi21/ci21,bi22/ci22,bi31/ci31,bi32/ci32}{
  \draw[tocomp] (\b.north) -- (\c.south);
}

\end{tikzpicture}
}
\vspace{4mm}
\caption{Illustration of the weighted Merge-and-Reduce algorithm (\Cref{lem:merge-reduce-w}).}
\label{fig:merge-reduce-weighted}
\vspace{5mm}
\end{figure}
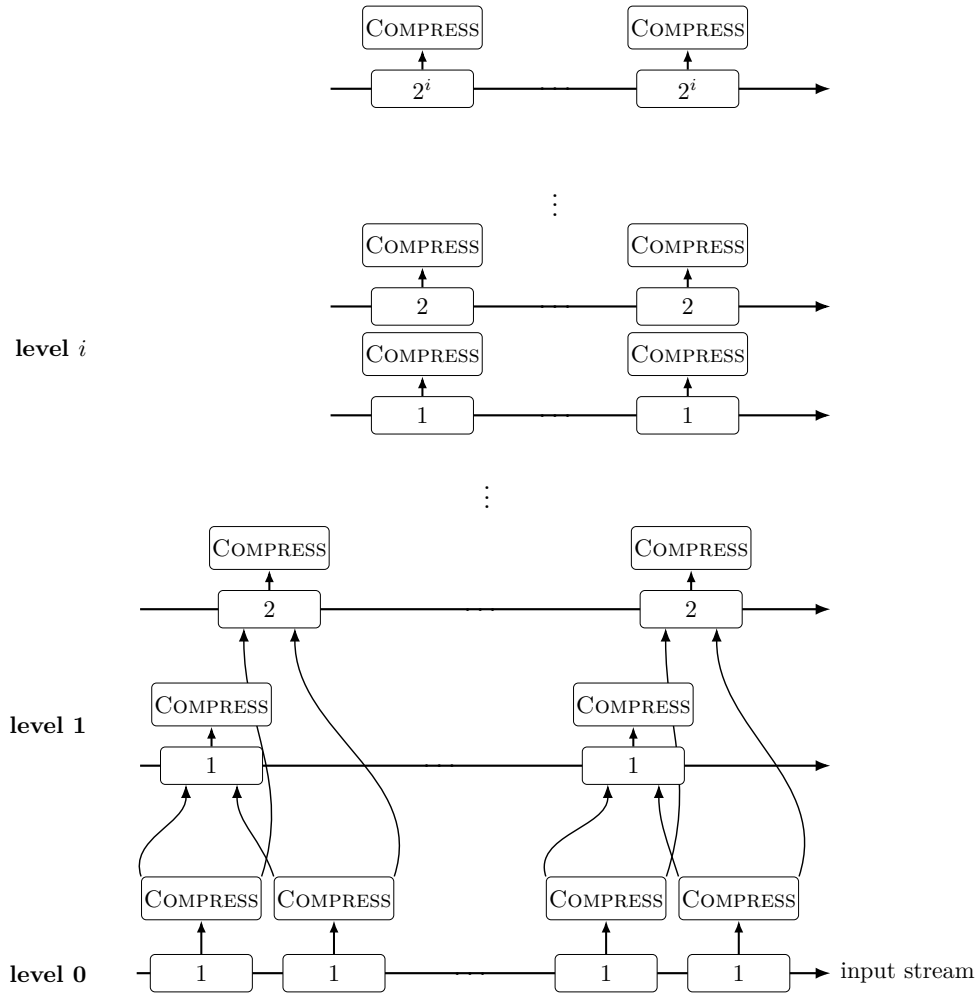

\begin{lemma}[Weighted Merge-and-Reduce]\label{lem:merge-reduce-w}    Let $\mathcal{A}$ be a black-box offline algorithm which takes as input a dataset $K \subset B(0,r)$, $|K| \leq n$,  and returns $\widetilde{K} \subset K$, $|\widetilde{K}| \leq 3|K|/4$, weight function $\alpha_x \in \{1, 2\}$ such that with probability $1 - \frac{1}{\poly n}$, for any fixed $q \in B(0, r)$,
\begin{align}\label{eq:error_bound-mr}
\left|\exp(K, q) - \exp(\tilde{K}, \alpha,  q)\right| \leq O\left(e^{r^2(1 + o(1))}\cdot n^{o(1)}\cdot \frac{\exp(K, q)}{|K|}\right).
\end{align}
   
   Let $T(m)$ denote the runtime $\mathcal{A}$ on inputs of size $m$. Assume that the space complexity of $\mathcal{A}$ is $O(m)$.

   Algorithm $\mathcal{A}$ can be adapted to process $K$ as a stream. At every step $j$ of the stream, it maintains a weighted coreset $C_j$ of the $j$ first elements of the stream $K_j$. Let $w_x^j$ denote the weight of a point $x \in C_j$ and $w^j$~-- the weight function. We have with probability $1-\frac{1}{\poly (n)}$, for any fixed $q \in U$ and desired error $\epsilon \in (0, 1)$,
   \[\left|
  \exp(K_j,q)-\exp(C_j, w^j, q)\right|
   \leq \epsilon \cdot \exp(K_j, q).
   \]
The size of the coreset, runtime and memory are defined in terms of the \textit{block size} of the algorithm $b$
\[
|C_j| \leq O\left(b\log^2\left(\frac{|K_j|}{b}\right)\right), \quad \text{ runtime} =  O\left(T(b)\cdot \log^2\left(\frac{|K_j|}{b}\right)\right), \quad \text{ space}=
   O\left(b\log^2\left(\frac{|K_j|}{b}\right)\right),
   \]
and the block size $b$ is defined by the desired error $\epsilon$ 
   \[
   b = \frac{e^{r^2(1 + o(1))}\cdot n^{o(1)}}{\epsilon}.
   \]
\end{lemma}
\begin{proof}
Let
\[
\alpha := e^{r^2(1+o(1))}\cdot n^{o(1)}.
\]
 Each level of the weighted Merge-and-Reduce algorithm consists of at most $O\left(\log \frac{|K_j|}{b}\right)$ weight classes but every point can belong to at most one stream per level. For every point $x$ which has appeared at the level $i$ at some point before the arrival of $k_{j+1}$, we use $\omega_x^i$ to denote the weight class of $x$. 

Since each invocation of $\mathcal{A}$ on a dataset shrinks it by a factor of $3/4$, at every time step $j$ there are at most $O(\log(|K_j|/b))$ non-empty levels. Therefore, because at any time we store at most $b$ data points per weight class at each level,
\[
|C_j| = O\!\left(b\log^2\!\left(\frac{|K_j|}{b}\right)\right),\qquad
\text{space}=O\!\left(b\log^2\!\left(\frac{|K_j|}{b}\right)\right).
\]
A new arrival in the stream can trigger at most one invocation of $\mathcal A$ per weight class at each level, so
\[
\text{runtime}=O\!\left(T(b)\log^2\!\left(\frac{|K_j|}{b}\right)\right).
\]

It remains to bound the error. Fix $q\in B(0, r)$. Using the guarantees of the offline algorithm given by \Cref{eq:error_bound-mr}, the total error produced by our algorithm at level $0$ at time step $j$ can be bounded by
\[
O\left(\alpha \cdot \frac{\exp(K_j, q)}{b}\right).
\]
Note that the error will be incurred block by block, but we use the fact that the softmax kernel is linear across partitions of $K$.

We can bound the total error at every subsequent level $i$ similarly despite the fact that the guarantees of $\mathcal{A}$ depend on error with the kernel of the prior layer, rather than the original kernel $\expker(K_j, q)$.
At the end of the proof, by choosing $b$ appropriately, we will ensure that the sum of the errors over all levels does not exceed $\epsilon \cdot \exp(K_j, q)$:
\[\left|\exp(K_j, q) - \sum_{k \in \text{level }1}\omega_k^1 e^{\langle k, q\rangle}\right| + \sum_{i \geq 1}\left|\sum_{k \in \text{level } i}\omega^i_ke^{\langle k, q\rangle} - \sum_{k \in \text{level } (i+1)}\omega^{i+1}_ke^{\langle k, q\rangle} \right| \leq \epsilon \cdot \exp(K_j, q),\]
so in particular we may use
\[
\left|\exp\left(K_j, q\right) - \sum_{k \in \text{level } i}\omega_k^i e^{\langle k, q\rangle}  \right| \leq \epsilon\cdot \exp(K_j, q).
\]

The sum of the errors produced by invocations of $\mathcal{A}$ on all blocks of all streams at level $i$ is, therefore, bounded by
\[
O\left(\alpha \cdot \frac{\sum_{k \in \text{level } i}\omega_k^i e^{\langle k, q\rangle} }{b}\right)\leq O\left(\alpha \cdot \frac{\exp\left(K_j, q\right)}{b}\right).
\]

Thus, we have bounded the error at every level by the same expression. Hence, the total error can be bounded by
\[
\left|\exp(C_j, w, q) - \exp(K_j,q )\right| \leq O\left(\alpha \cdot \log\frac{|K_j|}{b}\cdot  \frac{\exp(K_j, q)}{b}\right).
\]

In order to ensure that the right-hand side is bounded by $\epsilon \cdot \exp(K_j, q)$, it suffices to set
\[
b = O\!\left(\frac{\alpha \log |K|}{\epsilon}\right)
=
O\!\left(\frac{e^{r^2(1+o(1))}\cdot n^{o(1)}}{\epsilon}\right).
\]

Finally, note that our algorithm succeeds as long as all calls to $\mathcal{A}$ succeed. We invoke $\mathcal{A}$ at most $\frac{|K_j|}{b}$ times at each level, and there are at most $O\left(\log \frac{|K_j|}{b}\right)$ levels. Since every call to $\mathcal{A}$ succeeds with probability $1 - \frac{1}{\poly(n)}$, a union bound implies that the overall algorithm succeeds with probability $1-\frac{1}{\poly(n)}$.
\end{proof}

%\begin{proof}[Proof of \Cref{th:big_r_ub_lb}] The desired procedure is simply an invocation of the weighted Merge-and-Reduce framework together with \Compress{} as the black-box compression routine $\mathcal{A}$.
%\end{proof}
\subsection{Lower Bound}\label{sec:lb_low_temp}
In this section we are going to prove the lower bound from \Cref{th:big_r_ub_lb}. The proof exploits again the idea from the lower bound in \Cref{th:const_r_ub_lb} providing a reduction to $\texttt{INDEX}_{2n,n}$ problem with sublinear amount of advice bits. But the analysis is going to be different, since we are now in $r = \Omega(\log n)$ regime instead of a constant $r$. The signal term will not be dominated by a single term of the tail as in Hermite case. We will argue instead that the signal term is still of order $e^{r^2}$ by using Taylor expansion. This also allows us to prove the bound for the whole range of $d = (\log n)^{1 + \Omega(1)}$. Since we do not use orthogonal properties of polynomials anymore, we switch to uniform distribution on a sphere instead of Gaussian, so we do not need rescaling. 
    
Below is our main lemma that together with $R^{\text{pub}, \to}(\texttt{INDEX}_{2n, n}) = \Omega(n)$ from \Cref{thm:indexing} implies the lower bound of \Cref{th:big_r_ub_lb}. Note that we can assume not only $r^2 = \Omega(\log n)$, but $r^2 = \Theta(\log n)$ as otherwise we can consider smaller $r' = O(\log n)$ such that $n \ll \frac{e^{r'^2}}{\eps}$ and therefore show $\Omega(n)$ lower bound for original $r$.

\begin{lemma}
    \label{lem:reduction_big_r}
     %Consider a point set $P$ such that each $p \in P$ is drawn i.i.d.\ from a sphere of radius $r$ generated with public coins. Let $\bx$ be Alice's input from the $\texttt{INDEX}_{2n, n}$ problem.
    Let $S_{n, d, r, \eps}$ denote a deterministic upper bound on the space, measured in bits, of a (randomized) data structure for the Streaming KDE problem  with parameters $n, d, r, \eps$ and a failure probability $1/100$. Assume that $d = \omega(\log n), r^2 = \Theta(\log n)$, $ n < \frac{\exp(r^2 \p{1 - C\sqrt{\frac{\log n}{d}}})}{16\eps}$ for some constant $C$, there exists a communication protocol with a message size $S_{n, d, r, \eps} + o(n)$ which solves the $\texttt{INDEX}_{2n, n}$ problem with probability at least $9/10$.
\end{lemma}

The rest of the section is dedicated to proving \cref{lem:reduction_big_r}.

\paragraph{The reduction.} Using the public randomness, Alice and Bob agree on a point set $K = \{k_1, k_2, \cdots, $ $k_{2n}\} \subset \R^d$ such that each $k_i \in K$ is drawn i.i.d.\ uniformly at random from a $d$-dimensional Euclidean sphere of radius $r$. Given her input string $\bx \in \{0,1\}^{2n}$ from the $\texttt{INDEX}_{2n, n}$ problem, Alice defines a corresponding dataset $K_\bx = \{ k_i \in K\mid \bx_i = 1\}$. Then she constructs her message to Bob, which contains two components:  for the first component, Alice builds a KDE data structure on the dataset $K_\bx$ with parameters $n, d, r, \eps$. This takes $S_{n, d, r, \eps}$ bits. For the second component,   Alice constructs an additional sketch that approximates KDE for the same dataset $K_{\bx}$  but with polynomial kernel
\[
\expker_{\leq t}(k, q) = \sum_{l = 0}^{t} \frac{\<k, q>^l}{l!}.
\]

This sketch is defined via a pair of feature maps $\psi_A, \psi_B$ such that $\langle \psi_A(k), \psi_B(q) \rangle$ approximates $\expker_{\leq t}(k, q)$. The precise guarantee is described in the following lemma.
\begin{lemma}
    \label{lem:taylor_embed_memory}
    Let $t < d$. There exist embeddings $\psi_A, \psi_B$ such that 
    \[
    \abs{\Bigl \langle \sum_{k \in K_\bx} \psi_A(k), \psi_B(q) \Bigr \rangle - \expker_{\leq t}(K_\bx, q)} \leq \eps \expker(K_\bx, q),
    \]
    and that $\sum_{k \in K_\bx} \psi_A(k)$ can be stored in $O(\frac{(d+t)^t}{t!} \log^2 n)$ bits of space.
\end{lemma}
\begin{proof}
    As shown by \Cref{lem:embed_memory}, tensor powers embedding with polynomial precision gives the desired data structure with memory consumption ${d + t\choose t}\cdot O(t\log(d+t) + \log n) = O(\frac{(d+t)^t}{t!} \log^2 n)$.
\end{proof}

Alice sends the  KDE data structure of $K_\bx$ and $\sum_{k \in K_\bx} \psi_A(k)$ to Bob. We now describe the second part of the protocol, which is how Bob uses these two information to reveal $\bx_j$ for his input index $j$ from the $\texttt{INDEX}_{2n, n}$ problem, with probability at least $9/10$. This task is equivalent to detecting whether $k_j \in K_{\bx}$. For this, Bob first queries the data structure with $k_j$.
With $.99$ chance,   the KDE data structure outputs $\widehat{\expker}(K_\bx , k_j)$ with 
\[ |\widehat{\expker}(K_\bx , k_j) - \expker(K_\bx, k_j) | \leq \eps \expker(K_\bx, k_j).\] 
Bob also computes 
\[\widehat{\expker}_{\leq t}(K_{\bx}, k_j) := \Bigl\langle \sum_{k \in K_\bx} \psi_A(k), \psi_B(k_j)  \Bigr \rangle\] using the sketch $\sum_{k \in K_\bx} \psi_A(k)$ he receives.
After that, he can compute $\widehat{\expker}(K_{\bx}, k_j) - \widehat{\expker}_{\leq t}(K_{\bx}, k_j)$, which is a noisy evaluation of $\expker_{>t}(K_{\bx}, k_j)$. Formally, with $.99$ chance,
\begin{equation}\label{eq:approximation_err}
    \Bigl |\Bigl (\widehat{\expker}(K_{\bx}, k_j) - \widehat{\expker}_{\leq t}(K_{\bx}, k_j) \Bigr )  - \expker_{>t}(K_{\bx}, k_j)\Bigr| \leq 2\eps \expker(K_\bx, k_j).
\end{equation} 

We next argue that Bob can decide whether  $k_j \in K_{\bx}$ by thresholding the quantity $\widehat{\expker}(K_{\bx}, k_j) - \widehat{\expker}_{\leq t}(K_{\bx}, k_j) $. The idea is that when $k_j \in K_{\bx}$, $\expker_{> t}(K_{\bx}, k_j) $ contains a ``signal term'' $\expker_{> t}(k_j, k_j)$. We will argue that the contribution of this signal term typically stands out from (1) the noise induced by the contribution of $\expker_{> t}(k, k_j)$ for $k \neq k_j$ in $K_\bx$  and (2) the approximation error from \cref{eq:approximation_err}. Concretely, we will show that
there exists an appropriate choice of $t$ such that \begin{itemize}
    \item $\expker_{> t}(k_j, k_j) \geq \frac{1}{2}e^{r^2}$,
    \item $\Var{\expker_{>t}(K_{\bx} \setminus \{k_j\}, k_j)}= o(e^{2r^2})$, and
    \item $2\eps \expker(K_\bx, k_j) \leq  \frac{1}{8}e^{r^2}$.
\end{itemize}
Moreover, for this choice of $t$, $\widehat{\expker}(K_\bx , k_j)$ can indeed be sent using $o(n)$ bits. Once we have these, Bob can detect whether $k_j \in K_{\bx}$  by thresholding whether 
$\widehat{\expker}(K_{\bx}, k_j) - \widehat{\expker}_{\leq t}(K_{\bx}, k_j) \geq \mu + \frac{1}{4}e^{r^2}$ for explicitly computable $\mu := \E_{K}[\expker_{> t} (K_\bx , k_j) \mid \bx_j = 0)]$. This succeeds with high constant probability by an application of Chebyshev's inequality. This concludes the description of the communication protocol.

The rest of this section is organized as follows: we first bound the signal in \cref{lem:signal_bound}. Then we bound the variance with respect to $t$ in  \cref{lem:variance_bound}. Lastly, we state our choice of $t$, and show that it leads to the claimed final bounds for the variance, the approximation noise, and the bit complexity.

\begin{lemma}
    \label{lem:signal_bound}
    For any $t < r^2$ and any $k_j$ from a sphere of radius $r$, we have $\expker_{>t}(k_j, k_j) \geq \frac{1}{2}e^{r^2}$.
\end{lemma}
\begin{proof}
    Since the cutoff point $t < r^2$,
    \begin{align*}
        \expker_{>t}(k_j, k_j) = \sum_{l=t+1}^{\infty} \frac{r^{2l}}{l!} > \sum_{l=\floor{r^2}}^{\infty} \frac{r^{2l}}{l!} > \frac{1}{2}e^{r^2}
    \end{align*}
    where the last inequality uses the fact that the median of Poisson distribution is within distance $1$ from its mean.
    \end{proof}

\begin{lemma}
    \label{lem:variance_bound}
    For any $t > (\log n)^{1 - o(1)}$, any fixed $k_j$ from a sphere of radius $r$, and for uniformly random i.i.d. draws of $v_1, \cdots, v_n$ from a sphere of radius $r$, we have
    \begin{align*}
        \Var{\expker_{>t}(\{v_1, \cdots, v_n\}, k_j)} = O\qty(n\frac{(2r)^{4t}}{d^t t!}).
    \end{align*}
\end{lemma}
\begin{proof}
    For each $v$,
    \begin{align*}
        \Var{\expker_{>t}(v, k_j)} \leq \E \expker^2_{>t}(v, k_j) = \E\qty(\sum_{l=t+1}^{\infty} \frac{\<v, k_j>^{l}}{l!})^2
    \end{align*}
    Terms in this sum are heavily correlated, but they decay rapidly. Therefore we can bound
    \begin{gather*}
        \E\qty(\sum_{l=t+1}^{\infty} \frac{\<v, k_j>^{l}}{l!})^2 \leq 2\E\qty(\frac{\<v, k_j>^{t+1}}{(t+1)!})^2 + 2\E\qty(\sum_{l=t+2}^{\infty} \frac{\<v, k_j>^{l}}{l!})^2 \leq \dots \leq \sum_{l=t+1}^{\infty} \frac{2^{l-t}\E\<v, k_j>^{2l}}{(l!)^2}.
    \end{gather*}
    Using \Cref{lem:expectation_product} to bound moments of dot products and using $r^2 = O(\log n), d = (\log n)^{1 + \Omega(1)}, t > (\log n)^{1 - o(1)}$ to bound tail of the sum, we conclude
    \begin{gather*}
        \sum_{l=t+1}^{\infty} \frac{2^{l-t}\E\<v, k_j>^{2l}}{(l!)^2} \leq \sum_{l=t+1}^{\infty} \frac{2^{2l-t}r^{4l}}{d^{l}l!} \leq \frac{(2r)^{4t}}{d^t t!}.
    \end{gather*}
    Multiplying with $n$ gives the variance of $\expker_{>t}(\{v_1, \cdots, v_n\}, k_j)$  by independence.
    % Using the choice of $t$ we conclude
    % \begin{gather*}
    %     \Var{\expker_{>t}(Y, x)} \leq \frac{n}{d^t t!} = \frac{n}{d^{\floor{\frac{\log n - \log\log n}{\log d}}}}\frac{1}{t!} \leq \frac{d \log n}{t!} = \frac{\log^{\Theta(1)} n}{\log^{\omega(1)} n} = o(1).
    % \end{gather*}
\end{proof}
Now we choose the cutoff parameter $t$:
% such that \Cref{lem:taylor_embed_memory,lem:variance_bound} work and:
% \begin{itemize}
%     \item Alice sends $o(n)$ side information bits, i.e. $\frac{(d+t)^t}{t!}\log^2 n = o(n)$.
%     \item Noise is negligible, i.e. $\Var{\expker_{>t}(K_{\bx} \setminus \{k_i\}, k_i)} = O\qty(n\frac{r^{2t}}{d^t t!}) = o(e^{2r^2})$.
% \end{itemize}
let's first pick $t_{max} = o(\log n)$ such that $t_{\max} > \frac{\log n}{\log d/t_{max}}$. It already satisfies assumptions $t < r^2 = \Theta(\log n)$ and $t < d$ (which are the premises of \cref{lem:taylor_embed_memory} and \cref{lem:variance_bound}), but $\frac{(d+t)^t}{t!} = \omega(n)$, i.e. the bit complexity is too large. Thus, we start decreasing $t$ from $t_{max}$ until $\frac{(d+t)^t}{t!} > \frac{n}{\log ^3 n}$ holds. When we stop, it means that 
\begin{align*}
    \frac{(d+t)^t}{t!}\log ^2 n < \frac{n}{\log n} = o(n).
\end{align*}
Since it is the moment we stopped and $t = o(d)$,
\begin{align*}
    \frac{(d + t + 1)^{t+1}}{{(t+1)}!} > \frac{n}{\log^3 n} \quad\Rightarrow\quad \frac{d^t}{t!} > \frac{n}{2^t \cdot d \log^3 n} \quad\Rightarrow\quad \begin{cases}
        \frac{n}{d^t} = \frac{2^t(\log n)^{O(1)}}{t!} \\
        t > \frac{\log n}{\log d} = (\log n)^{1 - o(1)}
    \end{cases}
\end{align*}
Therefore, the variance is 
\begin{gather*}
    \Var{\expker_{>t}(K_{\bx} \setminus \{k_j\}, k_j)} = O\qty(n\frac{(2r)^{4t}}{d^t t!}) = \frac{2^{5t} r^{4t}(\log n)^{O(1)}}{(t!)^2} = \\
    \frac{(\log n)^{2t(1 + o(1))}}{(\log n)^{2t(1 - o(1))}} = (\log n)^{o(t)} < e^{o(\frac{\log n}{\log d - \log t})\log \log n} = n^{o(\frac{\log \log n)}{\Omega(\log \log n)})} = n^{o(1)} = o(e^{2r^2}).
\end{gather*}

It remains to argue that the approximation error
% This shows that for such $t$ and kernel $\expker_{>t}$ the contribution of one point $\expker_{>t}(k_i, k_i) > \Omega(e^{r^2})$ typically stands out a lot from the noise in contribution of all other points 
% \[
%     \sqrt{\Var{\expker_{>t}(K_{\bx} \setminus \{k_i\}, k_i)}} = o(e^{r^2}).
% \] In other words, the following holds with high probability (from Chebyshev's inequality; we use $\mu = \E_{K} [\expker_{>t}(K_{\bx}, k_i) \mid \bx_i = 0]$):

% \begin{center}
% \begin{tikzpicture}[
%     node distance=0.5cm,
%     every node/.style={align=center, minimum width=0.5cm, minimum height=1cm},
%     arrow/.style={->, thick}
% ]

% % Center node (question)
% \node (question) {Does $k_i \in K_{\bx}$? Or, equivalently, does $\bx_i = 1$?};

% % Left and right nodes
% \node (left) [below left=of question, xshift=-0.1cm] {Yes};
% \node (right) [below right=of question, xshift=0.1cm] {No};

% % Bottom nodes
% \node (leftbottom) [below=of left] {$\expker_{>t}(K_{\bx}, k_i) > \mu + \frac{1}{2}e^{r^2}$};
% \node (rightbottom) [below=of right] {$\expker_{>t}(K_{\bx}, k_i) < \mu + \frac{1}{100}e^{r^2}$};

% % Arrows
% \draw[arrow] (question) -- (left);
% \draw[arrow] (question) -- (right);
% \draw[arrow] (left) -- (leftbottom);
% \draw[arrow] (right) -- (rightbottom);

% \end{tikzpicture}
% \end{center}

% Therefore, Bob will make his decision based on whether $\expker_{>t}(K_{\bx}, k_i)$ is bigger than $\mu + \frac{1}{4}e^{r^2}$. Since he is only observing $\widehat{\expker}(K_{\bx}, k_i) - \widehat{\expker}_{\leq t}(K_{\bx}, k_i)$, we need to  bound the noise
\begin{equation*}    
     2\eps \expker(K_{\bx}, k_j) \leq \frac{1}{8}e^{r^2}.
\end{equation*}
Since the dataset is $K_{\bx}$ sampled uniformly at random from a sphere, with high probability, we have $\expker(K_{\bx}, k_j) \leq e^{r^2} + ne^{O(r^2 \sqrt{\frac{\log n}{d}})}$ by \Cref{lem:rand_dataset_kde}, therefore,
\[\eps \cdot \expker(K_{\bx}, k_j) \leq \eps (e^{r^2} + ne^{O(r^2 \sqrt{\frac{\log n}{d}})}) < \frac{1}{16}e^{r^2}\] as desired,  where for the second inequality we used
$n < \frac{e^{r^2 \p{1 - C\sqrt{\frac{\log n}{d}}}}}{16\eps}.$ Summarizing the above, we show that for our choice of $t$, the contribution  of the signal term $\expker_{> t}(k_j, k_j)$ typically stands out from the noise induced by the contribution of $\expker_{> t}(k, k_j)$ for $k \neq k_j$ in $K_\bx$  and the approximation error. And the bit complexity of the sketch is $o(n)$.  This concludes our proof of the reduction to the \texttt{INDEX} problem.
\subsection{Putting it together}
\label{sec:low_temp_putting_together}

\begin{proof}[Proof of \Cref{th:big_r_ub_lb}]
    The upper bound is shown in \Cref{sec:ub_low_temp} using pseudorandomification \\(\Cref{lem:ps_randomification}) and better compression for pseudo-random parts (\Cref{lem:ps_rand_coreset,lem:small_r_coreset}), resulting in compression procedure \Compress{}. The construction is maintained in a stream using weighted Merge-Reduce algorithm (\Cref{lem:merge-reduce-w}). Final block size $b$ to solve Softmax KDE problem is chosen as $\frac{e^{r^2 (1 + o(1))}}{\eps}n^{o(1)} = \frac{e^{r^2 (1 + o(1))}}{\eps}$ using $r^2 = \Omega(\log n)$. Therefore the final space complexity is $\tilde{O}(\frac{e^{r^2 (1 + o(1))}}{\eps})$.

    The lower bound is given in \Cref{sec:lb_low_temp}.
    As shown by \Cref{lem:reduction_big_r}, there is a $S_{n, d, r, \eps} + o(n)$ size protocol that solves $\texttt{INDEX}_{2n, n}$ when $n \leq \tilde{O}( \exp(r^2 \p{1 - o(1)})/\eps)$, where $S_{n, d, r, \eps}$ denotes the space complexity of any algorithm solving Softmax KDE problem. Combined with the bound $R^{\text{pub}, \to}(\texttt{INDEX}_{2n, n}) = \Omega(n)$ from \cref{thm:indexing} it implies the lower bound of \cref{th:big_r_ub_lb}.
\end{proof}

\section{From Softmax KDE to Attention}
In this section we provide formal results for the Attention problem (\Cref{def:streaming-attn}). First, we restate our main theorems.

\begin{theorem}[Formal version of \Cref{remark:const_r_attn}]
    \label{th:const_r_ub_lb_attn}
    Assume that $r = \log^{o(1)} n$, $\log n \leq d \leq \log^a n$ for some constant $a > 1$ and $n^{-100} < \eps < n^{-c}$ for any constant $c > 0$. Then the Attention problem can be solved using
    \begin{align*}
        \tilde{O}\qty(\qty(\frac{1}{\eps})^{1 - \frac{1}{a} + o(1)})
    \end{align*}
    memory. Furthermore, this bound is tight for any $d \geq \log^a n$, $a > 2$, i.e. there is a
    \begin{align*}
        \tilde{\Omega}\qty(\min\qty{n, \qty(\frac{1}{\eps})^{1 - \frac{1}{a} - o(1)}})
    \end{align*}
    lower bound for the streaming space complexity.
\end{theorem}

\begin{theorem}[Formal version of \Cref{remark:big_r_attn}]
    \label{th:big_r_ub_lb_attn}
    
   Assume that $r^2 = \Omega(\log n)$, i.e. $e^{r^2} = n^{\Omega(1)}$. Then the Attention problem can be solved using
    \begin{align*}
        \tilde{O}\qty(\frac{e^{r^2(1 + o(1))}}{\eps})
    \end{align*}
    memory. Furthermore, this bound is almost tight for any $d = (\log n)^{1 + \Omega(1)}$ and $r^2 < \frac{\log n}{2}$, i.e. there is a
    \begin{align*}
        \tilde{\Omega}\qty(\min\qty{n, \frac{e^{r^2(1 - o(1))}}{\eps}})
    \end{align*}
    lower bound for streaming space complexity.
\end{theorem}

We prove quite general upper bound reduction to Softmax KDE in \Cref{sec:attn_ub} below, lower bound reduction in \Cref{sec:attn_lb} below. These results are put together in \Cref{sec:attn_putting_together} to obtain proofs of \Cref{th:const_r_ub_lb_attn,th:big_r_ub_lb_attn}.

\subsection{Upper bounds for Attention}
\label{sec:attn_ub}
In this section we will show how one can adapt our data structure that solves Softmax KDE problem to solve Attention problem. First, we slightly generalize the setup.

\begin{definition}[Weighted Softmax KDE problem]\label{def:weighted-kde}
The algorithm receives key-weight pairs 
\\$(k_1, w(k_1)) \ldots, (k_n, w(k_n))\in \R^d \times \R$ in a stream and at every point $j = 1,\ldots, n$, 
    given any fixed $q\in \R^d$,  outputs an estimate $\widehat \expker(K_{j}, w_j, q)$ of $\expker(K_{j}, w_j, q)$ that satisfies
    \begin{gather}
        |\widehat{\expker}(K_{j}, w_j, q) - \expker(K_{j}, w_j, q)| \leq \eps \cdot \expker(K_{j},  q)
    \end{gather}
    with probability at least $1-\delta$ for some $\delta\in (0, 1)$ for any fixed query $q$.  Here $K_{j}=(k_1, \ldots, k_{j})$ is the sequence of keys received by time $j$,  $w_{j}=(w(k_1), \ldots, w(k_{j}))$ is the sequence of weights received by time $j$, $|w(k)| \leq 1$ for all $k \in K$, and 
    \begin{gather*}
        \expker(K, w, q) := \sum\limits_{k\in K} w(k)\expker(k, q).
    \end{gather*}
\end{definition}

 Weighted  Softmax  KDE can be solved using our current techniques for non-weighted  Softmax,  with the same guarantees and very little adjustments. Indeed, the only places where the algorithms and proofs use the KDE structure directly are \Cref{lem:embed_memory} and \cref{lem:simple_coreset}. Elsewhere, they depend on the KDE structure  only indirectly through these lemmas. Therefore, if we can show weighted analogues of \Cref{lem:embed_memory} and \cref{lem:simple_coreset}, the rest of the framework carries over unchanged. These weighted analogues are stated next:
 
 \begin{lemma}[Weighted version of \Cref{lem:embed_memory}]
\label{lem:weighted_embed_memory}
    Let $\gamma(k, q) := \sum_{l = 0}^{t} a_l \<k, q>^l$ for some $t$ and $a_l$. Then there exist embeddings $\psi_1, \psi_2$ such that for any $K \subset \R^d$  with $\|k\|_2  \leq r$ for all $k \in K$, any $q \in \R^d$ with $\|q\|_2 \leq r$ , and any weights $\{w(k)\}_{k \in K}$ with $|w(k)| \leq 1$ for all $k \in K$, we guarantee
    \[\qty|\Bigl \langle \sum_{k \in K} w(k) \psi_1(k), \psi_2(q) \Bigr \rangle - \sum_{k \in K}  w(k)\gamma(k, q)| < \delta,\] and the weighted aggregate $ \sum_{k \in K} w(k)\psi_1(k)$ can be stored in $O\Bigl(\frac{(d+t)^{t}}{t!}\log\frac{|K|\cdot (d+t)^{t} \cdot\max_l |a_l|\cdot \max\{r, 1\}^t}{\delta}\Bigr)$ bits of space.
\end{lemma}

\begin{lemma}[Weighted version of \Cref{lem:simple_coreset}, specialized to $\gamma = \expker$]
    \label{lem:weighted_streaming_coreset}
    Let $U \times V\subset   \R^d \times \R^d $ be a predefined subset of the domain. There is a randomized algorithm $\BaseCompress'$ which receives as inputs $K \subset U$ and any weights $\{w(k)\}_{k\in K}$ with $|w(k)| \leq 1$ . $\BaseCompress'(K, \{w(k)\}_{k\in K})$ outputs $\widetilde{K} \subset K$ such that $|\widetilde{K}| \leq |K|/2$ and for any $q \in V$, with probability $1 - \delta$,
           \[\left|\expker(K, w, q) - 2\expker(\widetilde{K}, w, q)\right| \leq O\left(\sqrt{\max_{u \in U}\expker(u, u)\cdot \max_{v \in V}\expker(v, v)}\cdot \log (|K|/\delta)\right).\]
    $\BaseCompress'(K,w)$ runs in $O(|K|^2 \cdot T)$ time, where $T$ is the maximum time to access $\expker(k,q)$.

\end{lemma}

 \begin{claim}
Assume that $r = \log^{o(1)} n$, $\log n \leq d \leq \log^a n$ for some constant $a > 1$ and $\eps < n^{-c}$ for any constant $c > 0$. The Weighted Softmax KDE problem (\cref{def:weighted-kde}) can be solved using
    \begin{align*}
        \tilde{O}\qty(\qty(\frac{1}{\eps})^{1 - \frac{1}{a} + o(1)})
    \end{align*}
    memory. If, on the other hand, $r^2 = \Omega(\log n)$, equivalently
$e^{r^2} = n^{\Omega(1)}$, then the Weighted Softmax KDE problem
\eqref{eq:streaming-kde} can be solved using
    \begin{align*}
    \tilde{O}\qty(\frac{e^{r^2(1 + o(1))}}{\eps})
    \end{align*}
    memory.
 \end{claim}
 \begin{proof}
     % Lower bounds obviously hold since weighted Softmax KDE problem is a strict generalization.

     Compactly, this claim says that the upper bounds in \Cref{th:const_r_ub_lb,th:big_r_ub_lb} also hold for weighted Softmax KDE problem. This is achieved through the same algorithms (embeddings and coresets) as for non-weighted KDE, using \Cref{lem:weighted_embed_memory,lem:weighted_streaming_coreset} instead of \Cref{lem:embed_memory,lem:simple_coreset} as subroutines.
 \end{proof}

 In the rest of this section, we first prove \cref{lem:weighted_embed_memory} and \cref{lem:weighted_streaming_coreset}, which imply the existence of streaming algorithms for the Weighted Softmax KDE problem. After that, we show a reduction from the Attention problem to the Weighted Softmax KDE problem. Combining these two parts concludes the proof of upper bounds for Attention.

\begin{proof}[Proof (of  \cref{lem:weighted_embed_memory})]
  As in the proof of \cref{lem:embed_memory}, we let 
    \[
        \psi_1(k) := \qty(a_l \cdot k^{{\otimes l}})_{l = \overline{0..t}} \quad \text{ and } \quad \psi_2(q) := \qty(q^{{\otimes l}})_{l= \overline{0..t}}.
    \]
    By linearity, in exact arithmetic we have:
    \[
        \Bigl \langle \sum_{k \in K} w(k) \psi_1(k), \psi_2(q) \Bigr \rangle = \sum_{k \in K}  w(k)\gamma(k, q).
    \]
    Just as in \cref{lem:embed_memory}, we first round
$\sum_{k \in K} w(k) \psi_1(k)$ to a finite number of bits and store the
result as $z$. Regarding the bit complexity of $z$, every coordinate of $\qty(k^{{\otimes l}})_{l= \overline{0..t}}$ and
$\qty(q^{{\otimes l}})_{l= \overline{0..t}}$ is still bounded by $\max\{r, 1\}^t$; thus every coordinate
of $\sum_{k \in K} w(k) \psi_1(k)$ has absolute value at most
    \[\sum_{k \in K} |w(k)| \max_l |a_l| \max\{r, 1\}^t \leq |K|\max_l |a_l| \max\{r, 1\}^t.\]
 If $\| z - \sum_{k \in K} w(k) \psi_1(k)\|_\infty \leq \eta$, then
    \begin{align*}
        \Bigl |\Bigl \langle z,  \psi_2(q) \Bigr \rangle - \Bigl \langle \sum_{k \in K} w(k) \psi_1(k),  \psi_2(q) \Bigr \rangle  \Bigr |  &\leq m \cdot \| z - \sum_{k \in K} w(k) \psi_1(k)\|_\infty \cdot  \| \psi_2(q)\|_\infty \\
        &\leq m \cdot \eta \cdot  \max\{r, 1\}^t 
    \end{align*}
    where $m$ is the target dimension of $\psi_1, \psi_2$ , which is again $O\left(\frac{(d+t)^t}{t!}\right)$, as in \cref{lem:embed_memory}. Therefore, to guarantee total error $\delta$, we can pick $\eta = \frac{\delta}{m\cdot  \max\{r, 1\}^t }$. The bit complexity to store each coordinate to this precision is $O(\log \frac{ |K|\max_l |a_l| \max\{r, 1\}^t \cdot (m\cdot \max\{r, 1\}^t)}{\delta})$, and multiplying with $m$ again gives the claimed space complexity.
\end{proof}

\begin{proof}[Proof (of  \cref{lem:weighted_streaming_coreset})]
    As in the proof of \Cref{lem:simple_coreset}, we again reduce the problem to
the inner-product setting, where we use a randomized algorithm $\mathcal{A}$
for the inner-product kernel $\gamma_0(k,q) := \langle k,q\rangle$
(see \cref{thm:compress_l2} for details). Since $\expker$ is positive definite, there exists a feature map $\psi:\R^d\to H$ into a Hilbert space $H$ such that
\[
\expker(x,y)=\langle \psi(x),\psi(y)\rangle_H
\qquad\text{for all }x,y\in\R^d.
\]

We consider the subspace $\mathcal{S}$ spanned  by $\{\psi(k): k \in K\}$ and $\psi(q)$. Choosing an orthonormal
basis of  $\mathcal{S}$, we may identify  $\mathcal{S}$ isometrically with $\R^m$ for some
dimension $m \leq |K|+1$. Under this identification, each $\psi(k)$ and $\psi(q)$ becomes a vector in $\R^m$, and all inner products among these vectors are preserved, i.e.,
\[
\expker(k,q)=\langle \psi(k),\psi(q)\rangle_H
=\langle \pi\psi(k),\pi\psi(q)\rangle_{\mathbb R^m} = \gamma_0(\pi\psi(k),\pi\psi(q))
\]
where $\pi: \mathcal{S}\rightarrow\R^m$ is an isometric map defined by the chosen orthonormal basis. Furthermore,
\[
w(k)\expker(k,q)=w(k)\langle \psi(k),\psi(q)\rangle_H
=\langle w(k)\pi\psi(k),\pi\psi(q)\rangle_{\mathbb R^m} = \gamma_0(w(k)\pi\psi(k),\pi\psi(q)).
\]
Therefore, we can apply \cref{thm:compress_l2} to the inner-product kernel $\gamma_0$, with dataset $\{w(k)\pi\psi(k)\}_{k \in K}$, and query $\pi\psi(q)$. The algorithm returns a subset $\tilde{K} \subset K$ of size at most $|K|/2$ and with probability $1-\delta$,
\begin{align*}
    &\Bigl| \sum_{k \in K} \langle w(k)\pi\psi(k),\pi\psi(q)\rangle - 2\sum_{k \in \tilde{K}} \langle w(k)\pi\psi(k),\pi\psi(q)\rangle \Bigr| \\
    &\hspace{4em}\leq  O\left(\sqrt{\max_{u \in U'}\langle u, u\rangle \cdot \max_{v \in V'}\langle v, v\rangle}\cdot \log (|K|/\delta)\right),
\end{align*} 
where $U' := \{w(u) \pi \psi(u) : u \in U, |w(u)| \leq 1\}$ and $V' := \{ \pi \psi(v) : v\in V\} $. The left-hand side is exactly $|\expker(K, w, q) - 2\expker(\widetilde{K}, w, q)|$. For the right-hand side, since $|w(u)| \leq 1$, 
\[\max_{u \in  U'} \langle u, u\rangle = \max_{u \in  U'} w(u)^2\langle u, u\rangle \leq \expker (u, u) \]
and similarly for $v$. Hence the right-hand side is at most the claimed error bound. Finally,  \cref{thm:compress_l2} only requires oracle access to inner products in the transformed space. These can be computed using the weights together with oracle access to $\expker$. Therefore we do not need to compute the embedding explicitly, and the runtime remains $O(|K|^2\cdot T)$  as in \cref{thm:compress_l2}.
\end{proof}

Finally, we show a reduction from the Attention problem to the Weighted Softmax KDE problem.
\begin{lemma}
\label{l:tokde}
    If there exists a streaming algorithm that solves the weighted Softmax KDE problem (\cref{def:weighted-kde}) with parameters $n, d, \eps, r$ using $S(n, d, \eps, r)$ space, then there  exists a streaming algorithm that solves the Attention problem (\cref{def:streaming-attn}) with parameters $n, d, \tilde{O}(\eps), r$  using $\tilde{O}(S(n, d, \eps, r))$ space.
\end{lemma}

\begin{proof}
Recall that in the Attention problem, at every time step of the stream $l$ the algorithm needs to output $\hat z^l$ such that 
    $$
    \|\hat z^l - \Attn(K_{l}, q_l, V_{l})\|_2\leq \epsilon \cdot \|\Softmax(K_{l}, q_l)\|_2 \|V_{l}\|_F.
    $$
    
Our algorithm amounts to running several instances of streaming KDE, which has per-step approximation guarantees. Therefore, we can fix one time step $l \leq n$ and drop the indexes $l$ until the rest of the proof.   Instead of bounding the $\ell_2$  norm of the error we will bound its $\ell_1$ norm:
    \begin{align*}
        \sum_{j=1}^{d} |\hat{z}_j - \Attn(K, q, V)_j| = \norm{\hat{z} - \Attn(K, q, V)}_{1}.
    \end{align*}
     Since we assume poly-log dimension $d = \tilde{O}(1)$,
    \begin{gather*}
        \frac{1}{\sqrt{d}}\norm{\hat{z} - \Attn(K, q, V)}_{1} \leq \norm{\hat{z} - \Attn(K, q, V)}_2 \leq \norm{\hat{z} - \Attn(K, q, V)}_{1}.
    \end{gather*}
    Therefore, we lose at most a $\sqrt d = \tilde O(1)$ factor by switching to
$\norm{\hat{z} - \Attn(K, q, V)}_{1}$. We also split RHS in the error guarantee by coordinates:
     \begin{gather*}
        \norm{V}_F = \sqrt{\sum_{j = 1}^{d} \sum_{i=1}^{l} V_{i, j}^2} \geq \frac{1}{\sqrt{d}}\sum_{j = 1}^{d}\sqrt{\sum_{i=1}^{l} V_{i, j}^2} = \tilde{\Omega}(1)\sum_{j = 1}^{d}\sqrt{\sum_{i=1}^{l} V_{i, j}^2}.
    \end{gather*}
    Therefore, to solve Attention problem with $\tilde{O}(\eps)$ precision it is enough to provide an estimator such that 
    \begin{align*}
        \norm{\hat{z} - \Attn(K, q, V)}_{1} \leq \eps \cdot \|\Softmax(K, q)\|_2 \cdot \sum_{j = 1}^{d}\sqrt{\sum_{i=1}^{l} V_{i, j}^2}.
    \end{align*}
    To do this we approximate each coordinate individually using $d = \tilde{O}(1)$ separate data structures with guarantees
    \begin{align}
        \label{eq:attn_coord}
        \qty|\hat{z}_j - \frac{\sum\limits_{i = 1}^{l} e^{\<k_i, q>}V_{i, j}}{\sum\limits_{i = 1}^{l} e^{\<k_i, q>}}| \leq O(\eps) \cdot \|\Softmax(K, q)\|_2 \cdot \sqrt{\sum_{i=1}^{l} V_{i, j}^2}.
    \end{align}
    Stacking them together will produce the required estimate $\hat{z}$. We fix one coordinate $j$ in what follows. To provide the guarantee of~\eqref{eq:attn_coord} we will give two estimators~-- $\hat{N}$ for the numerator and $\hat{D}$ for the denominator:
    \begin{align}
        \label{eq:attn-num}
        \qty|\hat{N} - \sum\limits_{i = 1}^{l} e^{\<k_i, q>}V_{i, j}| \leq
        \eps \cdot \sum\limits_{i = 1}^{l} e^{\<k_i, q>}|V_{i, j}|
    \end{align}
    and
    \begin{align*}
        \qty|\hat{D} - \sum\limits_{i = 1}^{l} e^{\<k_i, q>}| \leq \eps \cdot \sum\limits_{i = 1}^{l} e^{\<k_i, q>}.
    \end{align*}
    It is easy to see that we can construct $\hat{D}$ the same way as $\hat{N}$ by letting $V_{i, j} = 1$, so we only focus on constructing $\hat{N}$. But first, let's check that we do recover guarantee \eqref{eq:attn_coord} for $\hat{z} = \frac{\hat{N}}{\hat{D}}$:
    \begin{align}
        \label{eq:attn-softmax-norm}
        \notag
        \qty|\frac{\hat{N}}{\hat{D}} - \frac{\sum\limits_{i = 1}^{l} e^{\<k_i, q>}V_{i, j}}{\expker(K, q)}| \leq
        \qty|\frac{\hat{N}}{\hat{D}} - \frac{\hat{N}}{\expker(K, q)}| &+ \qty|\frac{\hat{N}}{\expker(K, q)} - \frac{\sum\limits_{i = 1}^{l} e^{\<k_i, q>}V_{i, j}}{\expker(K, q)}| \leq 
        \\ \notag
        |\hat{N}|\cdot\qty|\frac{\hat{D} - \expker(K, q)}{\hat{D} \cdot \expker(K, q)}| &+ \frac{1}{\expker(K, q)}\cdot\qty|\hat{N} - \sum\limits_{i = 1}^{l} e^{\<k_i, q>}V_{i, j}| \leq 
        \\ \notag
        (1 + \eps)\left(\sum\limits_{i = 1}^{l} e^{\<k_i, q>}|V_{i, j}|\right)\cdot\frac{\eps}{(1 - \eps)\expker(K, q)} &+ \frac{\eps \cdot \sum\limits_{i = 1}^{l} e^{\<k_i, q>}|V_{i, j}|}{\expker(K, q)} \leq
        \\ 
        O(\eps) \cdot \frac{\sum\limits_{i = 1}^{l} e^{\<k_i, q>}|V_{i, j}|}{\expker(K, q)} &\leq  O(\eps) \cdot \frac{\sqrt{\sum_{i=1}^{l} e^{2\<k_i, q>} \sum_{i=1}^{l}V_{i, j}^2}}{\expker(K, q)}
        \\ \notag
        & = O(\eps)\|\Softmax(K, q)\|_2  \cdot \sqrt{\sum_{i=1}^{l} V_{i, j}^2}.
    \end{align}
    Now we only need to show how to build an estimator with guarantee $\eqref{eq:attn-num}$:
     \begin{align*}
        \qty|\hat{N} - \sum\limits_{i = 1}^{l} e^{\<k_i, q>}V_{i, j}| \leq
        \eps \cdot \sum\limits_{i = 1}^{l} e^{\<k_i, q>}|V_{i, j}|.
    \end{align*}
    Note that if all $|V_{i, j}|$ have about the same value, i.e. if all $|V_{i, j}|$ lie in $[\alpha; 2\alpha]$ for some scaling parameter $\alpha$, then this estimator is given by weighted KDE problem (\Cref{def:weighted-kde}) with weights $w(p_i) = \frac{V_i}{2\alpha}$. To see this assume that we are given $\widehat{\expker}$ such that (from \Cref{def:weighted-kde})
    \begin{gather*}
        |\widehat{\expker}(K, w, q) - \expker(K, w, q)| \leq \eps \cdot \expker(K, q)
    \end{gather*}
    Then, since $\expker(K, w, q) = \frac{1}{2\alpha}\sum\limits_{i = 1}^{l} e^{\<k_i, q>}V_{i, j}$,
    \begin{gather*}
        |2\alpha \cdot \widehat{\expker}(K, w, q) - \sum\limits_{i = 1}^{l} e^{\<k_i, q>}V_{i, j}| \leq \eps \cdot 2\alpha \cdot \expker(K, q) \leq 2\eps \cdot \sum\limits_{i = 1}^{l} e^{\<k_i, q>}|V_{i, j}|.
    \end{gather*}
    Therefore, we can partition $P$ into $O(\log n)$ buckets according to the
magnitudes of the values $V_{i,j}$, and approximate each bucket separately,
incurring only a $\tilde{O}(1)$ factor overhead in both error and space. It
suffices to use $O(\log n)$ buckets because we assume that $e^{r^2}$ and
$1/\eps$ are at most $\poly(n)$; hence buckets corresponding to 
\[
    |V_{i,j}| < \frac{\max_{i \leq l} |V_{i,j}|}{\poly(n)}
\]
can be ignored, as their contribution is negligible compared to $\exp(K, q) \geq e^{-r^2}\cdot \max_i|V_{i, j}|$. This concludes the
reduction.
\end{proof} 
\subsection{Lower bounds for Attention}
\label{sec:attn_lb}
% \piotr{I moved this section after the module that discusses upper and lower bound for KDE.}

% \bor{So far we provide a hard instance for $r^2 < \frac{\log n}{2}$ and $\|\Softmax\| \sim \frac{1}{\sqrt{n}}$.}

In this section, we show that our hard instances for Softmax KDE problem translate to hard instances for Attention problem with slight modifications.
We prove the reduction specifically to the type of random KDE instances used in our lower bound instances, which we generalize below.

\begin{definition}[Random KDE Instance]
\label{def:random-kde-instance}
Let $\bx \in \{0,1\}^{2n}$ be a binary vector with exactly $n$ ones. Let $K \subset \R^d$ be a set of $2n$ points drawn i.i.d.\ from a rotationally invariant distribution, independently of $\bx$. Let $K_\bx$ be the set of points indexed by the ones of $\bx$.
A \emph{random KDE instance} is an instance of the streaming KDE problem (\cref{def:streaming-kde}) in which the keys $k_1, \ldots, k_n$ are the points in $K$ and the query $q$ is chosen from $K$. Vector $q$ is chosen to be one of the keys $k_i \in K$; the position in the dataset $i$ is independent of the randomness of $K$.
\end{definition}

\begin{lemma}
\label{l:fromkde}
    Assume $r^2 < \frac{\log n}{2}$. If there exists a streaming algorithm that solves the attention problem (\cref{def:streaming-attn}) with parameters $2n, d, \eps, r$ using $S(2n, d, \eps, r)$ space, then there  exists a streaming algorithm that solves the streaming KDE problem on random instances (\cref{def:random-kde-instance}) with parameters $n, d, \eps', r$  using $S(2n, d, \eps, r)$ space, where $\eps' = \eps \cdot O\left(e^{O(r^2\sqrt{\frac{\log n}{d}})}\right)$.
\end{lemma}

\begin{proof}
    Given a random KDE instance (\cref{def:random-kde-instance}) we choose key vectors for Attention instance to be $K$, i.e. $2n$ vectors chosen independently of $\bx$. Next we introduce value vectors.
    For each $i \in [n]$, we define the corresponding value vector $v_i = \bx_i \in \{0,1\}$ (we treat $v_i$ as $1$-dimensional, padding it with zeros if needed).

    With our construction, the attention output is a scalar, whose numerator $\sum_{i = 1}^{n} e^{\langle k_i, q \rangle} v_i$ equals $\expker(K_{\bx}, q)$, and the denominator $\sum_{i \in [n]} e^{\langle k_i, q \rangle}$ equals $\expker(K, q)$. The approximation guarantee from \cref{def:streaming-attn} ensures that the output $\hat{z}$ satisfies
    \[
        \qty|\hat{z} - \frac{\expker(K_{\bx}, q)}{\expker(K, q)}| \leq \eps \cdot \|\Softmax(K, q)\|_2 \cdot \|V\|_F.
    \]
    Note that $\|V\|_F = \sqrt{|K_{\bx}|} = \sqrt{n}$ by construction.
    We can bound $\|\Softmax(K, q)\|_2$ as $ O\left(\frac{e^{O(r^2\sqrt{\frac{\log n}{d}})}}{\sqrt{n}}\right)$ by \Cref{lem:rand_dataset_kde} (using randomness of the KDE instance) and our assumption $e^{2r^2} < n$:
    \begin{align*}
        \|\Softmax(K, q)\|^2_2 = \frac{\sum e^{2\langle k_i, q \rangle}}{(\sum e^{\langle k_i, q \rangle})^2} \stackrel{(\Cref{lem:rand_dataset_kde})} = \frac{e^{2r^2} + n\cdot e^{O(r^2\sqrt{\frac{\log n}{d}})}}{\Omega(n^2)} = O\qty(\frac{e^{O(r^2\sqrt{\frac{\log n}{d}})}}{n})
    \end{align*}
    i.e. $\|\Softmax(K, q)\|_2 \cdot \|V\|_F = e^{O(r^2\sqrt{\frac{\log n}{d}})}$. This guarantee means that $\widehat{\expker}(K_{\bx}, q) = \hat{z} \cdot \expker(K, q)$ satisfies
    \[
        \qty|\widehat{\expker}(K_{\bx}, q) - \expker(K_{\bx}, q)| \leq O(\eps) \cdot \expker(K, q) \cdot e^{O(r^2\sqrt{\frac{\log n}{d}})} = \eps' \cdot \expker(K_{\bx}, q).
    \]
Here, the last step follows from \Cref{lem:rand_dataset_kde}. By \Cref{lem:rand_dataset_kde}, since $|K| = 2n$ and $|K_{\bx}| = n$
\[\Omega(n) \leq \sum_{k \in K_{\bx}} e^{\langle k, q \rangle} =  \exp(K_{\bx}, q)\leq n\cdot e^{O\left( r^2\sqrt{\frac{\log n}{d}}\right)} + e^{r^2},\]

\[\Omega(n) \leq \sum_{k \in K} e^{\langle k, q \rangle} =  \exp(K, q)\leq 2n\cdot e^{O\left( r^2\sqrt{\frac{\log n}{d}}\right)} + e^{r^2}. \] Because $r^2 < \log n/2$, $e^{2r^2} < n$. Therefore,
\[
\expker(K, q) \leq O\left(n\cdot e^{O\left( r^2\sqrt{\frac{\log n}{d}}\right)}\right)  \leq O\qty(\expker(K_\bx, q) \cdot e^{O\left(r^2\sqrt{\frac{\log n}{d}}\right)}).
\]
    Therefore we can solve random KDE instance with parameters $n, d, \eps', r$ using $S(2n, d, \eps, r)$ space.
\end{proof}

\subsection{Putting it together}
\label{sec:attn_putting_together}

\begin{proof}[Proof of \Cref{th:const_r_ub_lb_attn}]
The upper bound is given by reduction to \Cref{th:const_r_ub_lb} in \Cref{sec:attn_ub} with just a $\Tilde{O}(1)$ factor loss in $\eps$, not affecting final bound
\begin{align*}
    \tilde{O}\qty(\qty(\frac{1}{\eps})^{1 - \frac{1}{a} + o(1)})
\end{align*}
after rescaling.

The lower bound reduction is shown in \Cref{sec:attn_lb}. Combining \Cref{l:fromkde} with the lower bound proof of \Cref{th:const_r_ub_lb}, namely \Cref{lem:reduction}, we get a 
\begin{align*}
    \tilde{\Omega}\qty(\qty(\frac{1}{\eps'})^{1 - \frac{1}{a} + o(1)})
\end{align*}
lower bound, where $\eps' = \eps \cdot O(e^{O(r^2\sqrt{\frac{\log n}{d}})})$. Since we are in the regime $d > \log^2 n, r^2 = (\log n)^{o(1)}$, $e^{O(r^2\sqrt{\frac{\log n}{d}})} = \Theta(1)$, finishing the proof.
\end{proof}

\begin{proof}[Proof of \Cref{th:big_r_ub_lb_attn}]
The upper bound is given by reduction to \Cref{th:big_r_ub_lb} in \Cref{sec:attn_ub} with just a $\Tilde{O}(1)$ factor loss in $\eps$, not affecting final bound
\begin{align*}
    \tilde{O}\qty(\frac{e^{r^2(1 + o(1))}}{\eps})
\end{align*}
after rescaling.

The lower bound reduction is shown in \Cref{sec:attn_lb}. Combining \Cref{l:fromkde} (where we crucially use $r^2 < \frac{\log n}{2}$) with the lower bound proof of \Cref{th:big_r_ub_lb}, namely \Cref{lem:reduction_big_r}, we get a 
\begin{align*}
        \tilde{\Omega}\qty(\min\qty{n, \frac{e^{r^2(1 - o(1))}}{\eps'}})
    \end{align*}
lower bound, where $\eps' = \eps \cdot O(e^{O(r^2\sqrt{\frac{\log n}{d}})})$. Since we are in the regime $d = (\log n)^{1 + \Omega(1)}$, we can bound $e^{O(r^2\sqrt{\frac{\log n}{d}})} = e^{o(r^2)}$,  finishing the proof.
\end{proof}

\newpage

\appendix
\crefalias{section}{appendix}
\section{Performance guarantees of \Partition{} (\Cref{alg:partition})}\label{sec:ub_lowt_proofs}
In this section, we provide the omitted proofs from \Cref{sec:ub_low_temp}. We start by proving \Cref{cor:certify-ball}, restated below:
\certifyball*

\Cref{cor:certify-ball} is very similar to Lemma 58 from \cite{CKNS20} (stated below as \Cref{lemma:certify}), except that the prior work applies only to datasets lying on a sphere. To prove \Cref{cor:certify-ball}, we reduce to the prior lemma by radially projecting a dataset contained in a ball onto its enclosing sphere.

\begin{lemma}[\cite{CKNS20}, Lemma 58]\label{lemma:certify}
There is a randomized procedure that, given a set
$K \subseteq S^{d-1}$ and parameters $\Delta,\tau,\delta \in \left(0,\tfrac13\right)$, runs in time
\[
O\!\left(\frac{d}{\Delta^{2}\tau}\log\!\left(\frac{2|K|}{\delta}\right)\cdot |K|\right)
\]
and, with probability at least $1-\delta$, does one of the following:
\begin{enumerate}
    \item returns a point $k^* \in S^{d-1}$ such that
    \[
    \left|B\!\left(k^*, \sqrt{2(1 - \Delta^2)}\right)\cap K\right|
    \;\ge\;
    \Omega(\Delta^2)\cdot \tau |K|
    \]
    \item or certifies that the set $K$ is $(\Delta',\tau)$-pseudo-random, where $\Delta' = \Theta(\Delta)$.
\end{enumerate}
\end{lemma}

\begin{proof}[Proof of \Cref{cor:certify-ball}]
By translation and scaling, it suffices to consider the case $c=0$ and $r=1$.

Let $\pi(K)$ denote the multiset obtained by radially projecting each point of
$K$ onto $S^{d-1}$. Restricted to the sphere, the key sets in \Cref{eq:sp_caps} are exactly those that fall in the spherical caps of radius at most $\sqrt{2(1-\Delta)}$. Therefore, if $K$ is
not $(\Delta,\tau)$-pseudo-random, then $\pi(K)$ is not $(\Delta,\tau)$-pseudo-random
either. Consequently, if \Cref{lemma:certify} certifies that $\pi(K)$ is
$(\Delta',\tau)$-pseudo-random, then $K$ is also $(\Delta',\tau)$-pseudo-random. Otherwise, \Cref{lemma:certify} returns a point $u^* \in S^{d-1}$ such that
\[
\left|B\!\left(u^*,\sqrt{2(1-\Delta^2)}\right)\cap \pi(K)\right|
\ge \Omega(\Delta^2)\cdot\tau |K|.
\]
Let
\[
C := B\!\left(u^*,\sqrt{2(1-\Delta^2)}\right)\cap S^{d-1}
\]
and define its radial cone by
\[
\cone(C) := \{\alpha x : \alpha \in [0,1],\ x \in C\}.
\]
Then
\[
|\cone(C)\cap K| \ge \Omega(\Delta^2)\cdot \tau |K|.
\]
We now output
\[
k^* := \Delta^2 u^* \in B(0,1).
\]
It remains to show that
\[
\cone(C)\subseteq B\!\left(k^*,\sqrt{1-\Delta^4}\right).
\]
Indeed, let $x=\alpha y \in \cone(C)$, where $\alpha\in[0,1]$ and $y\in C$.
Since $\|y-u^*\|_2\le \sqrt{2(1-\Delta^2)}$, we have $\langle y,u^*\rangle \ge \Delta^2$.
Therefore,
\[
\|x-k^*\|_2^2
=
\|\alpha y-\Delta^2u^*\|_2^2
=
\alpha^2+\Delta^4-2\alpha\Delta^2\langle y,u^*\rangle
\le
\alpha^2+\Delta^4-2\alpha\Delta^4
\le
1-\Delta^4,
\]
where the last inequality uses $\alpha\in[0,1]$. 
\end{proof}

Next, we provide the analysis of the performance guarantees of \Partition{}, restated below:

\partition*

\begin{proof}[Proof of \Cref{lem:ps_randomification}]  First of all, we assume without loss of generality that $ \qty(\frac{\log |K|}{\Delta^2\cdot \tau})^{O\qty(\frac{\log \frac{r}{r_0}}{\Delta^4})} \leq n$. Otherwise, we may simply partition the dataset into one point subsets.

We construct the decomposition $\{K_i\}_{i \in [N]}$ by simply recursively applying the procedure \Certify{} (\ref{cor:certify-ball}) which either identifies a dense spherical cap or correctly reports that none is present. We formalize this recursion in \Cref{alg:partition}.

\paragraph{Runtime of \Cref{alg:partition}.}\label{par:runtime}
At each level of recursion, the runtime of the algorithm is dominated by the
calls to \Certify{}. We therefore bound only the total cost of these calls.
Fix one invocation of \Partition{} on a dataset of size $|K|$. Each successful
call to \Certify{} removes an $\Omega(\Delta^2\tau)$ fraction of the current
dataset, which means the total work spent on this invocation is
\[
\sum_{t = 0}^{\infty}  O\!\left(\frac{d}{\Delta^2\tau}(1 - \Delta^2\tau)\cdot|K|\log n\right)= O\!\left(\frac{d}{\Delta^4\tau^2}\cdot |K|\log n\right).
\]

Now fix a recursion depth. The datasets processed at that depth are pairwise
disjoint subsets of the original input $K$, so the sum of their sizes is at
most $|K|$. Therefore the total runtime of all calls to \Certify{} at any fixed depth is
\[
O\!\left(\frac{d}{\Delta^4\tau^2}\cdot |K|\log n\right).
\]
Finally, the recursion depth is
\[
O\!\left(\frac{\log(r/r_0)}{\Delta^4}\right),
\]
because each recursive call shrinks the radius by a factor $\sqrt{1-\Delta^4}$.
This yields the claimed runtime bound.

\paragraph{Size of $\mathcal{K}$.} Note that we only output a set of $\mathcal{K}$ at a leaf of the recursion tree, so it suffices to bound the number of leaves. Because every time line~\eqref{line:while_certify} is triggered the size of $K$ decreases by a factor of $(1 - \Omega(\Delta^2\tau))$, \Partition{} makes $O\left(\frac{\log |K|}{\Delta^2\tau}\right)$ recursive calls to itself. The depth of recursion is bounded by $O\left(\frac{\log (r/r_0)}{\Delta^4}\right)$ because each recursive call targets a dataset enclosed in a ball of radius $\sqrt{1 - \Delta^4}r$. Therefore, 
\[|\mathcal{K}| \leq \qty(\frac{\log |K|}{\Delta^2\cdot \tau})^{O\qty(\frac{\log \frac{r}{r_0}}{\Delta^4})}\]

\paragraph{Probability of success of \Cref{alg:partition}.} The algorithm is successful if all runs of \Certify{} are successful. Note that the number of calls to \Certify{} at a node is at most one more than the number of its children. Therefore, we may bound the number of calls to \Certify{} as product of the number of leaves and the number of levels of recursion which is 
\[O\left(\qty(\frac{\log |K|}{\Delta^2\cdot \tau})^{O\qty(\frac{\log \frac{r}{r_0}}{\Delta^4})}\cdot \frac{\log(r/r_0)}{\Delta^4}\right).\]
By our convention, $\qty(\frac{\log |K|}{\Delta^2\cdot \tau})^{O\qty(\frac{\log \frac{r}{r_0}}{\Delta^4})} \leq n$. Therefore, we crudely bound the above expression by $n^2$. By the union bound argument, all of the calls to \Certify{} are simultaneously successful with probability at least $1 - \frac{1}{n^2}$.

\end{proof}

\bibliographystyle{alpha}
\bibliography{main.bib}

\end{document}